\documentclass[12pt]{article}
\usepackage[utf8]{inputenc}
\usepackage[top=1in, bottom=1in, left=0.75in, right=0.75in]{geometry}
\usepackage{setspace}
\usepackage{natbib}
\usepackage{amsmath}
\usepackage{amssymb}
\usepackage[dvipsnames]{xcolor}
\usepackage{graphicx}
\usepackage{amsthm}
\usepackage{mathtools}
\usepackage{algpseudocode}
\usepackage{algorithm}
\usepackage{appendix}
\usepackage{authblk}
\usepackage{subfigure}
\newtheorem{theorem}{Theorem}[section]

\newtheorem{lemma}[theorem]{Lemma}
\newtheorem{proposition}[theorem]{Proposition}

\newtheorem{assumption}{Assumption}
\definecolor{change}{rgb}{0,0,0}

\newcommand{\vectorize}[1]{\textnormal{vec}(#1)}
\newcommand{\argmax}{\arg\max}
\newcommand{\argmin}{\arg\min}
\newcommand{\zero}{\boldsymbol{0}}

\title{Estimation of High-Dimensional Markov-Switching VAR Models with an Approximate EM Algorithm}
\author[1]{Xiudi Li}
\author[2]{Abolfazl Safikhani}
\author[3,4]{Ali Shojaie}
\affil[1]{Division of Biostatistics, University of California, Berkeley}
\affil[2]{Department of Statistics, George Mason University}
\affil[3]{Department of Biostatistics, University of Washington}
\affil[4]{Department of Statistics, University of Washington}
\date{}

\begin{document}
\maketitle

\begin{abstract}
    Regime shifts in high-dimensional time series arise naturally in many applications, from neuroimaging to finance. This problem has received considerable attention in low-dimensional settings, with both Bayesian and frequentist methods used extensively for parameter estimation. The EM algorithm is a particularly popular strategy for parameter estimation in low-dimensional settings, although the statistical properties of the resulting estimates have not been well understood. Furthermore, its extension to high-dimensional time series has proved challenging. To overcome these challenges, in this paper we propose an approximate EM algorithm for Markov-switching VAR models that leads to efficient computation and also facilitates the investigation of asymptotic properties of the resulting parameter estimates. We establish the consistency of the proposed EM algorithm in high dimensions and investigate its performance via simulation studies. We also demonstrate the  algorithm by analyzing a brain electroencephalography (EEG) dataset recorded on a patient experiencing epileptic seizure.
\end{abstract}

\noindent \textbf{Keywords:} High-Dimensional Time Series; Penalized Estimation; $\beta$-Mixing.

\section{Introduction}
The presence of regime shifts is an important feature for many time series arising from various applications; that is, the time series may exhibit changing behavior over time. While some of these regime shifts are attributable to certain deterministic structural changes, in many cases the dynamics of the observed time series is governed by an exogenous stochastic process that determines the regime. For example, the behavior of key macroeconomic indicators may depend on the phase of the business cycle, such as recession versus expansion  \citep{hamilton1989new,artis2004european}; the relationship between stock market return and exchange rate may vary between high- and low-volatility regimes \citep{chkili2014exchange}; in neuroimaging studies, the connectivity between different brain regions may change over time according to the brain's underlying states  \citep{fiecas2021approximate}. As in these examples, the exogenous process that determines the regime is oftentimes latent or not directly observable. This relates naturally to the notion of state-space models \citep{koller2009probabilistic}. 

A prominent example of such a state-space model is the hidden Markov model (HMM) \citep[see, for example, ][]{rabiner1989tutorial}. In an HMM, the observations over time are conditionally independent given a latent process that determines the state, and the state process is assumed to be a Markov chain. Markov-switching vector autoregression (VAR) \citep{krolzig2013markov} can be regarded as a generalization of the HMM model that allows for autoregression. In a Markov-switching VAR model, the dynamics of the observed time series takes a VAR form, but the autoregressive parameters depend on the state of a latent finite-state Markov chain. As such, Markov-switching VAR models have found widespread applications, in areas such as macroeconomics \citep[][and the references therein]{krolzig2013markov} and ecology \citep{solari2011use}. Estimation of the autoregressive parameters and the transition of the latent regime process is often of central interest in these applications.

Bayesian approaches have been developed \citep{fox2010bayesian} and applied for parameter estimation in Markov-switching VARs using various prior distributions for the model parameters \citep{billio2016interconnections,droumaguet2017granger}. Alternatively, the expectation-maximization (EM) algorithm  \citep{dempster1977maximum} provides a general approach for computing the maximum likelihood estimator in latent variable and incomplete data problems. Application of the EM algorithm in Markov-switching VAR models dates back to \citet{lindgren1978markov} and \citet{hamilton1989new}, and the method has gained great popularity since then. However, to the best of our knowledge, the theoretical properties of the estimates obtained using the EM algorithm in Markov-switching VAR models have not been rigorously investigated, even in low-dimensional settings.

Several existing studies have investigated the theoretical properties of the EM algorithm in independent data settings. For instance, the seminal work by \citet{wu1983convergence} established the convergence of the EM algorithm to the unique global optimum under unimodality of the likelihood function \citep{wu1983convergence}; however, in general cases, where the likelihood function is multi-modal, only convergence to some local optimum has been established \citep[for example, ][]{mclachlan2007algorithm}. Recent work by \citet{balakrishnan2017statistical} establishes statistical guarantees for the EM estimate in low-dimensional settings, when the algorithm is initialized within a local region around the true parameter. Under certain conditions, the authors show geometric convergence to an EM fixed point that is within statistical precision of the true parameter. \citet{wang2014high} extends the EM algorithm to high-dimensional settings (where the number of parameters is larger than sample size), by introducing truncation in the M-step. In contrast, \citet{yi2015regularized} develops a regularized EM algorithm for high-dimensional problems that incorporates regularization in the M-step. It also establishes general statistical guarantee for the resulting parameter estimate. Regularized EM algorithms are also studied in specific contexts, including high-dimensional mixture regression \citep{stadler2010l1} and graphical models \citep{hao2017simultaneous}. All of these existing studies consider a sample of independent and identically distributed observations, and their results do not generalize to dependent observations. 

Regularized EM algorithms for Markov-switching VAR models have been considered by \citet{monbet2017sparse} and \citet{maung2021estimating}, to allow the number of parameters to diverge; however, these works did not analyze the statistical properties of the estimate, which is challenging in time series settings. A major source of complication in this setting arises from the dependence of the conditional expectation in the E-step on observations in all time points, due to the unobserved latent variables.
In this work, we develop a regularized EM algorithm for parameter estimation in high-dimensional Markov-switching VAR models, and rigorously establish its performance guarantee. To deal with the temporal dependence among observations, we introduce an approximate E-step, in which we compute approximate conditional expectations. In addition to facilitating theoretical analysis, this approximation also leads to improved computation. 

To establish the consistency of the proposed algorithm, we apply novel probabilistic tools for ergodic time series. We derive an upper bound on the estimation error of the EM estimate when the initialization lies within a local neighborhood of the true parameter value. Notably, our analysis does not require the radius of this neighborhood to shrink with sample size. Establishing this bound necessitates controlling statistical errors throughout the EM iterations, which, in turn, requires uniform concentration results over function classes indexed by the parameter values. To achieve this, we develop a general approach for proving such concentration results for dependent data under 
$\beta$-mixing conditions. Our technique is more broadly applicable and can be of interest beyond the VAR setting studied in this paper.

The rest of this paper is organized as follows. In Section~\ref{sec:EM}, we introduce Markov-switching VAR models and propose a modified EM algorithm for parameter estimation. In Section~\ref{sec:theoretical}, we derive performance guarantees for the proposed algorithm by establishing an upper bound on the estimation error of the resulting estimate under mild conditions. {\color{change}The proofs of our main theoretical results are outlined in  Supplementary Appendix~\ref{sec:proof outline} and the detailed proofs are given in Supplementary Appendices~\ref{app:prooflemma} to \ref{app:proofmainthm}.} In Section~\ref{sec:numerical}, we demonstrate the performance of the proposed method through simulation studies {\color{change}and apply it to analyze a brain electroencephalography (EEG) dataset recorded on a patient experiencing epileptic seizure}.
Section~\ref{sec:discussion} concludes the paper with a discussion.

\paragraph{Notation.} Throughout the paper, we denote the $L_p$-norm of a generic vector $u = (u_1, u_2, \ldots, u_d) \in \mathbb{R}^d$ by $\|u\|_p := (\sum_{i=1}^d |u_i|^p)^{1/p}$ and the spectral norm of a generic matrix $M$ by $\|M\|_2$. We denote the transpose of a generic matrix $M$ by $M^\top$. For a symmetric matrix $M$, we use $\lambda_{\min}(M)$ and $\lambda_{\max}(M)$ to denote its minimum and maximum eigenvalues, respectively. For generic stochastic process $\{X_t\}$, we write the random vector $(X_{t_1},X_{t_1+1},\ldots, X_{t_2})$ as $X_{t_1}^{t_2}$, for $t_1 \leq t_2$. The identity matrix of dimension $d$ is denoted as $\textnormal{Id}_d$. We use $I\{\cdot\}$ to denote the indicator function and use $\otimes$ to denote the Kronecker product of matrices. For positive sequences $\{a_n\}$ and $\{b_n\}$, we write $a_n = O(b_n)$ if $\limsup a_n/b_n \leq C$ for some constant $C$, and $a_n \asymp b_n$ if $a_n=O(b_n)$ and $b_n = O(a_n)$; we write $a_n = o(b_n)$ if $\limsup a_n/b_n = 0$.

\section{EM Algorithm for Markov-Switching VAR Models}\label{sec:EM}

\subsection{Markov-switching VAR model}
{\color{change}Let $Y_t \in \mathbb{R}^d$ denote the observed vector of variables at time $t$ and $Z_t \in \{1,\ldots,K\}$ denote the latent variable that determines the regime. For $1\leq i \leq K$, let $B_i \in \mathbb{R}^{d \times d}$ denote the VAR transition matrix 
in the $i$-th regime, corresponding to the latent variable $Z_t$ taking value $i$, and let $\sigma_i^2 \in \mathbb{R}^+$ denote the variance of the noise vector in the $i$-th regime. We assume that $\{Y_t\}$ follows a regime-switching vector-autoregressive (VAR) model,
\begin{equation}\label{msvar}
    Y_t = \sum_{i=1}^{K} I\{Z_t = i\} \left(B_i^\top Y_{t-1} + \sigma_i \epsilon_t\right),
\end{equation}
where the noise vectors $\epsilon_t$ follow $N(\zero,\textnormal{Id}_d)$ and are independent and identically distributed over time.} The unobserved process $\{Z_t\}$ follows a time-homogeneous first-order Markov chain with transition matrix $P_Z \in \mathbb{R}^{K\times K}$. Let
\begin{equation}\label{latentmc}
    p_{ij} = (P_Z)_{ij} = P\left(Z_t = j | Z_{t-1} = i \right), \quad (i,j) \in \{1,\ldots, K\} \times \{1,\ldots, K\}.
\end{equation}
Note that, for the process defined by \eqref{msvar} and \eqref{latentmc}, $\{(Y_t,Z_t)\}$ is also a first-order Markov chain. 

We focus primarily on VAR of lag 1 in this paper; however, the proposed method generalizes easily to Markov-switching VAR models with general lag in the vector autoregression part; this is because any VAR$(l)$ model can be written as a VAR(1) model with ${Y}_t^\dagger = (Y_t,Y_{t-1},\ldots, Y_{t-l+1})$. Moreover, for VAR with lag $l>1$, the method can be extended to the case where the $l$ regression coefficient matrices are determined jointly by $Z_{t-l_1}^t$ for some $l_1>0$. In such cases, a new regime indicator $Z_t^\ddagger$ can be defined as $(Z_{t-l_1},\ldots,Z_t)$, which again forms a Markov-Chain.

Let $\beta_i = \vectorize{B_i}$ denote the vectorized regression coefficients by concatenating the columns of $B_i$, for $i \in \{1,\ldots,K\}$. Letting $p = \vectorize{P_Z}$, $\sigma = (\sigma_1,\ldots, \sigma_K)$ and $\beta = (\beta_1^\top, \ldots, \beta_K^\top)^\top$, we denote the parameter vector as {\color{change}$\theta = (\beta^\top, p^\top, \sigma^\top)^\top$}. We will estimate $\theta$ via a (modified) EM algorithm, noting that 
identifiability of $\theta$ has been established \citep[][Chapter 6]{krolzig2013markov}. 
Hereafter, we use superscript $*$ to denote the true parameter value to be estimated. In addition, we let $S = \textnormal{supp}(\beta^*)$ denote the support of $\beta^*$, and $|S|$ denote its cardinality.

\subsection{An approximate EM algorithm}
Let $\{Y_t\}_{t=0}^T$ be the observed outcome vectors over time. Had the latent process $\{Z_t\}$ been observed, we would estimate $\theta$ by maximizing the full sample log-likelihood function, defined as 
{\color{change}\begin{multline}
    l(\theta; Y_1^T,Z_1^T) = \frac{1}{T}\sum_{t=1}^T\Bigg[ \left(\sum_{i=1}^K\sum_{j=1}^K I\{Z_{t-1} = i, Z_t = j \} \log p_{ij}\right) \\
    -\sum_{j=1}^K I\{Z_t= j\}\Big(\frac{d}{2}\log \sigma_j^2 + \frac{d}{2}\log 2\pi
    + \frac{1}{2\sigma_j^2}\left\|Y_t - (\textnormal{Id}_d \otimes Y_{t-1})^\top \beta_j \right\|_2^2 \Big) \Bigg]. \label{eqn:approxEX}
\end{multline}}
However, as $\{Z_t\}$ is not observed, maximizing the full sample log-likelihood is infeasible. Instead, we can consider the EM algorithm, where we iterate between the E-step and the M-step. Consider the $q$-th iteration of the EM algorithm. In the E-step, given the current parameter estimate $\theta^{(q-1)}$, the unobserved variables are replaced with their conditional expectations computed under $\theta^{(q-1)}$, conditioned on the observed variables. 
In particular, given a generic parameter value $\theta$, define the smoothed probabilities
\begin{equation}\label{exactwij}
    w_{ij,\theta}(Y_0^T) = P_\theta(Z_{t-1}=i,Z_t=j | Y_0,Y_1,\ldots,Y_T),
\end{equation}
and
\begin{equation}\label{exactwj}
    w_{j,\theta}(Y_0^T) = P_\theta(Z_t=j | Y_0,Y_1,\ldots,Y_T),
\end{equation}
respectively. Here the subscript $\theta$ indicates computing expectation or probability under the parameter value $\theta$. In the E-step of the $q$-th iteration, we could consider replacing $I\{Z_{t-1} = i, Z_t =j\}$ and $I\{Z_t = j\}$ in \eqref{eqn:approxEX} with the smoothed probabilities $w_{ij,\theta^{(q-1)}}(Y_0^T)$ and $w_{j,\theta^{(q-1)}}(Y_0^T)$, respectively, to form the objective function. Subsequently, in the M-step, we maximize the resulting objective function. 

The complex dependence structure in the process and the high-dimensionality of the problem pose significant challenges both theoretically and computationally if we directly apply the EM algorithm outlined in the previous paragraph. We thus propose a modified EM algorithm to overcome these challenges.

First, in the E-step, the exact conditional expectations defined in \eqref{exactwij} and \eqref{exactwj} depend on all the outcome $Y$'s up to time $T$. Theoretically, given such dependence, it might be difficult to establish certain concentration results (see Section~\ref{sec:theoretical}) to obtain the performance guarantee of the EM algorithm. At the same time, certain recursive algorithms are needed for efficient computation of these conditional probabilities without enumerating over all possible paths of $Z_1^T$ \citep[see, for example, ][]{baum1970maximization,lindgren1978markov,hamilton1989new,kim1994dynamic}. Given these difficulties, we propose the following modification to the E-step, in which we approximate the conditional expectations. Specifically, for a generic $\theta$, we define
\begin{align}
    m_{ij,\theta}(Y_{t-s}^{t+s}) &= P_\theta(Z_{t-1}=i,Z_t=j | Z_{t-s}=1, Y_0,Y_1,\ldots,Y_{t+s}) \nonumber \\
    &=  P_\theta(Z_{t-1}=i,Z_t=j | Z_{t-s}=1, Y_{t-s},\ldots,Y_{t+s}),
\end{align}
and
\begin{align}
    m_{j,\theta}(Y_{t-s}^{t+s}) &= P_\theta(Z_t=j | Z_{t-s}=1, Y_0,Y_1,\ldots,Y_{t+s}) \nonumber \\
    &=  P_\theta(Z_t=j | Z_{t-s}=1, Y_{t-s},\ldots,Y_{t+s}),
\end{align}
respectively, for a specified value of $s$. Then, in the $q$-th iteration, we replace $I\{Z_{t-1}=i,Z_t = j\}$ and $I\{Z_t = j\}$ in \eqref{eqn:approxEX} with $m_{ij,\theta^{(q-1)}}(Y_{t-s}^{t+s})$ and $m_{j,\theta^{(q-1)}}(Y_{t-s}^{t+s})$, respectively. These quantities depend only on $Y_{t-s}^{t+s}$, due to the Markovian property of $\{(Y_t,Z_t)\}$. As will be shown in Lemma~\ref{lemma:approximationerror}, under suitable conditions, the error resulting from such approximations will be small with appropriately chosen value of $s$. Under these conditions, the choice of $Z_{t-s} = 1$ is also arbitrary, and conditioning on $Z_{t-s} = i$ for any $i$ yields an equally accurate approximation. In Section~\ref{sec:theoretical}, we will show that we can choose $s \asymp \log(T)$. Thus, evaluating the approximate conditional expectations can be done via enumerating over only $K^{2s}$ paths of length $2s$, which is computationally efficient. 

In high-dimensional problems where $d$ is comparable to, or larger than, $T$, certain modifications of the M-step are also necessary, as otherwise the maximization is ill-posed. Similar to \citet{yi2015regularized} and \citet{hao2017simultaneous}, we employ a regularized optimization approach, where we add an $L_1$ penalty on the regression coefficients $\beta$. Hence, given the current parameter estimate $\theta^{(q-1)}$, we maximize the following objective function in terms of $\tilde\theta$:
\begin{align}\label{sampleQ}
    Q_{n,\lambda}(\tilde\theta|\theta^{(q-1)}) &= \frac{1}{T}\sum_{t=1}^{T}
    \Bigg[ \left(\sum_{i=1}^K\sum_{j=1}^K m_{ij,\theta^{(q-1)}}(Y_{t-s}^{t+s}) \log \tilde p_{ij}\right) \nonumber \\
    &\quad - \frac{d}{2} \sum_{j=1}^K m_{j,\theta^{(q-1)}}(Y_{t-s}^{t+s}) \left(\log 2\pi + \log \tilde\sigma_j^2 \right) \nonumber \\
    &\quad - \left( \sum_{j=1}^K m_{j,\theta^{(q-1)}}(Y_{t-s}^{t+s}) \frac{1}{2\tilde\sigma_j^2} \left\|Y_t - (\textnormal{Id}_d \otimes Y_{t-1})^\top \tilde\beta_j \right\|_2^2\right)\Bigg] - \lambda \sum_{j=1}^K \|\tilde\beta_j\|_1.
\end{align}
{\color{change}We observe that the updates for $p$ and $(\beta,\sigma)$ can be performed separately. Moreover, to find the optimizer over $(\beta,\sigma)$ we can optimize over $(\beta_j, \sigma_j)$ in each regime separately. Here, we use an overall penalty parameter $\lambda$. Alternatively, we may apply a different penalty parameter for each regime in \eqref{sampleQ}, which may increase the flexibility in handling different level of sparsity in the $K$ regimes. The update of $p$ has a closed-form solution,
\begin{equation}\label{eq: update of p}
    p^{(q)}_{ij} = \sum_{t=1}^T m_{ij,\theta^{(q-1)}}(Y_{t-s}^{t+s}) / \sum_{t=1}^T \sum_{j=1}^K m_{ij,\theta^{(q-1)}}(Y_{t-s}^{t+s}),
\end{equation}
where $p_{ij}$ denotes the $(i,j)$-th element of the transition matrix $P_Z$. Updating $(\beta_j,\sigma_j)$ is equivalent to solving the following optimization problem 
\begin{multline}\label{eq: update of beta}
    \left(\beta_j^{(q)},\sigma_j^{(q)}\right) = \argmin_{\tilde\beta_j,\tilde\sigma_j} 
    \frac{1}{T}\sum_{t=1}^{T}
    m_{j,\theta^{(q-1)}}(Y_{t-s}^{t+s})\left\{\frac{1}{2\tilde\sigma_j^2}\left\|Y_t - (\textnormal{Id}_d \otimes Y_{t-1})^\top \tilde\beta_j \right\|_2^2 + \frac{d}{2} \log \tilde\sigma_j^2 \right\} \\
    + \lambda \|\tilde\beta_j\|_1,
\end{multline}
which is a standard penalized regression problem wherein each observation is associated with a weight $m_{j,\theta^{(q-1)}}(Y_{t-s}^{t+s})$. In practice, we first solve for $\beta_j^{(q)}$ as
\begin{align}
    \beta_j^{(q)} &= \argmin_{\tilde\beta_j} 
    \frac{1}{T}\sum_{t=1}^{T}
    m_{j,\theta^{(q-1)}}(Y_{t-s}^{t+s})\left\|Y_t - (\textnormal{Id}_d \otimes Y_{t-1})^\top \tilde\beta_j \right\|_2^2+ \lambda \|\tilde\beta_j\|_1,
\end{align}
using standard software, such as \texttt{glmnet} in \texttt{R}, and then set $(\sigma_j^{(q)})^2$ as the weighted average of the residual sum of squares. The penalty parameter $\lambda$ will change over iterations, and we use $\lambda^{(q)}$ to denote its value in the $q$-th iteration. Let $\theta^{(q)} = ( (\beta^{(q)})^\top, (p^{(q)})^\top, (\sigma^{(q)})^\top )^\top$ denote the updated parameter. The proposed approximate EM algorithm is summarized in Algorithm~\ref{modifiedEM}. }

\begin{algorithm}[t]
\caption{Approximate EM algorithm for high-dimensional Markov-switching VAR}\label{modifiedEM}
\hspace*{\algorithmicindent} \textbf{Input:} Observations $\{Y_0,Y_1,\ldots,Y_T\}$, number of regimes $K$;\\
\hspace*{\algorithmicindent} 
\textbf{Output:} Parameter estimate $\hat\theta$ 
\begin{algorithmic}[]
\State Initialize the parameter $\theta^{(0)} = ( (\beta^{(0)})^\top, (p^{(0)})^\top, (\sigma^{(0)})^\top )^\top$
\State $q \leftarrow 1$
\While{Convergence condition not met}
\State update $p$: obtain $p^{(q)}$ as in \eqref{eq: update of p}
\State choose tuning parameter $\lambda^{(q)}$
\For{$j=1,2,\ldots, K$}
\State update $(\beta_j, \sigma_j)$: obtain $(\beta_j^{(q)},\sigma_j^{(q)})$ as in \eqref{eq: update of beta}
\EndFor
\State update $\theta$: $\theta^{(q)} \leftarrow ( (\beta^{(q)})^\top, (p^{(q)})^\top, (\sigma^{(q)})^\top )^\top$
\State $q \leftarrow q+1$
\EndWhile
\State $\hat\theta \leftarrow \theta^{(q-1)}$
\end{algorithmic}
\end{algorithm}

\section{Theoretical Guarantees}\label{sec:theoretical}

\subsection{Stationarity of Markov-switching VAR models}
Before studying the consistency of the EM algorithm, we first investigate the stationarity of Markov-switching VAR models. We introduce the following assumption: 
\begin{assumption}[Norm of coefficient matrix]\label{cond: operatornorm}
There exists some constant $\tilde c<1$, such that $\|B_i^*\|_2 \leq \tilde c$ for all $i \in \{1,\ldots,K\}$.
\end{assumption}

\begin{lemma}[Stationarity and geometric ergodicity]\label{lemma: msvarstationarity}
Under Assumption \ref{cond: operatornorm}, the process $\{(Y_t,Z_t)\}$ is strictly stationary and geometrically ergodic. Moreover, under sampling from the stationary distribution, $Y_t$ is a sub-Gaussian random vector. 
\end{lemma}

A VAR model will be stable if the spectral radius of the matrix of regression coefficients is upper bounded away from one \citep{lutkepohl2013introduction}. Here, our Assumption~\ref{cond: operatornorm} requires that the spectral norm of the regression coefficient matrix $B_i$ is upper bounded away from one, for all of the regimes. This is stronger than assuming that the spectral radii of all the $B_i$'s are bounded away from one, as spectral norm is generally larger than spectral radius. However, as shown in \citet{stelzer2009markov}, the weaker assumption on spectral radius is generally not sufficient to guarantee stationarity of the regime-switching VAR models, or to guarantee the existence of all moments of $Y_t$, which is essential for applying concentration results and deriving upper bounds on the estimation error.

Assumption~\ref{cond: operatornorm} is also sufficient for geometric ergodicity of the process $\{(Y_t,Z_t)\}$. Proposition 2 in \citet{liebscher2005towards} then implies that the process is geometrically $\beta$-mixing; that is, $b_{\textnormal{mix}}(l) = O(c^l)$ for some $c \in (0,1)$, where $b_{\textnormal{mix}}(l)$ is the $\beta$-mixing coefficient defined in, for example, \citet{bradley2005basic}. 

\subsection{Statistical analysis of the EM algorithm}
Similar to previous works on the theoretical properties of EM algorithm for independent data \citep{balakrishnan2017statistical,yi2015regularized,hao2017simultaneous}, the analysis of our proposed algorithm involves two major steps. In the first step, we show that if we had access to infinite amount of data, the output of the EM algorithm would converge geometrically to the true parameter value, given proper initialization. In the second step, we focus on one iteration of the proposed algorithm, and show that the updated estimate obtained from our modified EM algorithm using finite sample is close to the updated estimate we would get with infinite amount of data. Similar to the previous works mentioned earlier, the estimation error of the proposed algorithm includes a `statistical error' term and an `optimization error' term. However, as shown in Theorem~\ref{estimationerror}, we have an additional source of error, which we term `approximation error'. This is due to using an approximation of the true conditional expectations of the unobserved variables. Theoretically, establishing the desired error bound requires novel concentration results for dependent data and particularly those for ergodic stochastic processes, which distinguishes our work in time series settings from the previous works on the EM algorithm for independent data.

Before presenting our main result, we state some conditions that will be used to establish the upper bound on the estimation error. The first condition will be useful in showing that the EM algorithm with infinite amount of data converges to the true parameter, given proper initialization. To this end, we introduce a population objective function analogous to \eqref{sampleQ}:
\begin{multline}\label{populationQ}
    Q(\tilde\theta|\theta) = \frac{1}{T}\sum_{t=1}^{T}E_0\Bigg[\left(\sum_{i=1}^K\sum_{j=1}^K w_{ij,\theta}(Y_0^T) \log \tilde p_{ij}\right) - \frac{d}{2}\sum_{j=1}^K w_{j,\theta}(Y_0^T)\left(\log 2\pi + \log \tilde\sigma_j^2 \right)\\
    -\left(\sum_{j=1}^K w_{j,\theta}(Y_0^T)\frac{1}{2\tilde\sigma_j^2}\left\|Y_t - (\textnormal{Id}_d \otimes Y_{t-1})^\top \tilde\beta_j \right\|_2^2\right)\Bigg].
\end{multline}
If we had access to infinite number of realizations of the process, given the current parameter estimate $\theta$, we would maximize $Q(\cdot|\theta)$ in the M-step of the EM algorithm. Note that if the weighted covariance matrices are invertible (see Assumption~\ref{cond:minmaxeigenvalue}), there is a unique global optimizer for \eqref{populationQ} and $L_1$ regularization on $\beta$ is not necessary. Let $M(\theta) = \argmax_{\tilde\theta}Q(\tilde\theta|\theta)$. 
We introduce a measure of (inverse) signal strength on the population level. First, we partition $M(\theta)$ into the estimates of regression coefficients, transition probability estimates and variance estimates, denoted by $M_\beta(\theta)$, $M_p(\theta)$ and $M_\sigma(\theta)$, respectively. That is, $M(\theta) = (M_\beta(\theta)^\top, M_p(\theta)^\top,M_\sigma(\theta)^\top)^\top$. In particular,
\begin{align*}
    (M_\beta(\theta))_j &= \left\{\textnormal{Id}_d \otimes E_0\left[\frac{1}{T}\sum_{t=1}^Tw_{j,\theta}(Y_0^T)Y_{t-1}Y_{t-1}^\top\right]\right\}^{-1}E_0\left[\frac{1}{T}\sum_{t=1}^Tw_{j,\theta}(Y_0^T)\left(Y_t\otimes Y_{t-1}\right)\right]; \\
    (M_\sigma(\theta))_j^2 &= E_0\left[\frac{1}{T}\sum_{t=1}^T\frac{1}{d} w_{j,\theta}(Y_{0}^T)\left\|Y_t - \left(\textnormal{Id}_d \otimes Y_{t-1}\right)^\top (M_\beta(\theta))_j\right\|_2^2\right] \left\{E_0\left[\frac{1}{T}\sum_{t=1}^T w_{j,\theta}(Y_{0}^T)\right]\right\}^{-1}; \\
    (M_p(\theta))_{ij}&= E_0\left[\frac{1}{T}\sum_{t=1}^T w_{ij,\theta}(Y_0^T)\right]\left\{E_0\left[\frac{1}{T}\sum_{t=1}^T \sum_{j=1}^K w_{ij,\theta}(Y_0^T)\right]\right\}^{-1},
\end{align*}
where $(M_\beta(\theta))_j$ is the subvector of $M_\beta(\theta)$ corresponding to $\beta_j$, the vectorized regression coefficients in regime $j$. Let $\mathcal{B}(\tilde r;\theta^*)$ denote an $l_2$-ball of radius $\tilde r$ centered at $\theta^*$, the true parameter value. We introduce the following assumption on the mapping $\theta \mapsto M(\theta)$.

\begin{assumption}[Signal strength]\label{signalstrength}
Suppose that there exist constants $\kappa < 1$ and $r>0$ such that for all $\theta^\dagger \in \mathcal{B}(r;\theta^*)$,
\begin{equation}
    \left\|\frac{\partial M(\theta)}{\partial\theta}\Bigr\rvert_{\theta = \theta^\dagger}\right\|_2 \leq \kappa.
\end{equation}
\end{assumption}
\noindent Here, the operator norm of the gradient matrix $\partial M(\theta)/\partial\theta$ serves as our (inverse) signal strength measure. The constant $r$ characterizes the region of `proper initialization', which is a ball centered around $\theta^*$. When $\partial M(\theta)/\partial \theta$ is continuous in $\theta$ and the norm of the gradient matrix evaluated at $\theta^*$ is below 1, we would expect that there exists a neighborhood around $\theta^*$ in which the norm of the gradient matrix evaluated at any $\theta$ is bounded below 1, due to continuity. In Supplementary Appendix~\ref{app:signalempirical}, we empirically examine the norm of the gradient matrix at $\theta^*$ in some examples.

\begin{lemma}[Contraction of $M(\theta)$]\label{populationcontraction}
Under Assumption~\ref{signalstrength}, $\|M(\theta) - \theta^*\|_2 \leq \kappa \|\theta-\theta^*\|_2$, for $\theta \in \mathcal{B}(r;\theta^*)$.
\end{lemma}

The next condition ensures that the denominators of $M_\sigma(\theta)$ and $M_p(\theta)$ are bounded away from 0, so that these updates are well-defined in the population EM iteration.
\begin{assumption}[Transition probability estimate in the population EM]\label{denomtransition}
There exists a constant $\iota > 0$ such that for all $\theta \in \mathcal{B}(r;\theta^*)$ and all $i \in \{1,\ldots,K\}$, 
\begin{equation*}
    E_0\left[\frac{1}{T}\sum_{t=1}^T w_{i,\theta}(Y_0^T)\right] \geq \iota, \ E_0\left[\frac{1}{T}\sum_{t=1}^T\sum_{j=1}^K w_{ij,\theta}(Y_0^T)\right] \geq \iota.
\end{equation*}
\end{assumption}
\noindent Under this assumption, the denominators of $M_\sigma(\theta)$ and $M_p(\theta)$ are bounded away from 0 uniformly in $\theta$, which ensures that the updates for the noise variances and the transition probabilities are well-defined throughout the population EM iterations. In particular, at $\theta^*$, we have that
\begin{align*}
    E_0\left[\frac{1}{T}\sum_{t=1}^T\sum_{j=1}^K w_{ij,\theta^*}(Y_0^T) \right] &= E_0 \left[\frac{1}{T}\sum_{t=1}^T\sum_{j=1}^K P_{\theta^*}(Z_{t-1}=i,Z_t=j | Y_0,Y_1,\ldots,Y_T) \right] \\
    &= E_0 \left[\frac{1}{T}\sum_{t=1}^T P_{\theta^*}(Z_{t-1}=i | Y_0,Y_1,\ldots,Y_T) \right] \\
    &= \frac{1}{T}\sum_{t=1}^T P_{\theta^*}(Z_{t-1} = i) = P_{\theta^*}(Z_{t-1} = i),
\end{align*}
which is bounded away from 0. Similar argument can be made for the expectation involving $w_{j,\theta^*}$. Hence, similar to the discussion following Assumption~\ref{signalstrength}, when $w_{j,\theta}$ and $w_{ij,\theta}$ are continuous in $\theta$, we expect that the expectations in Assumption~\ref{denomtransition} are bounded away from 0 for $\theta$ in a neighborhood of $\theta^*$.

For our next condition, we define the following function of $P_Z$, the transition probability matrix of $\{Z_{t}\}$,
\begin{equation}
    \Xi(P_Z) = \frac{1}{2} \max\left(\max_{1\leq i,k \leq K}\sum_{j=1}^K \max_{l\neq j} \frac{\left|p_{ij}p_{kl}-p_{il}p_{kj}\right|}{p_{ij}p_{kl}+p_{il}p_{kj}}, \max_{1\leq i,k \leq K} \sum_{j} \max_{l \neq j} \frac{|p_{ji}  p_{lk}  -  p_{jk} p_{li}|}{p_{ji}   p_{lk} + p_{li}  p_{jk}} \right).
\end{equation}
This quantity is closely related to the approximation error, and we impose the following condition to control the approximation error. 
\begin{assumption}[Transition probability matrix]\label{cond:approximationerror}
There exists a constant $\phi < 1$, such that $\Xi(P_Z) \leq \phi$ for all $P_Z$ in the set $\{P_Z: p = \vectorize{P_Z}, \  \theta = (\beta^\top, p^\top, \sigma^\top )^\top \in \mathcal{B}(r,\theta^*)\}$. 
\end{assumption}
\begin{lemma}[Approximation error]\label{lemma:approximationerror}
Under Assumption~\ref{cond:approximationerror}, for all $j \in \{1,\ldots,K\}$ and all $y_0^T \in (\mathbb{R}^d)^{T+1}$,
\begin{equation}
    \left|m_{j,\theta}(y_{t-s}^{t+s}) - w_{j,\theta}(y_0^T)\right| \leq 3\phi^s, \quad \textnormal{and} \quad  \left|m_{ij,\theta}(y_{t-s}^{t+s}) - w_{ij,\theta}(y_0^T)\right| \leq 3\phi^{s-1}.
\end{equation}
\end{lemma}
\noindent Assumption~\ref{cond:approximationerror} restricts the transition matrix of $\{Z_t\}$ in a way that no entry of this matrix is too close to 0 or 1. For instance, in the case of binary $Z_t$, i.e., $K=2$, $\Xi(P_Z)$ can be replaced with the simpler quantity $|p_{11}p_{22} - p_{12}p_{21}|/(p_{11}p_{22} + p_{12}p_{21})$, which will be bounded away from 1 if $p_{ij}$ are all bounded away from 0 and 1.
Under this condition, Lemma~\ref{lemma:approximationerror} shows that the difference between the approximate and the exact conditional expectations will be exponential in $s$, uniformly in $Y$.

The next condition is on the minimum and maximum eigenvalues of the covariance matrix of $Y$, and will be useful in establishing the restricted eigenvalue (RE) condition. The RE condition is frequently imposed in the study of regularized estimators \citep{loh2012high,basu2015regularized}, and is also essential to our analysis.

\begin{assumption}[Minimum and maximum eigenvalues of covariance matrices]\label{cond:minmaxeigenvalue}
There exist constants $\rho_{\min}$ and $\rho_{\max}$ such that $0 < \rho_{\min} \leq \rho_{\max}<+\infty$, and for $\theta \in \mathcal{B}(r;\theta^*)$ and $j \in \{1,\ldots,K\}$,
\begin{align}
    \rho_{\min} &\leq \lambda_{\min}\left\{E_0\left[\frac{1}{T}\sum_{t=1}^T w_{j,\theta}(Y_0^T)Y_{t-1}Y_{t-1}^\top\right]\right\} \nonumber \\
    &\leq \lambda_{\max}\left\{E_0\left[\frac{1}{T}\sum_{t=1}^T w_{j,\theta}(Y_0^T)Y_{t-1}Y_{t-1}^\top\right]\right\} \leq \rho_{\max},
\end{align}
and
\begin{equation}
    \rho_{\min} \leq \lambda_{\min}\left\{E_0\left[Y_tY_t^\top\right]\right\} \leq \lambda_{\max}\left\{E_0\left[Y_tY_t^\top\right]\right\} \leq \rho_{\max}.
\end{equation}
\end{assumption}
\noindent The matrix $E_0[\sum_{t=1}^T w_{j,\theta}(Y_0^T)Y_{t-1}Y_{t-1}^\top/T]$ can be regarded as a weighted covariance matrix, where the weights are given by the conditional expectation of the unobserved variable $I\{Z_t = j\}$ under parameter value $\theta$. Assumption~\ref{cond:minmaxeigenvalue} requires that the minimum and maximum eigenvalues of the covariance matrix of $Y$ and the weighted covariance matrices are bounded away from 0 and infinity, respectively. 

\begin{assumption}[Geometric $\beta$-mixing]\label{betamixing}
Let $b_{\textnormal{mix}}(l)$ be the $\beta$-mixing coefficient of the process $\{Y_t\}$. Then, $b_{\textnormal{mix}}(l) \leq 2\exp(-c l^{\gamma_1})$ for some positive constants $c$ and $\gamma_1$ and all $l \in \mathbb{N}^{+}$. 
\end{assumption}
\noindent Assumption~\ref{betamixing} states that the process is geometrically $\beta$-mixing. Under Assumption~\ref{cond: operatornorm}, Lemma~\ref{lemma: msvarstationarity} implies that the process $\{(Y_t,Z_t)\}$ is geometrically ergodic, and thus by Proposition 2 in \citet{liebscher2005towards}, $\{Y_t\}$ is geometrically $\beta$-mixing. Moreover, given the discussion following Lemma~\ref{lemma: msvarstationarity}, it is reasonable to expect that $\gamma_1$ can be taken to be 1, at least for large $l$. 

Finally, to establish the RE condition uniformly over $\Theta$, we assume the following condition so that the entropy of a certain function class is controlled. We will assume a similar condition in Assumption~\ref{randomentropy}, and there we provide more discussion on the plausibility of assumptions of this type. Let $\{\tilde{Y}_{1-s}^{1+s},\tilde{Y}_{2-s}^{2+s},\ldots,\tilde{Y}_{n-s}^{n+s}\ldots\}$ be a sequence of i.i.d. random vectors whose marginal distribution is the same as the stationary distribution of $Y_{t-s}^{t+s}$.
\begin{assumption}[Upper bound on random entropy]\label{randomentropyRE} For some constant $\tilde C$ and $0<a<1$,
\begin{multline*}
    P\Bigg(\max_j\sup_{\tilde\theta \in \mathcal{B}(r,\theta^*)} \left[ \frac{1}{N}\sum_{n=1}^N \left\|\frac{\partial m_{j,\theta}(\tilde{Y}_{n-s}^{n+s})}{\partial \theta}\rvert_{\theta = \tilde\theta}\frac{\partial m_{j,\theta}(\tilde{Y}_{n-s}^{n+s})}{\partial \theta}\rvert_{\theta = \tilde\theta}^\top\right\|_2^2 \right]^{1/2}  \\
    > \frac{N^a}{\tilde{C}|S|(\log K + \log d)}\Bigg) \leq \tilde{u}(N,d),
\end{multline*}
for some function $\tilde u$ such that $\tilde{u}(T/(c\log T),d)\log T \rightarrow 0$ as $T \rightarrow \infty$ for any constant $c$.
\end{assumption}
\noindent The precise definition of the constant $\tilde C$ is given in Supplementary Appendix~\ref{app:proofRE}.

With these assumptions, we now establish the RE condition. Define $\gamma = (1/\gamma_1 + 1)^{-1}$, and note that $\gamma <1$. Furthermore, Lemma~\ref{lemma: msvarstationarity} implies that $Y_t$ is a sub-Gaussian random vector. Let $K_Y = \sup_{\nu \in \mathbb{R}^d,\|\nu\|=1}\sup_{k \geq 1}(E|\nu^\top Y_t|^k)^{1/k}k^{-1/2}$ and $\tilde K_Y:= 2 K^2_Y$. Moreover, define the following sets of parameter values:
\begin{align*}
    \Theta_\beta &= \{\beta: \|\beta - \beta^*\|_2 \leq r, \  \|(\beta - \beta^*)_{S^C}\|_1 \leq 4\sqrt{|S|}\|\beta-\beta^*\|_2\}; \\
    \Theta &= \{\theta = (\beta^\top,p^\top,\sigma^\top )^\top: \|\theta - \theta^*\|_2 \leq r, \ \beta \in \Theta_\beta\}.
\end{align*}
\begin{lemma}[Restricted Eigenvalue]\label{RE}
Suppose that Assumptions~\ref{cond: operatornorm} and \ref{cond:approximationerror}--\ref{randomentropyRE} hold. Define the following quantities: 
\begin{equation*}
    \alpha = \frac{\rho_{\min}}{3}, \quad \tau_{RE} = \frac{\rho_{\min}\log d}{3T^{1/5}}.
\end{equation*}
Then, for some constants $\tilde c$, $C$, and $C_2$ independent of $T$ and $d$, for $s \asymp \log T$ and sample size $T$ sufficiently large such that
\begin{equation*}
    T \geq 72\tilde{c}(\log T)(972K_Y^2/\rho_{\min})^2,
\end{equation*}
and
\begin{multline*}
    \sqrt{T/(\tilde{c}\log T)}\geq 1944C \rho_{\min}^{-1}\max\Bigg\{\rho_{\min}/972, 4K_Y^2+1, \\
    2C_2\sqrt{\log K + 2T^{1/5}\{1+(1+\log 30)/\log d\}}\Bigg\},
\end{multline*}
we have
\begin{equation*}
    v^\top \left[\frac{1}{T}\sum_{t=1}^{T} \textnormal{Id}_d \otimes  m_{j,\theta}(Y_{t-s}^{t+s})Y_{t-1}Y_{t-1}^\top\right]v \geq \alpha \|v\|_2^2 -\tau_{RE}\|v\|_1^2,    
\end{equation*}
for all $v \in \mathbb{R}^{d^2}$, $j \in \{1,\ldots,K\}$ and $\theta \in \Theta$, with probability at least $u_{RE}(T,d)$ such that $u_{RE}(T,d) \rightarrow 1$ as $T \rightarrow \infty$.
\end{lemma}

\noindent The precise specifications of the constants and the explicit form of $u_{RE}$ are given in Supplementary Appendix~\ref{app:proofRE}. 

For our next condition, we consider the rate of convergence in uniform law of large numbers over certain function classes. As we will see, this rate of convergence is closely related to the estimation error of the estimate from our proposed algorithm. To this end, we define the following
functions:
\begin{align*}
    f^{ijk}_\theta(Y_{t-s}^{t+s}) &= Y_{t-1,k}(Y_{t,i}-\beta_{ji}^{*\top}Y_{t-1})m_{j,\theta}(Y_{t-s}^{t+s}); \\
    f^{\sigma,j}_{\theta,\tilde\beta}(Y_{t-s}^{t+s}) &= \frac{1}{d}m_{j,\theta}(Y_{t-s}^{t+s}) \left\|Y_{t} - \left(\textnormal{Id}_d \otimes Y_{t-1}\right)^\top \tilde\beta_j\right\|_2^2,
\end{align*}
for $\theta \in \Theta$ and $\tilde\beta \in \Theta_\beta$. We assume that uniform law of large numbers holds for certain classes of functions, with a suitable rate of convergence. 

\begin{assumption}\label{DB}
Suppose that for some small probability $\delta_1(T,d)$ such that $\delta_1 \rightarrow 0$ as $T \rightarrow \infty$, there exist $\Delta$, $\Delta_\sigma$ and $\Delta_p$, all of which are functions of $T$, $d$ and $K$, such that the following holds:
\small
\begin{align*}
    &\max_{i,j,k}\sup_{\theta \in \Theta} \left|\frac{1}{T}\sum_{t=1}^{T} f^{ijk}_\theta(Y_{t-s}^{t+s})-E\left[f^{ijk}_\theta(Y_{t-s}^{t+s})\right]\right| \leq \Delta; \\
    &\max_{j}\sup_{\theta \in \Theta, \tilde\beta \in \Theta_\beta} 
    \left|\frac{1}{T}\sum_{t=1}^T f^{\sigma,j}_{\theta,\tilde\beta}(Y_{t-s}^{t+s}) 
    \left\{\frac{1}{T}\sum_{t=1}^T m_{j,\theta}(Y_{t-s}^{t+s})\right\}^{-1}
    - E\left[ f^{\sigma,j}_{\theta,\tilde\beta}(Y_{t-s}^{t+s})\right]
    E\left[m_{j,\theta}(Y_{t-s}^{t+s})\right]^{-1}
    \right| \leq \Delta_\sigma;
\end{align*}
and
\begin{multline*}
    \max_{i,j} \sup_{\theta \in \Theta} \Bigg|
    \left\{\frac{1}{T}\sum_{t=1}^T m_{ij,\theta}(Y_{t-s}^{t+s})\right\}\left\{\frac{1}{T}\sum_{t=1}^T \sum_{j=1}^K m_{ij,\theta}(Y_{t-s}^{t+s})\right\}^{-1} - \\
    \left\{E\left[m_{ij,\theta}(Y_{t-s}^{t+s})\right]\right\}\left\{E\left[\sum_{j=1}^K m_{ij,\theta}(Y_{t-s}^{t+s})\right]\right\}^{-1}\Bigg| 
    \leq \Delta_p,
\end{multline*}
\normalsize
with probability at least $1-\delta_1$.
\end{assumption}

\noindent This assumption is similar to the deviation bound condition in, for example, \citet{loh2012high} and \citet{wong2020lasso}, in that it assumes the difference between the sample average of the gradient of the objective function and its population counterpart is controlled. However, in our case, we need the difference to be controlled uniformly over the parameter space $\Theta$. This is because, in each EM iteration, the objective function in the M-step depends on the parameter estimate from the previous iteration, which itself is random and changes over iterations.

With Assumption~\ref{DB}, we are now ready to state our main theorem on the estimation error. Define a vector $D_\beta = (D_{\beta,1}^\top,\ldots,D_{\beta,K}^\top)^\top$, where $D_{\beta,j} = ( (\beta_1^*-\beta_j^*)^\top,\ldots,(\beta_{j-1}^*-\beta_j^*)^\top,(\beta_{j+1}^*-\beta_j^*)^\top,\ldots,(\beta_{K}^*-\beta_j^*)^\top)^\top$. The vector $D_\beta$ captures the difference in the regression coefficients between different regimes. Moreover, define constants $\eta$ and $\tau$ as
\begin{equation*}
     \eta = \max\left\{\sqrt{\left(\frac{2\rho_{\max}}{\alpha}\right)^2 + 2K\left(\frac{2\rho_{\max}r}{d}\right)^2}, \sqrt{2}\right\}, \quad 
     \tau = 4\kappa\eta\left(1+ \frac{8.2\sqrt{K}\rho_{\max}\kappa\eta r}{d}\right).
\end{equation*}

\begin{theorem}[Estimation error]\label{estimationerror}
Suppose that Assumptions~\ref{cond: operatornorm}--\ref{DB} hold and $\Delta$, $\Delta_p$ and $\Delta_\sigma$ in Assumption~\ref{DB} are such that $\max\{\sqrt{|S|}\Delta,\Delta_p,\Delta_\sigma\} = o(1)$ as $T \to \infty$. Moreover, suppose that $\tau < 1$. Then, for the regularized approximate EM algorithm with initialization $\theta^{(0)} \in \Theta$ and with $\lambda^{(q)}$ chosen such that
\begin{multline}\label{lambdachoice}
    \lambda^{(q)} = \tau^q \frac{\alpha}{4\sqrt{|S|}}\left(1 + \frac{8.2\sqrt{K}\rho_{\max}\eta\kappa r}{d}\right)^{-1}\|\theta^{(0)} - \theta^*\|_2 + \frac{1-\tau^q}{1-\tau} \max\Bigg\{2\Delta + 2C_1\phi^s, \\
    \frac{\alpha}{\sqrt{|S|}}\left(\frac{6}{\alpha} \phi^s\rho_{\max}(K-1)^{1/2}\|D_\beta\|_2 + K \Delta_p + C_2\phi^s + \sqrt{K} \Delta_\sigma + C_3\phi^s  \right)\Bigg\} 
\end{multline}
for all $q \geq 1$ for some constants $C_1$, $C_2$ and $C_3$, we have the following upper bound on the estimation error for all $q \geq 1$
\begin{equation*}
    \|\theta^{(q)} - \theta^*\|_2 \leq \left(1 + \frac{8.2\sqrt{K}\rho_{\max}\eta\kappa r}{d}\right) \frac{4\lambda^{(q)}\sqrt{|S|}}{\alpha},
\end{equation*}
with probability at least $1-(\delta+\delta_1)$, for $T$ sufficiently large.
\end{theorem}
\noindent The small probabilities $\delta = 1- u_{RE}(T,d)$ and $\delta_1$ are defined in Lemma~\ref{RE} and Assumption~\ref{DB}, respectively. The precise specification of the constants $C_1$, $C_2$ and $C_3$ is given in the proof of Theorem~\ref{estimationerror} in Supplementary Appendix~\ref{app:proofmainthm}. As we will see, $\max\{\sqrt{|S|}\Delta,\Delta_p,\Delta_\sigma\}$ ultimately determines the estimation error of our estimate, and the assumption that it converges to 0 as $T$ approaches infinity is sensible. Indeed, under appropriate conditions, we can show that this assumption will hold provided that $d$ and $|S|$ do not increase too fast as $T$ increases. We provide more discussion on this later in Proposition~\ref{errororder}. The idea is that, when $d$ and $|S|$ do not increase too fast, we can control the entropy of certain function classes including $\{f_\theta^{ijk}: \theta \in \Theta, \  1 \leq i,k \leq d, \ 1\leq j \leq K\}$ so that we can apply uniform concentration results for $\beta$-mixing processes. 
The constant $\tau$ becomes smaller as $\kappa$ gets smaller. As discussed earlier, $\kappa$ serves as a measure of inverse signal strength. Thus, with a strong enough signal-to-noise ratio, we expect that $\kappa$ can be small enough such that $\tau$ is below 1.

Substituting in the expression for our choice of $\lambda$ in \eqref{lambdachoice}, we get a more explicit upper bound on the estimation error:
\begin{align}\label{error}
    \|\theta^{(q)} - \theta^*\|_2 &\leq \tau^{q} \|\theta^{(0)} - \theta^*\|_2 + \left(1 + \frac{8.2\sqrt{K}\rho_{\max}\eta\kappa r}{d}\right) \frac{4\sqrt{|S|}}{\alpha}\frac{1-\tau^{q}}{1-\tau} \max\Bigg\{2\Delta + 2C_1\phi^s, \nonumber \\
    &\quad\frac{\alpha}{\sqrt{|S|}}\left(\frac{6}{\alpha} \phi^s\rho_{\max}(K-1)^{1/2}\|D_\beta\|_2 + K \Delta_p + C_2\phi^s + \sqrt{K} \Delta_\sigma + C_3\phi^s  \right)\Bigg\}.
\end{align}
\noindent The upper bound in \eqref{error} consists of three types of terms. The term $\tau^q\|\theta^{(0)} - \theta^*\|_2$ is the `optimization error', which converges geometrically to 0 as $q$, the number of EM iterations, approaches infinity. Hence, the `optimization error' can be made negligible by selecting a sufficiently large value of $q$, that is, by running sufficient EM iterations. The terms involving $\phi^s$ are the `approximation error'. In particular, if we choose $s = \log T/(-2\log\phi)$, then $\phi^s = T^{-1/2}$. As will be seen from Proposition~\ref{errororder} below, with this choice of $s$, the approximation error is dominated by the $\Delta$ terms. We call $\Delta, \Delta_p$ and $\Delta_\sigma$ `statistical error', as each of them is a difference between a population-level quantity and its sample counterpart, and can be controlled using concentration results. We also note that the constants appearing in the estimation error bound may not be optimal.

We assume in Theorem~\ref{estimationerror} that the initialization $\theta^{(0)}$ lies in $\Theta$. In fact, the conclusion in Theorem~\ref{estimationerror} still holds when $\theta^{(0)}$ is randomly chosen in $\mathcal{B}(r;\theta^*)$ independent of the observed data.
In this case, to analyze the estimation error in the first iteration, we need concentration results similar to those in Assumption~\ref{DB} to hold pointwise for $\theta \in \mathcal{B}(r;\theta^*)$. However, such pointwise results are considerably easier to establish compared to the uniform concentration results in Assumption~\ref{DB}. We can then marginalize over the distribution of $\theta^{(0)}$. Consequently, under random initialization with $\theta^{(0)}$ independent of the observed data, similar upper bounds on the estimation error can be established, where the high probability statement is now with respect to the joint distribution of $\theta^{(0)}$ and the observed $Y$'s.

Next, we characterize the magnitude of the statistical error in Assumption~\ref{DB} under some conditions. For this purpose, let $\{\tilde{Y}_{1-s}^{1+s},\tilde{Y}_{2-s}^{2+s},\ldots,\tilde{Y}_{n-s}^{n+s}\ldots\}$ be a sequence of i.i.d. random vectors whose marginal distribution is the same as the stationary distribution of $Y_{t-s}^{t+s}$.
\begin{assumption}[Upper bound on random entropy]\label{randomentropy}
Either (a) there exists a sequence $l(T,d) \geq 1$ such that
\begin{multline*}
    P\Bigg( \sup_{\tilde\theta \in \mathcal{B}(r,\theta^*)}\max_{i,j,k} \Bigg\|\frac{1}{N}\sum_{n=1}^N\left\{\frac{\partial m_{j,\theta}(\tilde{Y}_{n-s}^{n+s})}{\partial \theta}\rvert_{\theta = \tilde{\theta}}\right\} \times \\
    \left\{h^{ijk}(\tilde{Y}_{n-1}^n)\right\}^2 \left\{\frac{\partial m_{j,\theta}(\tilde{Y}_{n-s}^{n+s})}{\partial \theta}\rvert_{\theta = \tilde{\theta}}\right\}^\top\Bigg\|_2 > l(N,d) \Bigg)    \leq u(N,d),
\end{multline*}
for some sequence $u(T,d)$ such that $u(T/(c\log T),d)\log T \rightarrow 0$ as $T \rightarrow \infty$ for a constant $c$, where $h^{ijk}(\tilde{Y}_{n-1}^n) = \tilde{Y}_{n-1,k}(\tilde{Y}_{n,i}-\beta_{ji}^{*\top}\tilde{Y}_{n-1})$;

or (b) there exists sequence $l(T,d) \geq 1$ such that for $j \in \{1,\ldots,K\}$,
\begin{equation*}
     P\left(\sup_{\tilde\theta \in \mathcal{B}(r,\theta^*)} \left[ \frac{1}{N}\sum_{n=1}^N \left\|\frac{\partial m_{j,\theta}(\tilde{Y}_{n-s}^{n+s})}{\partial \theta}\rvert_{\theta = \tilde\theta}\frac{\partial m_{j,\theta}(\tilde{Y}_{n-s}^{n+s})}{\partial \theta}\rvert_{\theta = \tilde\theta}^\top\right\|_2^2 \right]^{1/2} > l(N,d)\right) \leq \tilde{u}_j(N,d),
\end{equation*}
for some sequence $\tilde{u}_j(T,d)$ such that $\tilde{u}_j(T/(c\log T),d)\log T \rightarrow 0$ as $T \rightarrow \infty$ for a constant $c$.
\end{assumption}
\noindent Assumption~\ref{randomentropy} is useful for controlling the entropy under an empirical norm of the function class $f_\theta^{ijk}$ when varying $\theta$ over $\Theta$. Specifically, we relate the entropy of this function class to the entropy of the parameter set $\Theta$, where Assumption~\ref{randomentropy} is useful in showing that functions in this class are Lipschitz in $\theta$. Note that the entropy under the empirical norm is random, where the randomness results from the randomness in the sample used to define the empirical norm. We take the same approach as in Section~5.1 in \citet{geer2000empirical}, and assume that this random entropy number is upper bounded by a deterministic sequence with high probability. As in the discussion following Lemma~5.1 in \citet{geer2000empirical}, the random entropy number can also be controlled if the function class under consideration is a Vapnik-Chervonenkis (VC) subgraph class. Showing that $\{f_\theta^{ijk}: \theta \in \Theta\}$ is indeed a VC subgraph class is left for future research.

Under Assumption~\ref{randomentropy}, the following proposition quantifies the magnitude of $\Delta$.
\begin{proposition}\label{errororder}
Under Assumptions~\ref{cond: operatornorm}, \ref{betamixing} and \ref{randomentropy}, 
\begin{multline*}
    \max_{i,j,k}\sup_{\theta\in\Theta}\left|\frac{1}{T}\sum_{t=1}^Tf_\theta^{ijk}(Y_{t-s}^{t+s}) -E\left[f_\theta^{ijk}(Y_{t-s}^{t+s})\right]\right|  = \\
    O_P\left(\sqrt{\frac{|S|l(T,d)(\log T)^3 (\log K + \log d) + (\log T)^4}{T}} \right),
\end{multline*}
when $(\log d + \log K)^2\log T = o(T)$. 
\end{proposition}
\noindent In the case that 
$\sup_{\tilde\theta \in \mathcal{B}(r,\theta^*)}\left\|\frac{\partial m_{j,\theta}(\tilde{Y}_{n-s}^{n+s})}{\partial \theta}\rvert_{\theta = \tilde\theta}\frac{\partial m_{j,\theta}(\tilde{Y}_{n-s}^{n+s})}{\partial \theta}\rvert_{\theta = \tilde\theta}^\top\right\|_2^2$ 
has a bounded expectation, Markov inequality implies that we can take $l(T,d)$ to be $(\log T)^2$, and condition (b) in Assumption~\ref{randomentropy} is satisfied with $\tilde{u}_j(T,d)$ on the order of $1/(\log T)^2$. We demonstrate empirically in Supplementary Appendix~\ref{app:signalempirical} that bounded expectation is plausible. As a result, $\Delta$ is of the order $O_P\left(\sqrt{|S|(\log T)^5(\log K + \log d)/T}\right)$, and the statistical error is of the order $O_P\left(\sqrt{|S|^2(\log T)^5(\log K + \log d)/T}\right)$ when $\Delta_p$ and $\Delta_\sigma$ are of the same order as $\Delta$. 

Compared to the estimation error of $l_1$-regularized regression with i.i.d. data, additional $\log T$ factors appear under the square root in our convergence rate. One of these factors appears due to the temporal dependence in the time series setting, as we essentially divide the entire time series into $\log T$ blocks that are approximately independent in order to show a uniform concentration result. A coupling argument as in \citet{merlevede2011bernstein} might remove this factor, but it is unclear whether the coupling technique is directly applicable when the goal is to establish uniform concentration results over a class of functions. Another $(\log T)^2$ appears when we use Markov inequality to control the random entropy in Assumption~\ref{randomentropy} and take $l(T,d) = (\log T)^2$, as outlined in the previous paragraph. The Markov inequality may be crude in this case, as it ignores the fact that a sample average appears in the quantities we aim to control in Assumption~\ref{randomentropy}. In light of this, concentration results may be useful in showing that we may in fact be able to take $l(T,d)$ to be of a lower order. 

Finally, an additional $(\log T)^2$ results from the entropy of the class $\Theta$, which contains vectors that are \textit{weakly sparse}, in the sense that the $l_1$-norm on the inactive set is small. This entropy appears in the convergence rate in our uniform concentration result. With i.i.d. data, a sample splitting approach might be employed to avoid the need for uniform concentration, where in each iteration a new block of data is used. This way, the parameter estimate prior to an iteration is independent of the data used in the next iteration to perform the update. However, in the time series setting, even if we divide the data into non-overlapping blocks and use different block for each iteration, these non-overlapping blocks are still dependent---although this dependence can be made weaker by using blocks further apart from each other in time. This means that a sample splitting approach does not remove the need for uniform concentration results in a trivial way. Alternatively, we can consider thresholding the parameter estimate after each iteration so that the parameter estimates over the iterations vary in a smaller set containing only \textit{exactly} sparse vectors. However, the sparsification step will introduce additional estimation error by introducing false negatives, and therefore the threshold level needs to be carefully chosen so that such error can be controlled. Nonetheless, we present a variant of the proposed algorithm that includes an additional thresholding step in Supplementary Appendix~\ref{app:EMTalgorithm}. 

\section{Numerical Experiments}\label{sec:numerical}

\subsection{Simulation studies}\label{sec: simulation}
In this section, we use simulations to illustrate the performance of the proposed algorithm. We consider the case where $Z_t$ takes value in $\{1,2\}$ (i.e., there are 2 regimes), with transition probability $P(Z_t=1|Z_{t-1}=1) = 0.7$ and $P(Z_t=1|Z_{t-1}=2) = 0.3$. The dimension $d$ is varied in $\{30,90\}$, and the sample size $T$ is varied within $\{500,1000,2000\}$.

We consider two settings for generating the regression coefficient matrices. For Setting~I, we first define matrices $A, \tilde A$, both in $\mathbb{R}^{3\times 3}$, with
\[
A = 
\begin{bmatrix}
0.5 & 0.1 & 0 \\
0 & 0.1 & 0.2 \\
0 & 0.3 & 0.3
\end{bmatrix}, \quad 
\tilde A = 
\begin{bmatrix}
0.3 & 0   & 0.2 \\
0.2 & 0   & 0   \\
0   & -0.5 & -0.3
\end{bmatrix}.
\]
We then set $B_1 = \textnormal{Id}_{d/3} \otimes A$; that is, $B_1$ is a block diagonal matrix with all $d/3$ diagonal blocks set to $A$. The matrix $B_2$ is the same as $B_1$ except that the $k$-th diagonal block is changed to $\tilde A$, for $k \in \{1,2,5,10\}$ when $d=30$ and for $k \in \{1,2,5,10,11,12,15,20,21,22,\allowbreak 25,30\}$ when $d=90$. For Setting~II, we first generate an adjacency matrix $A^{\mathrm{adj}} \in \mathbb{R}^{d\times d}$ randomly by drawing its entries independently from a Bernoulli$(0.1)$ distribution. When $d=30$, we let $(B_1)_{ij} = 0$ if $A^{\mathrm{adj}}_{ij} =0$; and for $A^{\mathrm{adj}}_{ij} =1$, $(B_1)_{ij}$ takes value $0.2,-0.2,0.4,-0.4$ with probability $0.45,0.45,0.05,0.05$, respectively. When $d=90$, the active regression coefficients take value $0.12,-0.12,0.24,-0.24$ with probability $0.45,0.45,0.05,0.05$, respectively, so that the spectral norm of $B_1$ is the same as in the case of $d=30$ and is below 1. To generate $B_2$, we randomly subset $50\%$ of the active entries in $B_1$ and flip its sign. In both Settings~I and II, the conditional variance $\sigma_j^2$ is set to 1.

To generate the observed data $\{Y_t\}$, we use a burn-in period of 5,000 steps. In each iteration of the algorithm, the tuning parameter $\lambda$ is selected via a 10-fold cross validation. The algorithm is terminated when $\|\theta^{(q)} - \theta^{(q-1)}\|_\infty$ is below a tolerance level, set to $10^{-4}$. The initial estimate of the regression coefficients $\beta^{(0)}$ is generated from $N_{2d^2}(0,{0.5}^2\textnormal{Id}_{2d^2})$, and ${p_{ij}}^{(0)} = 0.5$, $\sigma^{(0)}=1$. Such random initialization has been employed in previous works, including \citet{hao2017simultaneous}, and we found it reliable across the simulation settings we considered, especially with larger sample sizes. For each simulated dataset, we use 5 random initializations, and when different initializations lead to meaningfully different parameter estimates, we select the initialization and consequently the parameter estimate that results in the largest value of the expected log-likelihood (unpenalized objective function). 

For comparison, we define an oracle estimator $\hat\theta_{\mathrm{oracle}} = (\hat\beta_{\mathrm{oracle}}^\top,\hat p_{\mathrm{oracle}}^\top,(\hat\sigma^2_{\mathrm{oracle}})^\top)^\top$, assuming that we observe $\{Z_t\}$. Specifically, $\hat\beta_{j, \mathrm{oracle}}$ is estimated with the lasso on the subset of data corresponding to regime $j$, defined as $\{Y_t\}_{t\in \mathcal{T}_j}$ where $\mathcal{T}_j = \{t: Z_t=j\}$. The transition probability $\hat p_{ij, \mathrm{oracle}}$ is defined as $\sum I\{Z_{t-1} = i,Z_t=j\}/\sum I\{Z_{t-1}=i\}$. The conditional variance estimator $\hat\sigma^2_{j,\mathrm{oracle}}$ is the mean residual sum of squares across dimension and time $t$ in $\mathcal{T}_j$. We also define another estimator $\hat\theta_{\mathrm{exact}}$ obtained via the regularized EM algorithm that uses the exact conditional expectations $w_{j,\theta}(Y_0^T)$ and $w_{ij,\theta}(Y_0^T)$ in the objective function in the M-step without using any approximation. These exact conditional expectations are computed using the forward-backward iterations. We use $\hat\theta_{\mathrm{exact}}$ as a benchmark to examine the effect of the approximation in the E-step on the estimation error. For the regression coefficients, we define the estimation error of a generic estimator $\hat\beta$ as the $L_2$-norm of the difference between the estimate and the truth, i.e., $\|\hat\beta - \beta^*\|_2$. Various estimators have been proposed in the literature to estimate the conditional variance in high-dimensional linear regressions \citep{sun2012scaled, yu2019estimating}, but the simple estimator $\hat\sigma^2_{\mathrm{oracle}}$ worked reasonably well in our simulations (assuming that $Z_t$ is observed).

\begin{figure}[t]
    \centering
    \includegraphics[width=0.95\textwidth]{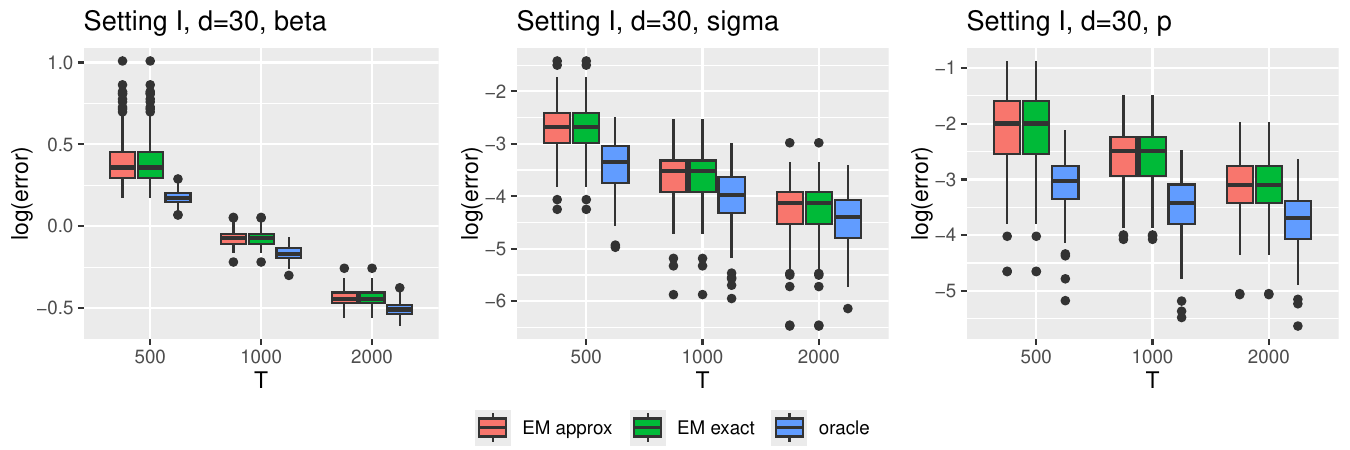}
    \includegraphics[width=0.95\textwidth]{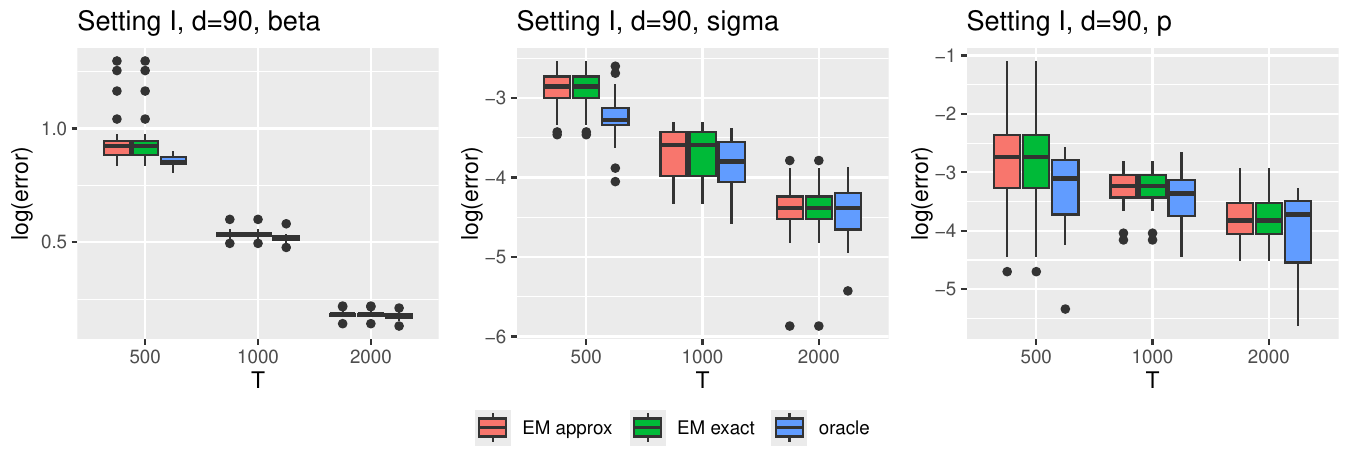}
    \caption{Estimation error of regression coefficients $\beta$ (left column), conditional variance $\sigma^2$ (middle column), and transition probabilities $p$ (right column) in Setting I, when we vary $d \in \{30, 90\}$ and $T \in \{500,1000,2000\}$. Log error is defined as $\log(\|\hat\beta - \beta^*\|_2)$, $\log(\|\hat\sigma^2 - (\sigma^*)^2\|_2)$, and $\log(\|\hat p - p^*\|_2)$, respectively. Results are based on 100 simulation replications for $d=30$, and 20 for $d=90$.}
    \label{fig: Setting1}
\end{figure}

Figure~\ref{fig: Setting1} and Figure~\ref{fig: Setting2} show the results for Settings I and II, respectively, when we vary $d \in \{30, 90\}$ and $T \in \{500,1000,2000\}$. Results for $d=30$ are based on 100 simulation replications, while results for $d=90$ are based on 20 replications. When $d=30$, we observe that in general the EM algorithm has slightly larger estimation error than the oracle estimator, but the estimation errors are mostly comparable. Not surprisingly, the performance also improves as sample size increases. Furthermore, we observe an approximately linear relationship between logarithm of estimation error and logarithm of sample size, with slope approximately $-1/2$. The results for $d = 90$ display similar overall trend. In this case, the dimension of the parameter vector increases to $16,204$ and initialization becomes more challenging with smaller sample sizes. Furthermore, in Setting II, the magnitude of individual regression coefficient is also smaller compared with $d=30$. With sample size 500, we observe that the EM algorithm has larger estimation error compared to the oracle, especially in Setting II. But as sample size increases, the performance of the EM estimate improves greatly and becomes comparable to the oracle estimates. We also observe that in certain simulation replications the EM estimate has large estimation error, which may suggests more than 5 random initializations are required to locate a good initialization. It is worth noting that for $T=500$, the EM algorithm may have smaller estimation error in $\sigma^2$ compared to the oracle. This is because the oracle often slightly underestimates the noise variance, while the EM has larger estimation error in the regression coefficients and consequently larger residual sum of squares in general, making the mean residual sum of squares larger and closer to the conditional variance. We also observe that the estimation error of our approximate EM algorithm, which uses an approximation of the conditional expectations in the E-step, is very similar to that of the exact EM algorithm, which computes the exact conditional expectations. This suggests that the approximation error is negligible with our choice of $s = \lceil \log T \rceil$.

\begin{figure}[t]
    \centering
    \includegraphics[width=0.95\textwidth]{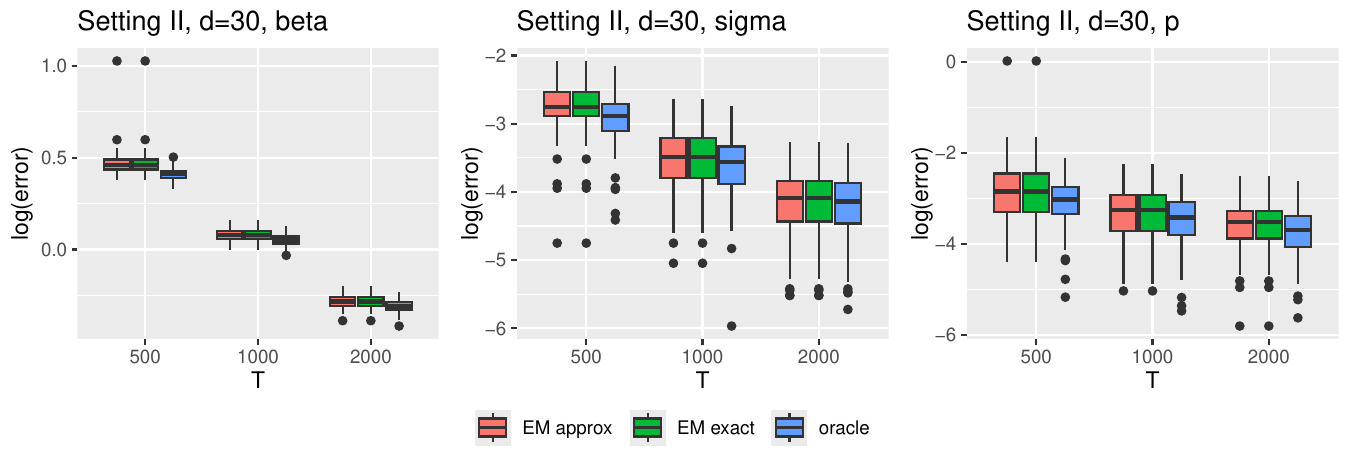}
    \includegraphics[width=0.95\textwidth]{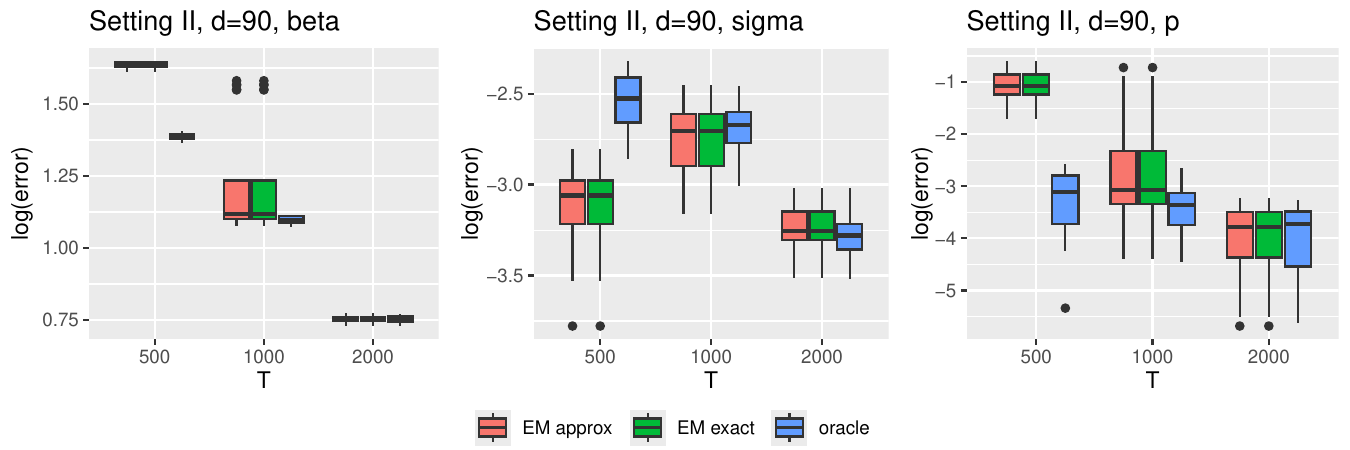}
    \caption{Estimation error of regression coefficients $\beta$ (left column), conditional variance $\sigma^2$ (middle column), and transition probabilities $p$ (right column) in Setting II, when we vary $d \in \{30, 90\}$ and $T \in \{500,1000,2000\}$. Log error is defined as $\log(\|\hat\beta - \beta^*\|_2)$, $\log(\|\hat\sigma^2 - (\sigma^*)^2\|_2)$, and $\log(\|\hat p - p^*\|_2)$, respectively. Results are based on 100 simulation replications for $d=30$, and 20 for $d=90$.}
    \label{fig: Setting2}
\end{figure}

Next, we consider an even more challenging setting with 3 regimes, where the conditional variance of the noise also differs across regimes. Specifically, the simulation setting is the same as Setting II, except that we add an additional regime with coefficient matrix $B_3 = B_1$ and $\sigma_3 = 0.5$. This results in $24,309$ parameters in total when $d=90$. The transition probabilities between regimes are given by $P_{1.} = (0.3, 0.3, 0.4)$, $P_{2.} = (0.2, 0.5, 0.3)$ and $P_{3.} = (0.5, 0.3,0.2)$, where $P_{j.}$ is the $j$-th row in the transition probability matrix $P_Z$. We refer to this setting as Setting III, and here we vary sample size $T \in \{500,1000,2000,5000\}$. The corresponding results are presented in Figure~\ref{fig: Setting3}, based on 50 simulation replications for $d=30$ and 20 simulation replications for $d=90$. We observe similar patterns as in Setting II. In particular, the estimation error of the EM estimates is larger than that of the oracle estimates for small sample size, but improves as sample size increases and becomes comparable to the oracle with moderate to large sample size. Moreover, the estimation errors of our approximate EM algorithm and the exact EM algorithm are again very similar.

\begin{figure}[t]
    \centering
    \includegraphics[width=0.95\textwidth]{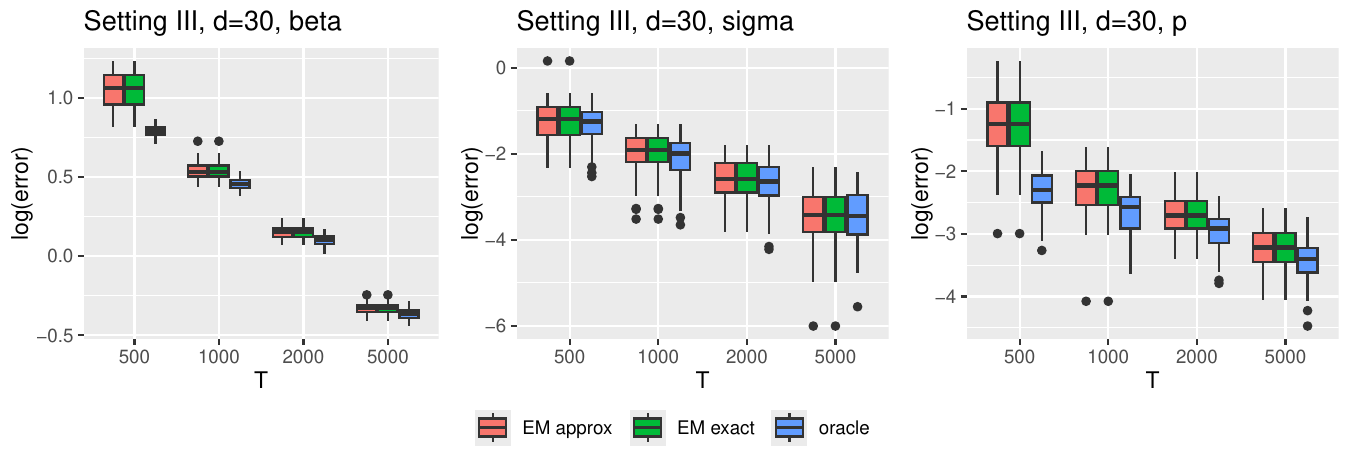}
    \includegraphics[width=0.95\textwidth]{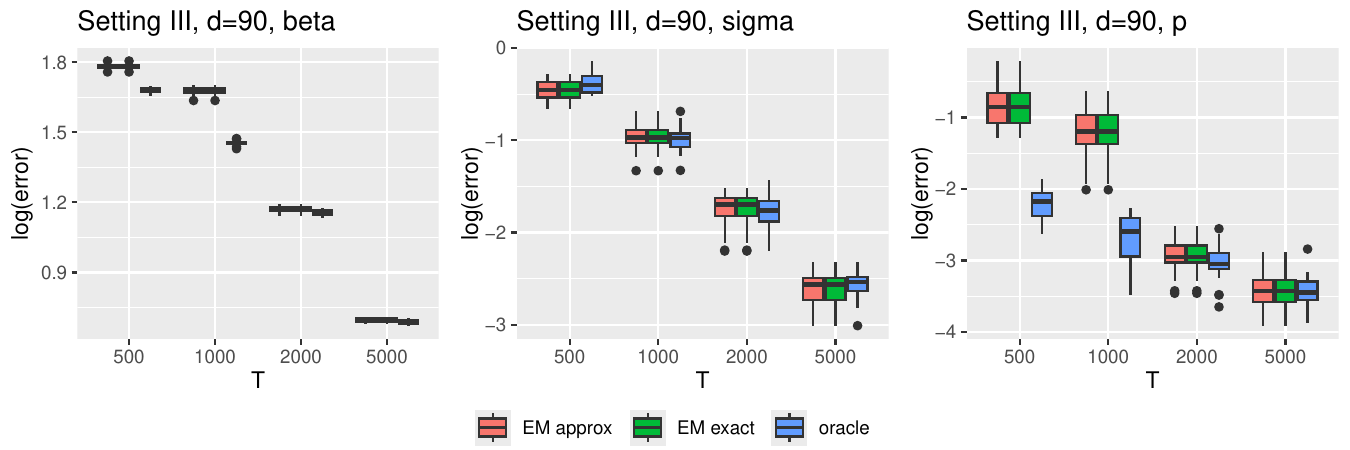}
    \caption{Estimation error of regression coefficients $\beta$ (left column), conditional variance $\sigma^2$ (middle column), and transition probabilities $p$ (right column) in Setting III with 3 regimes, when we vary $d \in \{30, 90\}$ and $T \in \{500,1000,2000, 5000\}$. Log error is defined as $\log(\|\hat\beta - \beta^*\|_2)$, $\log(\|\hat\sigma^2 - (\sigma^*)^2\|_2)$, and $\log(\|\hat p - p^*\|_2)$, respectively. Results are based on 50 simulation replications for $d=30$, and 20 for $d=90$.}
    \label{fig: Setting3}
\end{figure}

\subsection{Application to EEG data}\label{sec: EEG data}

We apply the approximate EM algorithm to analyze an EEG dataset, which measures the signals at 18 EEG channels during an epileptic seizure from a patient diagnosed with left temporal lobe epilepsy \citep{ombao2005slex}. The signals were measured over an approximately 228-second period with 100Hz frequency, resulting in 22,768 time steps in total. Our goal is to investigate how the brain connectivity network may change over time for a patient experiencing seizure. 

To apply our method, we first down-sample the observations to ten observations per second consistent with \cite{safikhani2022joint}, which results in $T=2,276$ time points. We then standardize the signal from each channel so that the overall mean and variance are 0 and 1, respectively. The top panel of Figure~\ref{fig: EEG measurement} depicts the standardized signals from the 18 channels over time.

\begin{figure}[t]
    \centering
    \includegraphics[width=0.9\linewidth]{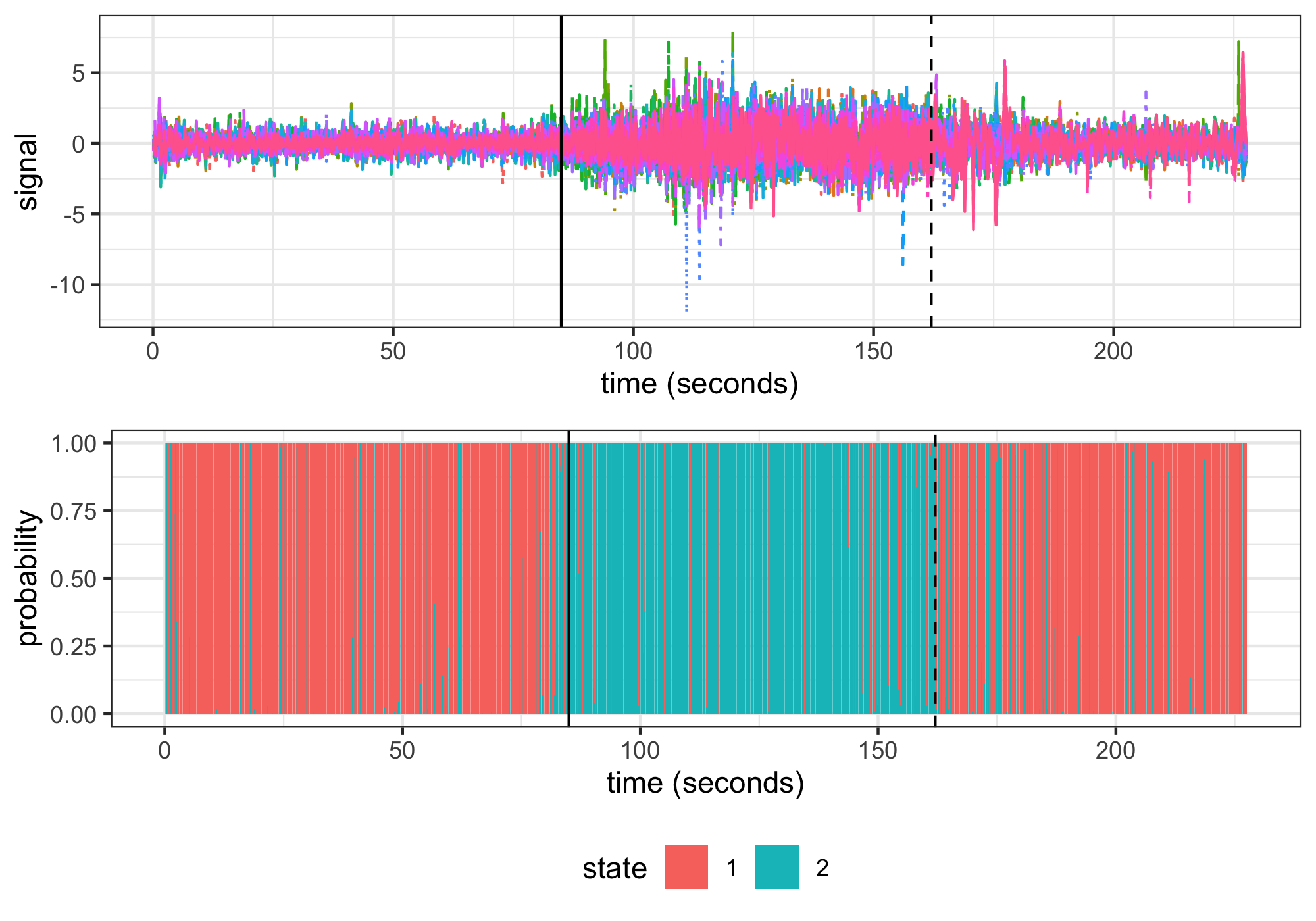}
    \caption{Top panel: standardized electroencephalograms (EEG) signals from 18 EEG channels measured from a patient experiencing an epileptic seizure over approximately 228 seconds, with different channels represented by different colors and line shapes. Bottom panel: conditional probability of each state estimated by the regularized approximate EM algorithm in a Markov-switching VAR model with 2 states. Height of colored bar represents the estimated probability of the corresponding state. Vertical solid line ($t=85$ seconds) corresponds to the time when the seizure took place, and vertical dashed line ($t=162$ seconds) corresponds to the change point estimated in \cite{safikhani2022joint}.}
    \label{fig: EEG measurement}
\end{figure}

We estimate Markov-switching VAR models with number of regimes $K \in \{2,3,4\}$ and select the best fit using a BIC criterion. We use $s=8$ in computing the approximate conditional expectations and a random initialization. The iteration is terminated if the change in the expected log likelihood function is smaller than $5\times 10^{-5}$ with a maximum of 100 iterations. {\color{change}The choice $K=2$ yields the smallest value of a high-dimensional BIC criterion \citep{wang2011consistent} and the most interpretable results. Therefore, we present the results for 2 regimes in this section, and defer the results for $K=3$ and 4 to Supplementary Appendix~\ref{app:additionalsim}.} 

The bottom panel of Figure~\ref{fig: EEG measurement} shows the approximate posterior probability of each state, $m_{j,\hat\theta}$, where $\hat\theta$ is the final EM estimate. The seizure was estimated to take place at around 85 seconds according to an expert neurologist \citep{safikhani2022joint}, which is marked by the vertical solid line in Figure~\ref{fig: EEG measurement}. We observe that starting from around that time, the process appears to switch from state 1 to state 2, as state 2 starts to consistently have (posterior) probabilities close to 1 whereas state 1 generally has higher probabilities prior to that. This corresponds well to the estimated time of the seizure occurrence by experts. Starting from $t \approx 160$ seconds, we observe oscillations between the 2 states after which the process seems to return to state 1 according to our estimates. This might indicate a recovery from the seizure, although the oscillations may also suggest increased volatility in the brain network after the seizure. The time $t \approx 160$ also agrees with the change point identified in \citet{safikhani2022joint} marked by the vertical dashed line in Figure~\ref{fig: EEG measurement}. Finally, we visualize the estimated brain connectivity networks for each regime for $K=2$ in Figure~\ref{fig: brain connectivity}. We further threshold the estimated coefficients such that approximately 10\% of entries in the coefficient matrix are non-zero in the regime before and after seizure. We observe that the brain connectivity network, as represented by the regression coefficient matrix, is much denser during the seizure period, suggesting that the brain connectivity may be highly unstable and volatile during seizure.

\begin{figure}
    \centering
    \includegraphics[width=0.6\linewidth]{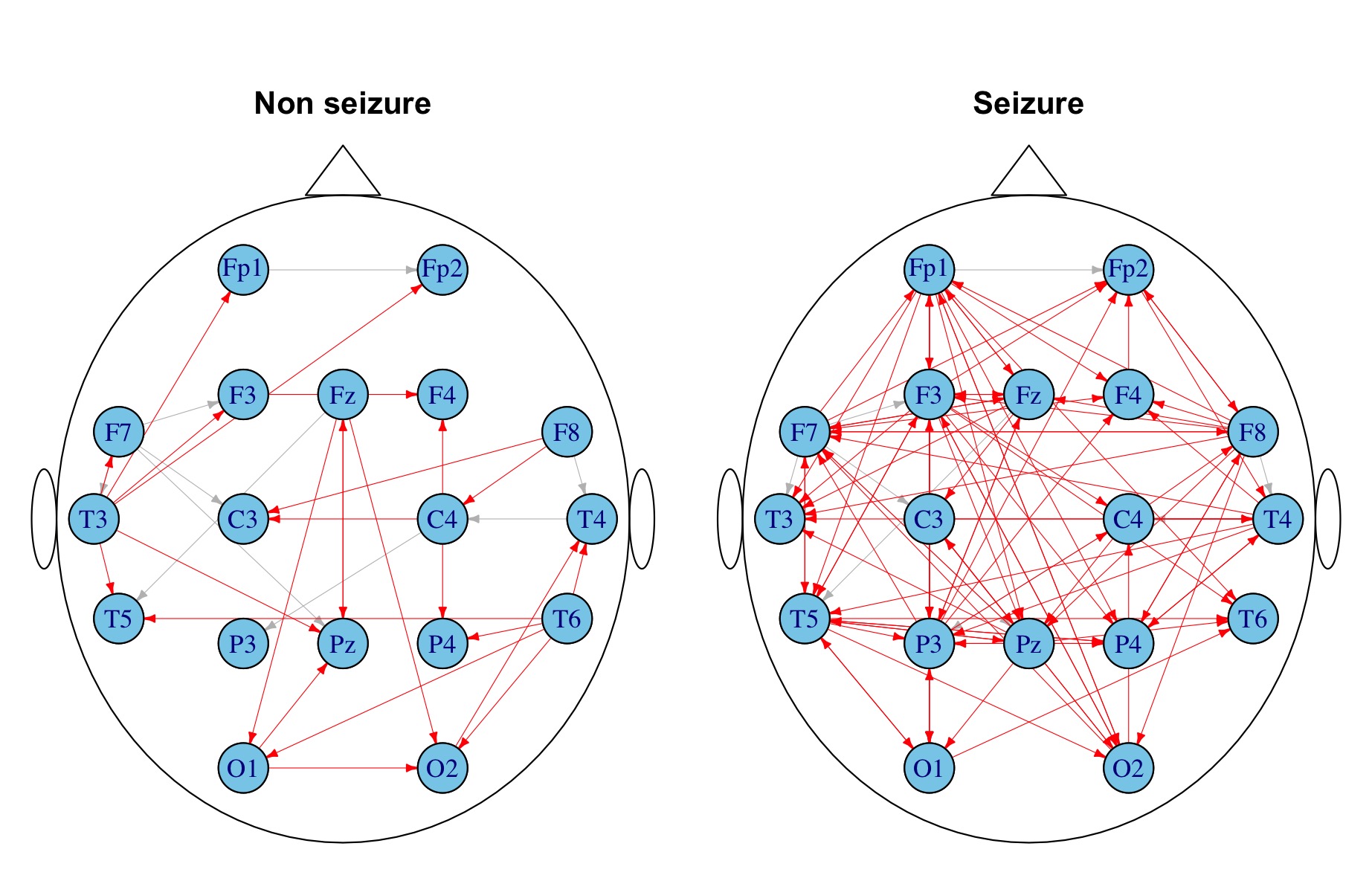}
    \caption{Brain connectivity networks in state 1 (non seizure) and state 2 (seizure). Common edges are plotted with gray color, while edges distinct to each state are plotted with red color. Regression coefficient matrices are thresholded so that approximately 10\% of entries in the adjacency matrix in state 1 is non-zero.}
    \label{fig: brain connectivity}
\end{figure}

\section{Discussion}\label{sec:discussion}
In this paper, we developed a regularized approximate EM algorithm for parameter estimation in high-dimensional Markov-switching VAR models. The proposed algorithm uses an approximation of the conditional expectation in the E-step, and allows the dimension of the observation $Y_t$ to diverge exponentially with the sample size. We also established statistical guarantees for the resulting estimate using probabilistic tools for ergodic time series.
To our knowledge, the results for convergence of EM in dependent data settings are the first of their kind. 

The proposed algorithm can be implemented efficiently as the optimization can be solved with closed-form solutions or with readily available software packages, such as \texttt{glmnet} in \texttt{R}. In our theoretical derivation (see Theorem~\ref{estimationerror}), the tuning parameter needs to be updated in each iteration. However, specifying $\lambda$ according to \eqref{lambdachoice} is challenging, as it requires knowing the true magnitude of the estimation error. In our simulations and EEG data analysis, we choose $\lambda$ based on cross-validation in each iteration, and we found that it worked well in the settings we considered. 

To the best of our knowledge, we are not aware of theoretical guarantees for the EM algorithm with arbitrary initialization. In our theory, we require that the initialization falls within a neighborhood of the true parameter value. When initial (perhaps less precise) estimates are available, the EM algorithm can be initialized using these initial estimates. We leave the development of such initial estimates in the Markov-switching VAR setting to future research. When initial estimates are not available, using multiple random initializations provides a viable solution.

In our proposed algorithm, we approximate the conditional expectations of $I\{Z_t =j\}$ and $I\{Z_{t-1}=i, Z_t =j\}$ given all observed data $Y_0^T$. These conditional expectations are referred to as \emph{smoothed probabilities} in, for example, \citet{krolzig2013markov}. An alternative is to use \emph{filtered probabilities}, defined as the conditional expectations of $I\{Z_t =j\}$ and $I\{Z_{t-1}=i, Z_t =j\}$ given only observations up to time $t$, that is, $Y_0^t$. This differs slightly from the ordinary EM algorithm.

In our consistency analysis for the smoothed-probability version, two ingredients are key: (1) the population EM mapping is a contraction with the true parameter value as a fixed point, and (2) concentration bounds control the difference between a single iteration of the sample-based EM and the population EM. It can be shown that, for the variant of the population EM based on filtered probabilities, the true parameter value also remains a fixed point. Consequently, under similar conditions ensuring sufficient signal strength and controlling both approximation and statistical errors, we expect analogous consistency results for an approximate regularized EM algorithm that uses approximations of the filtered probabilities in the E-step. In defining these approximations, one may condition only on $Y_{t-s}^t$ for a suitably chosen $s$. We examine the empirical performance in simulation studies of this variant in Supplementary Appendix~\ref{app:additionalsim} and leave its theoretical analysis for future work.

\bibliographystyle{abbrvnat}
\bibliography{references}

\newpage
\section*{Appendix}
\appendix

\section{Proof outlines}\label{sec:proof outline}
In this section, we outline the proofs of our main theoretical results in Section~\ref{sec:theoretical}, including Theorem~\ref{estimationerror} and Proposition~\ref{errororder}. We also note that the proof of Lemma~\ref{RE} follows the same general strategy as the proof of Proposition~\ref{errororder} by considering a function class indexed by sparse vector $v$ first. We then apply similar arguments as in \cite{loh2012high} to extend the result to all vectors $v$. Complete proofs of all theoretical results are presented in Supplementary Appendices~\ref{app:prooflemma} to \ref{app:proofmainthm}.

\subsection{Proof of Theorem~\ref{estimationerror}}
We prove the theorem in two main steps. In Step~I, we focus on one iteration of the EM algorithm and show that with appropriately chosen $\lambda$, the estimation error of the updated parameter estimate can be upper bounded in terms of $\lambda$. In doing so, we also derive explicit requirements on the choice of $\lambda$ to establish this upper bound. In Step~II, we choose a specific sequence of $\lambda$ values over the iterations. We use induction to show that in each iteration, our chosen $\lambda$ value satisfies the requirements in Step I, and hence our upper bound on the estimation error holds in each iteration.

\textbf{Step I: estimation error in one iteration when $\lambda$ is chosen appropriately.} We first focus on the $q$-th iteration of the EM algorithm. Let $\theta^{(q-1)}$ denote the parameter estimate prior to the $q$-th iteration, and let $\theta^{(q)}$ be the updated parameter estimate after the $q$-th iteration. For the ease of notation, in this proof, we will often write $\theta^{(q-1)}$ as $\theta = (\beta^\top, p^\top, \sigma^\top)^\top$ and $\theta^{(q)}$ as $\hat\theta = (\hat\beta^\top,\hat p^\top,\hat\sigma^\top)^\top$.

\textbf{Step I-I: estimation error of $\beta$.} For simplicity, we use $\tilde\beta$ as a shorthand notation for $M_\beta(\theta)$. Using the fact that $\hat\beta$ is the minimizer in the M-step, and rearranging the terms, we get

\begin{align*}
    &\frac{1}{T}\sum_{t=1}^{T}\sum_{j=1}^K m_{j,\theta}(Y_{t-s}^{t+s})\|(B_j^*-\hat B_j)^\top Y_{t-1}\|_2^2 \\
    &\ \ \leq \underbrace{\frac{2}{T}\sum_{t=1}^{T}\sum_{j=1}^K m_{j,\theta}(Y_{t-s}^{t+s})\left\{Y_t - (B_j^*)^\top Y_{t-1}\right\}^\top(\hat B_j -B_j^*)^\top Y_{t-1}}_{\text{term 1}} + \underbrace{\lambda \left(\|\beta^*\|_1 - \|\hat\beta\|_1\right)}_{\text{term 2}} \\
    &\quad \leq \frac{3\lambda}{2}\|(\hat\beta-\beta^*)_S\|_1-\frac{\lambda}{2}\|(\hat\beta-\beta^*)_{S^C}\|_1  \\
    &\quad \quad + (\rho_{\max}\|\tilde\beta-\beta^*\|_2 + 3\phi^s\rho_{\max}(K-1)^{1/2}\|D_\beta\|_2 )\|\hat\beta-\beta^*\|_2,
\end{align*}
if $\lambda$ is such that $\lambda\geq 2\Delta+2C_1\phi^s$. Here, to obtain the last inequality, we decompose term 1 into terms related to the statistical error, the approximation error, and the estimation error in a hypothetical iteration with infinite amount of data, and subsequently upper bound it by $\Delta\|\hat\beta-\beta^*\|_1 + \rho_{\max} \|\hat\beta-\beta^*\|_2 \|\tilde\beta-\beta^*\|_2 + 3\phi^s \rho_{\max}(K-1)^{1/2}\|\hat\beta - \beta^*\|_2\|D_\beta\|_2 +\phi^s C_1\|\hat\beta-\beta^*\|_1$.

If $\lambda$ also satisfies $2(\rho_{\max}\|\tilde\beta-\beta^*\|_2  + 3\phi^s\rho_{\max}(K-1)^{1/2}\|D_\beta\|_2) \leq \lambda\sqrt{|S|}$, we can further show that $\|\hat\beta-\beta^*\|_1 \leq 5\sqrt{|S|}\|\hat\beta-\beta^*\|_2$. This puts $\hat\beta$ into $\Theta_\beta$ if $\hat\theta$ stays in $\mathcal{B}(r, \theta^*)$ which we will show is indeed the case. Applying the restricted eigenvalue condition, for sufficiently large $T$, we have
\begin{align*}
    \frac{1}{T}\sum_{t=1}^{T}\sum_{j=1}^K m_{j,\theta}(Y_{t-s}^{t+s})\|(B_j^*-\hat B_j)^\top Y_{t-1}\|_2^2 \geq \alpha/2 \|\hat\beta-\beta^*\|_2^2,
\end{align*}
which, combined with the previous display, implies that
\begin{equation*}
    \|\hat\beta-\beta^*\|_2 \leq \frac{3\lambda\sqrt{|S|}}{\alpha} + \frac{2}{\alpha} \left(\rho_{\max}\|\tilde\beta-\beta^*\|_2 +  3\phi^s\rho_{\max}(K-1)^{1/2}\|D_\beta\|_2 \right).
\end{equation*}

\textbf{Steps I-II and I-III: estimation errors of $p$ and $\sigma$.} Leveraging the closed-form expression for the update and decomposing the difference between $\hat p$ and $p^*$ into terms related to the statistical error, the approximation error, and the estimation error in a population EM iteration, we show that
\begin{equation*}
    \|\hat p - p\|_2 \leq K \Delta_p + C_2\phi^s + \|M_p(\theta) - p^*\|_2.
\end{equation*}
Following similar steps, we can upper bound the estimation error of $\hat\sigma^2$ as follows
\begin{multline*}
    \|\hat\sigma^2 - (\sigma^*)^2\|_2 
    \leq \sqrt{K}\Delta_\sigma + \|(M_\sigma(\theta))^2 - (\sigma^*)^2\|_2 + C_3\phi^s + \\
    \frac{2\sqrt{K}\rho_{\max}}{d}\|\hat\beta - \beta^*\|_2^2 +
    \frac{2\sqrt{K}\rho_{\max}r}{d}\|\tilde\beta - \beta^*\|_2,
\end{multline*}
where we note that there are additional terms related to the estimation error of $\hat\beta$. 

\textbf{Step I-IV: estimation error of $\theta$ and requirements on $\lambda$.} Combining the estimation error for $\beta$, $p$ and $\sigma$, we can show that
\begin{multline*}
     \|\hat\theta - \theta^*\|_2 \leq \frac{3\lambda\sqrt{|S|}}{\alpha} + \frac{6}{\alpha} \phi^s\rho_{\max}(K-1)^{1/2}\|D_\beta\|_2 + K \Delta_p + C_2\phi^s + \sqrt{K} \Delta_\sigma + \\
     C_3\phi^s + \frac{2\sqrt{K}\rho_{\max}}{d}\|\hat\beta - \beta^*\|_2^2 + \eta\kappa \|\theta - \theta^*\|_2.
\end{multline*}
This estimation error can be further upper bounded as follows
\begin{equation}
    \|\hat\theta - \theta^*\|_2 \leq \left(1 + \frac{8.2\sqrt{K}\rho_{\max}\eta\kappa r}{d}\right) \frac{4\lambda\sqrt{|S|}}{\alpha}, \label{eqn:lambdarequirement}
\end{equation}
if $\lambda$ can be chosen such that
\begin{align*}
    \frac{4\lambda\sqrt{|S|}}{\alpha} &\leq 4.1\eta\kappa r, \quad \quad \lambda \geq 2\Delta + 2C_1\phi^s, \\
    \frac{\lambda \sqrt{|S|}}{\alpha} &\geq \frac{6}{\alpha} \phi^s\rho_{\max}(K-1)^{1/2}\|D_\beta\|_2 + K \Delta_p + C_2\phi^s + \sqrt{K} \Delta_\sigma + C_3\phi^s  + \eta\kappa \|\theta - \theta^*\|_2.
\end{align*}
These are stronger than the requirements on $\lambda$ used to establish the upper bound on the estimation error of $\beta$. We thus need to establish that our choice of $\lambda$ indeed satisfies these requirements. 

\textbf{Step II: induction.} 
We use induction to show that,  in each iteration, the requirements for $\lambda$ are satisfied by our choice. We can verify $\lambda^{(1)}$ indeed satisfies all the requirements above. Next, we show that if $\lambda^{(q)}$ satisfies all the requirements so that \eqref{eqn:lambdarequirement} holds,
then $\lambda^{(q+1)}$ also satisfies all the requirements so that the upper bound holds at iteration $q+1$. This can be verified by noting that: (a)
\begin{equation*}
    \lambda^{(q+1)} \geq \frac{1-\tau^{(q+1)}}{1-\tau} (2\Delta + 2C_1\phi^s) \geq 2\Delta + 2C_1\phi^s;
\end{equation*}
(b) for sufficiently large $T$, 
\begin{multline*}
    \frac{4\sqrt{|S|}}{\alpha}\frac{1}{1-\tau} \max\Bigg\{2\Delta + 2C_1\phi^s,  \\
    \frac{\alpha}{\sqrt{|S|}}\left(\frac{6}{\alpha} \phi^s\rho_{\max}(K-1)^{1/2}\|D_\beta\|_2 + K \Delta_p + C_2\phi^s + \sqrt{K}\Delta_\sigma + C_3\phi^s  \right)\Bigg\} \leq 0.1\eta\kappa r;
\end{multline*}
and (c) by leveraging the explicit expressions for $\lambda^{(q)}$ and $\lambda^{(q+1)}$, we get 
\begin{align*}
    &\quad \frac{\alpha}{\sqrt{|S|}}\left\{\frac{6}{\alpha} \phi^s\rho_{\max}(K-1)^{1/2}\|D_\beta\|_2 + K \Delta_p + C_2\phi^s + \sqrt{K} \Delta_\sigma + C_3\phi^s  + \eta\kappa \|\theta^{(q)} - \theta^*\|_2\right\} \\
    &\leq \frac{\alpha}{\sqrt{|S|}}\left\{\frac{6}{\alpha} \phi^s\rho_{\max}(K-1)^{1/2}\|D_\beta\|_2 + K \Delta_p + C_2\phi^s + \sqrt{K} \Delta_\sigma + C_3\phi^s  \right\}  + \tau \lambda^{(q)} \leq \lambda^{(q+1)}.
\end{align*}
Finally, by plugging in the expression of $\lambda^{(q)}$ into the upper bound on the estimation error, we can show that $\|\theta^{(q)} - \theta^*\|_2 \leq  \|\theta^{(0)} - \theta^*\|_2$ for sufficiently large $T$ and thus $\theta^{(q)}$ always remains in $\mathcal{B}(r,\theta^*)$. This concludes the proof. 

\subsection{Proof of Proposition~\ref{errororder}}
For the ease of notation, define the following functions:
\begin{align*}
    h^{ijk}(Y_{t-1}^t) &= Y_{t-1,k}(Y_{ti}-\beta_{ji}^{*\top}Y_{t-1}), \\
    f_\theta^{ijk}(Y_{t-s}^{t+s}) &= Y_{t-1,k}(Y_{ti}-\beta_{ji}^{*\top}Y_{t-1})m_{j,\theta}(Y_{t-s}^{t+s}) = h^{ijk}(Y_{t-1}^t)m_{j,\theta}(Y_{t-s}^{t+s}),
\end{align*}
and the function class $\mathcal{G}$ as $\mathcal{G} = \bigcup_{i,j,k} \left\{f_{\theta}^{ijk}: \theta \in \Theta \right\}.$ 
We will show that as $T\rightarrow \infty$,
\begin{equation*}
    P\left( \sup_{g \in \mathcal{G}} \left|\frac{1}{T}\sum_{t=1}^T g(Y_{t-s}^{t+s}) - E\left[g(Y_{t-s}^{t+s})\right]\right| \geq \delta \right) \rightarrow 0,
\end{equation*}
with $\delta = C^* \sqrt{(|S|l(T,d)(\log T)^3 (\log K + \log d) + (\log T)^4)T^{-1}}$. We do this in three steps. In step~I, we control the tail behavior of the random variable $h^{ijk}(Y_{t-1}^t)$, which will be useful for the concentration result later. In step~II, we establish a uniform concentration result over the function class $\mathcal{G}$ using an entropy argument, which holds for independent and identically distributed observations. In step~III, we generalize the concentration result from step~II by using $\beta$-mixing to establish  
a uniform concentration result for the original time series.

\textbf{Step I: control the tail of $h^{ijk}$.} Assumption~\ref{cond: operatornorm} implies that $\|\beta_{ji}\|_2 \leq 1$ for all $i,j$. Lemma~\ref{lemma: msvarstationarity} implies that $E[|v^\top Y_t|^q] \leq K_Y^q q^{q/2}$ for all unit vectors $v$. Using this, we can show that for all $q \geq 1$, $$(E|h^{ijk}(Y_{t-1}^t)|^q)^{1/q} \leq 4K_Y^2q, \quad \forall i,j,k$$
As $m$ is bounded, the $L_2$-norm of $g \in \mathcal{G}$ is upper bounded by $8K_Y^2$.

\textbf{Step II: uniform concentration for i.i.d. data.} Let $N = T/\{\tilde{c}\log T\}$ for some fixed constant $\tilde{c}$. Let $\{\tilde{Y}_{n-s}^{n+s}\}_{n=1}^N$ be an i.i.d. sample wherein the marginal distribution of $\tilde{Y}_{n-s}^{n+s}$ is the same as the marginal distribution of $Y_{t-s}^{t+s}$, and let $X_n = \tilde{Y}_{n-s}^{n+s}$. 

\noindent\textbf{Symmetrization.} Applying the symmetrization argument in Theorem~\ref{symmetrization}, we have
\begin{equation*}
    P\left(\sup_{g\in\mathcal{G}}\left|\frac{1}{N}\sum_{n=1}^N g(X_n) - E[g(X_n)]\right| \geq \delta \right) \leq 4P\left( \sup_{g\in\mathcal{G}}\left|\frac{1}{N}\sum_{n=1}^N W_ng(X_n)\right| \geq \delta/4 \right),
\end{equation*}
for i.i.d. Rademacher random variables $W_n$'s.

\noindent\textbf{Empirical norm of $g \in \mathcal{G}$.} Define event $\mathcal{A}$ such that $\mathcal{A}$ occurs if and only if $$\sup_{g \in \mathcal{G}}N^{-1}\sum_{n=1}^N g(X_n)^2 \leq 64K_Y^4 + 1.$$ Applying the concentration result for sub-Weibull random variables as in Lemma 13 of \citet{wong2020lasso} and a union bound, we get a uniform upper bound on the empirical norm of functions in $\mathcal{G}$ that holds with high probability,
\begin{equation*}
    P\left(\sup_{g \in \mathcal{G}}\frac{1}{N}\sum_{n=1}^N g(X_n)^2 > 64K_Y^4+ 1\right) \leq Kd^2 N\exp\left\{-\frac{N^{1/2}}{8K_Y^2\tilde{C}_1}\right\} + Kd^2 \exp\left\{-\frac{N}{\tilde{C}_2 (64K_Y^4)^2}\right\}.
\end{equation*}

\noindent\textbf{Conditioning on $X_1,\ldots,X_N$.} Applying a Sudakov minoration argument, we get $\sqrt{\log N_c(\epsilon,\Theta,\|\cdot\|_2)} \leq C_1r \sqrt{|S|(\log K + \log d)}/\epsilon$ for some constant $C_1$. We can further show that functions $g \in \mathcal{G}$ is Lipschitz in $\theta$ with respect to the empirical norm, with Lipschitz constant $L_G(X_1^N) \geq 1$, for given $i,j,k$. Therefore, taking a union over $i,j,k$, we have
\begin{equation*}
    \sqrt{\log N_c(\epsilon,\mathcal{G},\|\cdot\|_{Q_n})}  \leq C_2 r L_G(x_1^N) \sqrt{|S|(\log K + \log d)}/\epsilon.
\end{equation*}
Under Assumption~\ref{randomentropy}(b), the Lipschitz constant is upper bounded with high probability:
\begin{multline*}
    P\left(L_G^2(X_1^N) > l(N,d)(16^2K_Y^4 + 1) \right) \leq
    Kd^2 \exp\left(-\frac{N}{K_4^2\tilde C_2}\right) + Kd^2 N \exp\left(-\frac{N^{1/4}}{K_4^{1/4}\tilde C_1}\right) \\
    + \sum_{j=1}^K \tilde{u}_j(N,d),
\end{multline*}
where $K_4 = 2^4(64K_Y^4)^2$. Similar upper bound can be established under Assumption~\ref{randomentropy}(a). This will provide a high-probability upper bound on the random entropy. Now, consider a set of values $\{x_1,\ldots,x_N\}$ such that all the high-probability events defined so far occur, for sufficiently large $N$ (sufficiently large $T$), we will apply Theorem~\ref{corollary83}. In applying this result, we take $R^2 = 64K_Y^4 + 1$ and $\delta_1 = C_3 \sqrt{(|S|L_G^2(x_1^N)(\log K + \log d)(\log T)^3 + (\log T)^4)/T}$. We need to verify that $R > \delta_1$, which follows from the fact that $\delta_1 = o(1)$ if $T\gg |S|\log d$, and verify an upper bound on the entropy integral below:
\begin{multline*}
    2C \int_{\delta_1/8}^R \sqrt{\log N_c(\epsilon,\mathcal{G},\|\cdot\|_{Q_n})} d\epsilon \leq \tilde{C}_2 r L_G(x_1^N) \sqrt{|S|(\log K + \log d)} \log T \\
    \leq C_3 \sqrt{|S|L_G^2(x_1^N)(\log K + \log d)(\log T)^2 + (\log T)^3} = \sqrt{N}\delta_1,
\end{multline*}
Applying Theorem~\ref{corollary83} with the specified value of $\delta_1$, we have that
\begin{equation*}
    P\left(\sup_{g \in \mathcal{G}} \left|\frac{1}{N}\sum_{n=1}^N W_n g(x_n)\right| \geq \delta_1 \right) \leq C\exp\left\{-\frac{N\delta_1^2}{4C^2(64K_Y^4 + 1)}\right\}.
\end{equation*}

\noindent\textbf{Marginalize over $X_1,\ldots,X_N$.} With appropriately chosen $C^*$ and $C_3$, $\delta > 4\delta_1$ when the Lipschitz constant $L_G^2(X_1^N)$ is upper bounded. Therefore, marginalizing over $X_1,\ldots,X_N$, we get 

\begin{align*}
    P\left(\sup_{g\in\mathcal{G}}\left|\frac{1}{N}\sum_{n=1}^N g(X_n) - E[g(X_n)]\right| \geq \delta \right) \leq 4P\left( \sup_{g\in\mathcal{G}}\left|\frac{1}{N}\sum_{n=1}^N W_ng(X_n)\right| \geq \delta/4 \right) \leq u_{err}(N,d),
\end{align*}
where we use $u_{err}(N,d)$ to denote the sum of all the small probabilities that have appeared so far.

\textbf{Step III: uniform concentration for $\beta$-mixing process.} Under Assumption~\ref{betamixing}, applying Theorem~\ref{karandikarthm}, we get
\begin{equation*}
    P\left( \sup_{g \in \mathcal{G}} \left|\frac{1}{T}\sum_{t=1}^T g(Y_{t-s}^{t+s}) - E\left[g(Y_{t-s}^{t+s})\right]\right| \geq \delta \right) \leq C_4 (\log T) u_{err}(T/(\tilde{c}\log T),d) + 2T^{-1/2},
\end{equation*}
which converges to 0. This concludes the proof. Again, the detailed proof is given in Supplementary Appendix~\ref{app:proofmainthm}.

\section{Proof of lemmas in Section~\ref{sec:theoretical}}\label{app:prooflemma}
In this section, we prove lemmas in Section~\ref{sec:theoretical} of the main paper.

\begin{proof}[Proof of Lemma~\ref{lemma: msvarstationarity}] Under Assumption~\ref{cond: operatornorm}, stationarity follows by directly applying Theorem 3.1 and Corollary 3.1 in \citet{stelzer2009markov}, and geometric ergodicity follows by applying Theorem 5.1 and Proposition 5.3 in \citet{stelzer2009markov}. 

It remains to show that $Y_t$ is a sub-Gaussian random vector. We start by noting that Theorem 4.2 in \citet{stelzer2009markov} implies that all moments of $Y_t$ exist. Iteratively applying \eqref{msvar}, we have that
\begin{align*}
    Y_t &= \left(\sum_{i=1}^{K} I\{Z_t = i\} B_i^\top\right) Y_{t-1} + \left(\sum_{i=1}^{K} I\{Z_t = i\}\sigma_i\right)\epsilon_t \\
        &= \left(\sum_{i=1}^{K} I\{Z_t = i\} B_i^\top\right)\left(\sum_{i=1}^{K} I\{Z_{t-1} = i\} B_i^\top\right) Y_{t-2} \\
        &\quad + \left(\sum_{i=1}^{K} I\{Z_t = i\} B_i^\top\right) \left(\sum_{i=1}^{K} I\{Z_{t-1} = i\} \sigma_i\right) \epsilon_{t-1} + 
        \left(\sum_{i=1}^{K} I\{Z_t = i\}\sigma_i\right)\epsilon_t \\
        &= \left(\sum_{i=1}^{K} I\{Z_t = i\} B_i^\top\right) \left(\sum_{i=1}^{K} I\{Z_{t-1} = i\} B_i^\top\right) \left(\sum_{i=1}^{K} I\{Z_{t-2} = i\} B_i^\top\right)Y_{t-3} \\
        &\quad + \left(\sum_{i=1}^{K} I\{Z_t = i\} B_i^\top\right) \left(\sum_{i=1}^{K} I\{Z_{t-1} = i\} B_i^\top\right) \left(\sum_{i=1}^{K} I\{Z_{t-2} = i\}\sigma_i\right) \epsilon_{t-2} \\
        &\quad + \left(\sum_{i=1}^{K} I\{Z_t = i\} B_i^\top\right) \left(\sum_{i=1}^{K} I\{Z_{t-1} = i\} \sigma_i\right) \epsilon_{t-1} + 
        \left(\sum_{i=1}^{K} I\{Z_t = i\}\sigma_i\right)\epsilon_t \\
        &= \cdots \\
        &= \left[\prod_{k=0}^J\left(\sum_{i=1}^{K} I\{Z_{t-k} = i\} B_i^\top\right)\right]Y_{t-J-1} \\
        &\quad + \sum_{l=1}^J \left[ \prod_{k=0}^{l-1} \left(\sum_{i=1}^{K} I\{Z_{t-k} = i\} B_i^\top\right) \left(\sum_{i=1}^{K} I\{Z_{t-l} = i\}\sigma_i\right) \epsilon_{t-l}\right] + \left(\sum_{i=1}^{K} I\{Z_t = i\}\sigma_i\right) \epsilon_t,
\end{align*}
for any positive integer $J$. For a generic time point $t$, we define the matrix $A_t = \sum_{i=1}^K I\{Z_t =i\}B_i^\top$. Although the matrix $A_t$ is random due to the randomness in $Z_t$, Assumption~\ref{cond: operatornorm} implies that $\|A_t\|_2 \leq \tilde{c}$ with probability 1. Similarly, define a random variable $\Gamma_t = \sum_{i=1}^{K} I\{Z_t = i\}\sigma_i$ whose randomness arises due to the randomness in $Z_t$.  With the definition of $A_t$ and $\Gamma_t$, the above display can be written equivalently as
\begin{equation*}
    Y_t = \left(\prod_{k=0}^J A_{t-k}\right)Y_{t-J-1} + \sum_{l=1}^J\left\{\left(\prod_{k=0}^{l-1}A_{t-k}\right)\Gamma_{t-l}\epsilon_{t-l}\right\} + \Gamma_t\epsilon_t. 
\end{equation*}
In fact, we can continue expanding $Y_t$, and Theorem 4.2 in \citet{stelzer2009markov} implies that the stationary distribution of $Y_t$ admits the following representation
\begin{equation*}
    Y_t = \sum_{l=1}^\infty \left\{\left(\prod_{k=0}^{l-1}A_{t-k}\right)\Gamma_{t-l}\epsilon_{t-l}\right\} + \Gamma_t\epsilon_t,
\end{equation*}
where the series on the right-hand side in the above display converges in the norm $\|\cdot\|_{L^r}$ for any $r \geq 1$, with the norm defined as $(E\|X\|_2^r)^{1/r}$ for a random vector $X$.

Since $\epsilon_t$ follows a Gaussian distribution, there exists a constant $K_1>0$ such that for all $v \in \mathbb{R}^{d}$ with $\|v\|_2 = 1$,
\begin{equation}\label{errortail}
    \left(E\left|v^\top\epsilon_t\right|^p\right)^{1/p} \leq K_1 p^{1/2}.
\end{equation}
Fix an arbitrary unit-vector $v \in \mathbb{R}^d$. For the ease of notation, we introduce a truncated version of the series representation of $Y_t$ defined as $Y_{t}^J = \sum_{l=1}^J \{(\prod_{k=0}^{l-1}A_{t-k})\Gamma_{t-l}\epsilon_{t-l}\} + \Gamma_t\epsilon_t$, for a positive integer $J$. By Minkowski inequality,
\begin{align*}
    \left(E\left|v^\top Y_t\right|^p\right)^{1/p} &\leq \underbrace{\left(E\left|v^\top \left(Y_t- Y_t^J\right)\right|^p\right)^{1/p}}_{\textnormal{term a}} + \underbrace{\sum_{l=1}^J \left(E\left|v^\top \left(\prod_{k=0}^{l-1}A_{t-k}\right) \Gamma_{t-l}\epsilon_{t-l}\right|^p\right)^{1/p}}_{\textnormal{term b}}\\
    &\quad + \underbrace{\left(E\left|v^\top \Gamma_t\epsilon_t\right|^p\right)^{1/p}}_{\textnormal{term c}}.
\end{align*}
We study each term in the above display separately. Let $\sigma_{\max} = \max(\sigma_1,\ldots,\sigma_K)$. Then, $\Gamma_t \leq \sigma_{\max}$ with probability 1. First, term c is upper bounded by $\sigma_{\max} K_1 p^{1/2}$ as $\epsilon_t$ is a Gaussian random vector and $\Gamma_t$ is upper bounded. For term a, we note that
\begin{equation*}
    \textnormal{term a} \leq \left(E\left[\|v\|_2^p \left\|Y_t- Y_t^J\right\|_2^p\right]\right)^{1/p} = \left(E\left\|Y_t- Y_t^J\right\|_2^p\right)^{1/p} = \|Y_t - Y_t^J\|_{L^p}.
\end{equation*}
We now study term b. To start, we note that $\epsilon_{t-l}$ is independent of $\{Z_{t-l},Z_{t-l+1},\ldots,Z_t\}$, and hence independent of the random matrices $A_{t-l+1},\ldots,A_t$ as well as $\Gamma_{t-l}$. Now define a random vector $U_{t,l} = (\prod_{k=0}^{l-1}A_{t-k})^\top v\Gamma_{t-l}$. Then, each term in the sum in term b is equivalent to $(E|U_{t,l}^\top \epsilon_{t-l}|^p)^{1/p}$. Note that
\begin{align*}
    E\left|U_{t,l}^\top \epsilon_{t-l}\right|^p &= E\left[\left|\epsilon_{t-l}^\top\frac{U_{t,l}}{\|U_{t,l}\|_2}\right|^p\|U_{t,l}\|_2^p\right] \\
    &= E\left[E\left[\left|\epsilon_{t-l}^\top\frac{U_{t,l}}{\|U_{t,l}\|_2}\right|^p \mid Z_{t-l}, Z_{t-l+1},\ldots,Z_t\right]\|U_{t,l}\|_2^p\right].
\end{align*}
Here, condition on $\{Z_{t-l}, Z_{t-l+1},\ldots,Z_t\}$, the vector $U_{t,l}$ becomes deterministic, but the distribution of $\epsilon_{t-l}$ is unchanged due to the independence. Thus, 
\begin{equation*}
    E\left[\left|\epsilon_{t-l}^\top\frac{U_{t,l}}{\|U_{t,l}\|_2}\right|^p \mid Z_{t-l}, Z_{t-l+1},\ldots,Z_t\right] \leq \left(K_1 p^{1/2}\right)^p,
\end{equation*}
and consequently 
\begin{equation*}
    E\left|U_{t,l}^\top \epsilon_{t-l}\right|^p \leq \left(K_1 p^{1/2}\right)^p E\left[\|U_{t,l}\|_2^p\right].
\end{equation*}
The norm $\|U_{t,l}\|_2$ is upper bounded by $\Gamma_{t-l}\prod_{k=0}^{l-1}\|A_{t-k}\|_2$, which is upper bounded by $\sigma_{\max}\tilde{c}^l$ with probability 1. Thus, $E\left[\|U_{t,l}\|_2^p\right] \leq \sigma_{\max}^p \tilde{c}^{lp}$. Combining these results, we get that 
\begin{equation*}
    \left(E\left|v^\top \left(\prod_{k=0}^{l-1}A_{t-k}\right)\Gamma_{t-l}\epsilon_{t-l}\right|^p\right)^{1/p} \leq \sigma_{\max} \tilde{c}^l K_1 p^{1/2}.
\end{equation*}
Putting the upper bounds for terms a, b and c together, we have
\begin{align*}
    \left(E\left|v^\top Y_t\right|^p\right)^{1/p} &\leq \|Y_t - Y_t^J\|_{L^p} + \sum_{l=1}^J\sigma_{\max}\tilde{c}^l K_1 p^{1/2} + \sigma_{\max}K_1p^{1/2} \\
    &=\|Y_t - Y_t^J\|_{L^p} + \frac{1-\tilde{c}^{J+1}}{1-\tilde{c}} \sigma_{\max}K_1p^{1/2}.
\end{align*}
The above display holds for any positive integer $J$, and hence we can take the limit as $J$ approaches infinity. By Theorem 4.2 in \citet{stelzer2009markov}, $\|Y_t - Y_t^J\|_{L^p}$ converges to 0 and thus,
\begin{equation*}
    \left(E\left|v^\top Y_t\right|^p\right)^{1/p} \leq \frac{1}{1-\tilde{c}} \sigma_{\max} K_1p^{1/2},
\end{equation*}
which implies that $Y_t$ is a sub-Gaussian random vector.

\end{proof}

\begin{proof}[Proof of Lemma~\ref{populationcontraction}]
This lemma follows directly from the mean-value inequality and the fact that $M(\theta^*) = \theta^*$.
\end{proof}

\begin{proof}[Proof of Lemma~\ref{lemma:approximationerror}]
We study the approximation error of the smoothed probabilities $P(Z_t = j | Y_0,\ldots, Y_T)$ and $P(Z_{t-1}=i, Z_t = j | Y_0,\ldots, Y_T)$ in two steps. In Step I, we consider intermediate approximations that truncate the dependence on distant future and only condition on $Y_0^{t+s}$. In Step II, we further remove the dependence on the distant past by conditioning on $Z_{t-s}$. 

\textbf{Step I.} With the Markov structure, we see that
\begin{equation*}
    P(Z_t = j | Y_0,\ldots,Y_T, Z_{t+s} = i) = P(Z_t = j | Y_0,\ldots,Y_{t+s}, Z_{t+s} = i),
\end{equation*}
and that conditioning on $Z_{t+s}$ removes the dependence of $Z_t$ on future $Y$. Also, we have the following equalities
\begin{align*}
    &\quad P(Z_t = j | Y_0,\ldots,Y_T, Z_{t+s} = i) - P(Z_t = j | Y_0,\ldots,Y_T) \\
    &= P(Z_t = j | Y_0,\ldots,Y_T, Z_{t+s} = i) \\
    &\quad - \sum_{1\leq l \leq K} P(Z_t = j| Y_0,\ldots,Y_T, Z_{t+s} = l)P(Z_{t+s} = l | Y_0,\ldots,Y_T) \\
    &= \sum_{1\leq l \leq K}P(Z_t = j | Y_0,\ldots,Y_T, Z_{t+s} = i)P(Z_{t+s} = l | Y_0,\ldots,Y_T) \\
    &\quad - \sum_{1\leq l \leq K} P(Z_t = j| Y_0,\ldots,Y_T, Z_{t+s} = l)P(Z_{t+s} = l | Y_0,\ldots,Y_T) \\
    &= \sum_{1\leq l \leq K}\left\{P(Z_t = j | Y_0,\ldots,Y_{t+s}, Z_{t+s} = i) - P(Z_t = j | Y_0,\ldots,Y_{t+s}, Z_{t+s} = l)\right\} \\
    &\quad \quad \quad \quad \quad \times P(Z_{t+s} = l | Y_0,\ldots,Y_T),
\end{align*}
and similarly
\begin{align*}
    &\quad P(Z_t = j | Y_0,\ldots,Y_{t+s}, Z_{t+s} = i) - P(Z_t = j | Y_0,\ldots,Y_{t+s}) \\
    &= P(Z_t = j | Y_0,\ldots,Y_{t+s}, Z_{t+s} = i) \\
    &\quad - \sum_{1\leq l \leq K} P(Z_t = j| Y_0,\ldots,Y_{t+s}, Z_{t+s} = l)P(Z_{t+s} = l | Y_0,\ldots,Y_{t+s}) \\
    &= \sum_{1\leq l \leq K}P(Z_t = j | Y_0,\ldots,Y_{t+s}, Z_{t+s} = i)P(Z_{t+s} = l | Y_0,\ldots,Y_{t+s}) \\
    &\quad - \sum_{1\leq l \leq K} P(Z_t = j| Y_0,\ldots,Y_{t+s}, Z_{t+s} = l)P(Z_{t+s} = l | Y_0,\ldots,Y_{t+s}) \\
    &= \sum_{1\leq l \leq K}\left\{P(Z_t = j | Y_0,\ldots,Y_{t+s}, Z_{t+s} = i) - P(Z_t = j | Y_0,\ldots,Y_{t+s}, Z_{t+s} = l)\right\} \\
    &\quad \quad \quad \quad \quad \times P(Z_{t+s} = l | Y_0,\ldots,Y_{t+s}).
\end{align*}
Hence, it suffices to first upper bound the following quantity
\begin{equation*}
    \max_{i,l}\left|P(Z_t = j | Y_0,\ldots,Y_{t+s}, Z_{t+s} = i) - P(Z_t = j | Y_0,\ldots,Y_{t+s}, Z_{t+s} = l)\right|.
\end{equation*}
Now for generic $t_1 < t_2$, define a stochastic matrix $P^{t_1,t_2} \in \mathbb{R}^{K\times K}$ such that
\begin{equation*}
    (P^{t_1,t_2})_{ij} = P(Z_{t_1} = j | Y_0,\ldots,Y_{t_2}, Z_{t_2} = i).
\end{equation*}
Note that this is indeed a stochastic matrix as the sum of entries over $j$ (columns) for fixed $i$ (row) is 1. By the law of total expectation
\begin{align*}
    &\quad P(Z_{t_1} = j | Y_0,\ldots,Y_{t_2}, Z_{t_2} = i) \\
    &= \sum_l P(Z_{t_1} = j, Z_{t_1+1} =l | Y_0,\ldots,Y_{t_2}, Z_{t_2} = i) \\
    &= \sum_l P(Z_{t_1} = j | Z_{t_1+1} =l, Y_0,\ldots,Y_{t_2}, Z_{t_2} = i) P(Z_{t_1+1} =l | Y_0,\ldots,Y_{t_2}, Z_{t_2} = i) \\
    &= \sum_l P(Z_{t_1} = j | Z_{t_1+1} =l, Y_0,\ldots,Y_{t_1+1}) P(Z_{t_1+1} =l | Y_0,\ldots,Y_{t_2}, Z_{t_2} = i) \\
    &= \sum_l (P^{t_1,t_1+1})_{lj} (P^{t_1+1,t_2})_{il},
\end{align*}
and hence $P^{t_1,t_2} = P^{t_1+1,t_2}P^{t_1,t_1+1}$, and
\begin{equation*}
    P^{t_1,t_2} = P^{t_2-1,t_2}P^{t_2-2,t_2-1} \cdots P^{t_1+1,t_1+2}P^{t_1,t_1+1}.
\end{equation*}
Let $\psi(P^{t_1,t_2}) = \max_{1\leq i,l \leq K} \max_{1\leq j \leq K} |(P^{t_1,t_2})_{ij} - (P^{t_1,t_2})_{lj}|$, and the quantity we aim to bound is $\psi(P^{t,t+s})$. Let $\zeta(P^{t,t+1}) = 1/2 \max_{1\leq i,l \leq K} \sum_{j} |(P^{t,t+1})_{ij}-(P^{t,t+1})_{lj}|$. By Lemma~\ref{hajnallemma}, if $\zeta(P^{t,t+1}) \leq \phi$, $\psi(P^{t,t+s}) \leq \phi^s$.

{\small
\begin{align*}
    &\quad (P^{t,t+1})_{ij} \\
    &= P(Z_t = j | Y_0,\ldots,Y_{t+1},Z_{t+1}=i) \\
    &= \frac{P(Y_0,\ldots,Y_{t-1},Z_t = j, Y_t,Z_{t+1}=i, Y_{t+1})}{P(Y_0,\ldots,Y_t,Z_{t+1}=i, Y_{t+1})} \\
    &= \frac{\sum_{z_0,\ldots,z_{t-1}}P(Y_0,z_0)P_{z_0z_1}p_{z_1}(Y_1|Y_0)P_{z_1z_2}p_{z_2}(Y_2|Y_1)\cdots P_{z_{t-2}z_{t-1}}p_{z_{t-1}}(Y_{t-1}|Y_{t-2})P_{z_{t-1}j}p_j(Y_t|Y_{t-1})P_{ji}p_i(Y_{t+1}|Y_t)}{\sum_l \sum_{z_0,\ldots,z_{t-1}}P(Y_0,z_0)P_{z_0z_1}p_{z_1}(Y_1|Y_0)P_{z_1z_2}p_{z_2}(Y_2|Y_1)\cdots P_{z_{t-2}z_{t-1}}p_{z_{t-1}}(Y_{t-1}|Y_{t-2})P_{z_{t-1}l}p_l(Y_t|Y_{t-1})P_{li}p_i(Y_{t+1}|Y_t)}.
\end{align*}}
where $p_j(\cdot|\cdot)$ denotes the conditional density function of $Y_{t+1}|(Y_t,Z_t = j)$, for $j \in \{1,\ldots,K\}$. Let $\pi = (z_0,\ldots,z_{t-1})$ denote a path of $z$ from time 0 to $t-1$ and $\Pi$ denote all possible such paths. For each path, define
\begin{equation*}
    q_{j\pi} = P(Y_0,z_0)P_{z_0z_1}p_{z_1}(Y_1|Y_0)P_{z_1z_2}p_{z_2}(Y_2|Y_1)\cdots P_{z_{t-2}z_{t-1}}p_{z_{t-1}}(Y_{t-1}|Y_{t-2})P_{z_{t-1}j}p_j(Y_t|Y_{t-1}),
\end{equation*}
then we have
\begin{equation*}
    (P^{t,t+1})_{ij} = \frac{\sum_{\pi \in \Pi} q_{j\pi} P_{ji} }{\sum_{\pi \in \Pi}\sum_l q_{l\pi} P_{li} },
\end{equation*}
and therefore
\begin{align*}
    \left|(P^{t,t+1})_{ij} - (P^{t,t+1})_{kj}\right| &= \left| \frac{\sum_{\pi \in \Pi} q_{j\pi} P_{ji} }{\sum_{\pi \in \Pi}\sum_l q_{l\pi} P_{li} } - \frac{\sum_{\tilde\pi \in \Pi} q_{j\tilde\pi} P_{jk} }{\sum_{\tilde\pi \in \Pi}\sum_{\tilde l} q_{\tilde{l}\tilde\pi} P_{\tilde{l}k} }\right| \\
    &= \left|\frac{\sum_{\pi \in \Pi}\sum_{\tilde\pi \in \Pi}\left(q_{j\pi}P_{ji}\left\{\sum_{\tilde l} q_{\tilde{l}\tilde\pi} P_{\tilde{l}k} \right\} - q_{j\tilde\pi} P_{jk}\left\{\sum_l q_{l\pi} P_{li} \right\}\right)}{\sum_{\pi \in \Pi}\sum_{\tilde\pi \in \Pi} \left\{\sum_l q_{l\pi} P_{li} \sum_{\tilde l} q_{\tilde{l}\tilde\pi} P_{\tilde{l}k} \right\}} \right| \\
    &= \left|\frac{\sum_{\pi \in \Pi}\sum_{\tilde\pi \in \Pi} \sum_{l \neq j} \left(  q_{j\pi}P_{ji} q_{l\tilde\pi} P_{lk}  - q_{j\tilde\pi} P_{jk} q_{l\pi} P_{li}\right)}{\sum_{\pi \in \Pi}\sum_{\tilde\pi \in \Pi} \left\{\sum_l q_{l\pi} P_{li} \sum_{\tilde l} q_{\tilde{l}\tilde\pi} P_{\tilde{l}k} \right\}} \right| \\
    &\leq \left|\frac{\sum_{\pi \in \Pi}\sum_{\tilde\pi \in \Pi} \sum_{l \neq j} \left(  q_{j\pi}P_{ji} q_{l\tilde\pi} P_{lk}  - q_{j\tilde\pi} P_{jk} q_{l\pi} P_{li}\right)}{\sum_{\pi \in \Pi}\sum_{\tilde\pi \in \Pi} \sum_{l \neq j} \left\{ q_{j\pi} P_{ji}  q_{l\tilde\pi} P_{lk} + q_{l\pi} P_{li} q_{j\tilde\pi} P_{jk} \right\}} \right| \\
    &= \left|\frac{\sum_{\pi \in \Pi}\sum_{\tilde\pi \in \Pi} \sum_{l \neq j} q_{j\pi}q_{l\tilde\pi} \left(  P_{ji}  P_{lk}  -  P_{jk} P_{li}\right)}{\sum_{\pi \in \Pi}\sum_{\tilde\pi \in \Pi} \sum_{l \neq j} q_{j\pi}q_{l\tilde\pi}\left\{ P_{ji}   P_{lk} + P_{li}  P_{jk} \right\}} \right| \\
    &\leq \left|\frac{\sum_{\pi \in \Pi}\sum_{\tilde\pi \in \Pi} \sum_{l \neq j} q_{j\pi}q_{l\tilde\pi} \left(  P_{ji}  P_{lk}  +  P_{jk} P_{li}\right) \max_{l \neq j} \frac{|P_{ji}  P_{lk}  -  P_{jk} P_{li}|}{P_{ji}   P_{lk} + P_{li}  P_{jk}}}{\sum_{\pi \in \Pi}\sum_{\tilde\pi \in \Pi} \sum_{l \neq j} q_{j\pi}q_{l\tilde\pi}\left\{ P_{ji}   P_{lk} + P_{li}  P_{jk} \right\}} \right| \\
    &\leq \max_{l \neq j} \frac{|P_{ji}  P_{lk}  -  P_{jk} P_{li}|}{P_{ji}   P_{lk} + P_{li}  P_{jk}}.
\end{align*}
Note that to get from the second line to the third line, we change the dummy variable in the numerator from $\tilde l$ to $l$. From the third to the fourth line, we keep only a subset of terms in the denominator and change the dummy variable from $\tilde l$ to $l$. From the fourth to the fifth line, we swap the dummy variables $\pi$ and $\tilde\pi$ in the second term of the denominator and in the second term of the numerator. Therefore, under Assumption~\ref{cond:approximationerror},
\begin{equation*}
    \zeta(P^{t,t+1}) \leq \frac{1}{2} \max_{1\leq i,k \leq K} \sum_{j} \max_{l \neq j} \frac{|P_{ji}  P_{lk}  -  P_{jk} P_{li}|}{P_{ji}   P_{lk} + P_{li}  P_{jk}} \leq \phi.
\end{equation*}
As a result,
\begin{align*}
    &|P(Z_t = j | Y_0,\ldots,Y_T, Z_{t+s} = i) - P(Z_t = j | Y_0,\ldots,Y_T)| \leq \phi^s;\\
    &|P(Z_t = j | Y_0,\ldots,Y_{t+s}, Z_{t+s} = i) - P(Z_t = j | Y_0,\ldots,Y_{t+s})| \leq \phi^s.
\end{align*}
Moreover, $P(Z_t = j | Y_0,\ldots,Y_T, Z_{t+s} = i)$ and $P(Z_t = j | Y_0,\ldots,Y_{t+s}, Z_{t+s} = i)$ are equal due to the Markov property. Therefore, by triangle inequality,
\begin{equation}\label{eq: approximation truncate future}
    |P(Z_t = j | Y_0,\ldots,Y_T) - P(Z_t = j | Y_0,\ldots,Y_{t+s})| \leq 2\phi^s.
\end{equation}
Next, noting that
\begin{align*}
    P(Z_{t-1}=i, Z_t = j | Y_0,\ldots,Y_T) &= P(Z_{t-1}=i| Z_t = j, Y_0,\ldots,Y_T)P(Z_t = j | Y_0,\ldots,Y_T) \\
    &= P(Z_{t-1}=i| Z_t = j, Y_0,\ldots,Y_t)P(Z_t = j | Y_0,\ldots,Y_T),
\end{align*}
and
\begin{align*}
    P(Z_{t-1}=i, Z_t = j | Y_0,\ldots,Y_{t+s}) &= P(Z_{t-1}=i| Z_t = j, Y_0,\ldots,Y_{t+s})P(Z_t = j | Y_0,\ldots,Y_{t+s}) \\
    &= P(Z_{t-1}=i| Z_t = j, Y_0,\ldots,Y_t)P(Z_t = j | Y_0,\ldots,Y_{t+s}).
\end{align*}
Consequently, the approximation error bound in \eqref{eq: approximation truncate future} implies that
\begin{equation}\label{eq: approximation truncate future double}
    |P(Z_{t-1}=i, Z_t = j | Y_0,\ldots,Y_T) - P(Z_{t-1}=i, Z_t = j | Y_0,\ldots,Y_{t+s})| \leq 2\phi^s.
\end{equation}

\textbf{Step II.} With a slight abuse of notation, we now define the matrix $P^{t_1,t_2}$ in a different way. For generic $t_1 < t_2 \leq t$, define matrix $P^{t_1,t_2} \in \mathbb{R}^{K\times K}$ with the $(i,j)$-th entry
\begin{equation*}
    (P^{t_1,t_2})_{ij} = P_\theta(Z_{t_2} = j | Y_{t_1},\ldots,Y_{t+s},Z_{t_1} = i).
\end{equation*}
Note that regardless of the values of $t_1, t_2$, we always condition on the outcome vectors $Y$ until time $t+s$. Recall that
\begin{equation*}
    m_{j,\theta}(Y_{t-s},\ldots,Y_{t+s})=P_\theta(Z_t=j|Y_{t-s},\ldots,Y_{t+s},Z_{t-s}=1),
\end{equation*}
which, with our new definition, can be equivalently written as $(P^{t-s,t})_{1j}$. Meanwhile, we can re-write our intermediate approximation as follows.
\begin{align*}
    &\quad P(Z_t=j|Y_0,\ldots,Y_{t+s}) \\
    &= \sum_{i} P(Z_t=j|Y_0,\ldots,Y_{t+s},Z_{t-s}=i)P(Z_{t-s}=i|Y_0,\ldots,Y_{t+s}) \\
    &= \sum_{i} P(Z_t=j|Y_{t-s},\ldots,Y_{t+s},Z_{t-s}=i)P(Z_{t-s}=i|Y_0,\ldots,Y_{t+s}) \\
    &= \sum_{i} P(Z_t=j|Y_{t-s},\ldots,Y_{t+s},Z_{t-s}=1)P(Z_{t-s}=i|Y_0,\ldots,Y_{t+s}) \\
    &\quad + \sum_{i}\{P(Z_t=j|Y_{t-s},\ldots,Y_{t+s},Z_{t-s}=i)- P(Z_t=j|Y_{t-s},\ldots,Y_{t+s},Z_{t-s}=1)\}P(Z_{t-s}=i|Y_0,\ldots,Y_{t+s}) \\
    &= (P^{t-s,t})_{1j} + \sum_{i}\{(P^{t-s,t})_{ij}-(P^{t-s,t})_{1j}\}P(Z_{t-s}=i|Y_0,\ldots,Y_{t+s}).
\end{align*}
Hence, 
\begin{align*}
    &\quad |P(Z_t=j|Y_0,\ldots,Y_{t+s}) - m_{j,\theta}(Y_{t-s},\ldots,Y_{t+s})|  \\
    &= \left|\sum_{i}\{(P^{t-s,t})_{ij}-(P^{t-s,t})_{1j}\}P(Z_{t-s}=i|Y_0,\ldots,Y_{t+s})\right| \\
    &\leq \sum_{i}\left|(P^{t-s,t})_{ij}-(P^{t-s,t})_{1j}\right|P(Z_{t-s}=i|Y_0,\ldots,Y_{t+s}) \\
    &\leq \max_{i,k}\max_j|(P^{t-s,t})_{ij}-(P^{t-s,t})_{kj}|\sum_{i}P(Z_{t-s}=i|Y_0,\ldots,Y_{t+s}) \\
    &= \psi(P^{t-s,t}),
\end{align*}
where $\psi(\cdot)$ is defined as before and in the same way as in Supplementary Appendix~\ref{app:usefullemma}. Thus, if we can show that $\psi(P^{t-s,t}) \leq \phi^s$, we can get the desired upper bound on the approximation error.

We first note that condition on $Y$'s, $\{Z_t\}$ form a time-inhomogeneous Markov chain, and thus
\begin{equation*}
    P^{t-s,t} = P^{t-s,t-s+1}P^{t-s+1,t-s+2} \cdots P^{t-2,t-1}P^{t-1,t}.
\end{equation*}
By Lemma~\ref{hajnallemma}, it suffices to show that $\zeta(P^{t_1,t_1+1}) \leq \phi$ for generic $t_1 < t$, where $\zeta(\cdot)$ is defined in Supplementary Appendix~\ref{app:usefullemma}. Note that we can take the stochastic matrix on the very right in Lemma~\ref{hajnallemma} to be the identity matrix, and $\psi(\textnormal{Id}_K) = 1$. By Bayes rule, 
\begin{align*}
    &\quad (P^{t_1,t_1+1})_{ij} \\
    &= P(Z_{t_1+1} = j | Y_{t_1},\ldots,Y_{t+s},Z_{t_1} = i) \\
    &= \frac{P(Z_{t_1} = i, Y_{t_1},Z_{t_1+1} = j, Y_{t_1+1},Y_{t_1+2} \ldots,Y_{t+s})}{P(Z_{t_1} = i,Y_{t_1}, Y_{t_1+1},\ldots,Y_{t+s})} \\
    &= \frac{P(Z_{t_1} = i, Y_{t_1},Z_{t_1+1} = j, Y_{t_1+1}, \ldots,Y_{t+s})}{\sum_{l}P(Z_{t_1} = i,Y_{t_1},Z_{t_1+1} = l,Y_{t_1+1},\ldots,Y_{t+s})} \\
    &= \frac{\sum_{z_{t_1+2},\ldots,z_{t+s}}P_{ij}p_j(Y_{t_1+1}|Y_{t_1})P_{jz_{t_1+2}}p_{z_{t_1+2}}(Y_{t_1+2}|Y_{t_1+1})\ldots P_{z_{t+s-1}z_{t+s}}p_{z_{t+s}}(Y_{t+s}|Y_{t+s-1})}{\sum_{l}\sum_{z_{t_1+2},\ldots,z_{t+s}}P_{il}p_l(Y_{t_1+1}|Y_{t_1})P_{lz_{t_1+2}}p_{z_{t_1+2}}(Y_{t_1+2}|Y_{t_1+1})\ldots P_{z_{t+s-1}z_{t+s}}p_{z_{t+s}}(Y_{t+s}|Y_{t+s-1})},
\end{align*}
where $p_j(\cdot|\cdot)$ denotes the conditional density function of $Y_{t+1}|(Y_t,Z_t = j)$, for $j \in \{1,\ldots,K\}$. Let $\Pi = \{(z_{t_1+2},\ldots,z_{t+s}): z_{t_1+2},\ldots,z_{t+s} \in \{1,\ldots,K\}\}$. Furthermore, for $l \in \{1,\ldots,K\}$ and $\pi \in \Pi$, define
\begin{equation*}
    f_{l\pi} = p_l(Y_{t_1+1}|Y_{t_1})P_{lz_{t_1+2}}p_{z_{t_1+2}}(Y_{t_1+2}|Y_{t_1+1})\ldots P_{z_{t+s-1}z_{t+s}}p_{z_{t+s}}(Y_{t+s}|Y_{t+s-1}).
\end{equation*}
Then,
\begin{equation*}
     (P^{t_1,t_1+1})_{ij} = \frac{\sum_{\pi \in \Pi}P_{ij}f_{j\pi}}{\sum_{l}\sum_{\pi \in \Pi}P_{il}f_{l\pi}}.
\end{equation*}
Thus,
\begin{align*}
    &\quad \left|(P^{t_1,t_1+1})_{ij} - (P^{t_1,t_1+1})_{kj}\right| \\
    &= \left|\frac{\sum_{\pi \in \Pi}P_{ij}f_{j\pi}}{\sum_{l}\sum_{\pi \in \Pi}P_{il}f_{l\pi}} - \frac{\sum_{\tilde\pi \in \Pi}P_{kj}f_{j\tilde\pi}}{\sum_{\tilde l}\sum_{\tilde \pi \in \Pi}P_{k\tilde l}f_{\tilde l\tilde \pi}} \right| \\
    &= \frac{\left|\sum_{\pi \in \Pi}\sum_{\tilde\pi \in \Pi}\left[P_{ij}f_{j\pi}\sum_{\tilde l}P_{k\tilde l}f_{\tilde l\tilde \pi} - P_{kj}f_{j\tilde\pi}\sum_{l}P_{il}f_{l\pi}\right] \right|}{\sum_{\pi \in \Pi}\sum_{\tilde\pi \in \Pi}\left[\sum_{l}P_{il}f_{l\pi}\sum_{\tilde l}P_{k\tilde l}f_{\tilde l\tilde \pi}\right]} \\
    &\leq \frac{\left|\sum_{\pi \in \Pi}\sum_{\tilde\pi \in \Pi}\sum_{l\neq j}\left(P_{ij}f_{j\pi}P_{kl}f_{l\tilde\pi} - P_{kj}f_{j\tilde\pi}P_{il}f_{l\pi}\right)\right|}{\sum_{\pi \in \Pi}\sum_{\tilde\pi \in \Pi}\sum_{l\neq j}\left(P_{il}f_{l\pi}P_{kj}f_{j\tilde \pi}+P_{ij}f_{j\pi}P_{kl}f_{l\tilde \pi}\right)} \\
    &= \frac{\left|\sum_{\pi \in \Pi}\sum_{\tilde\pi \in \Pi}\sum_{l\neq j}\left(P_{ij}f_{j\tilde\pi}P_{kl}f_{l\pi} - P_{kj}f_{j\tilde\pi}P_{il}f_{l\pi}\right)\right|}{\sum_{\pi \in \Pi}\sum_{\tilde\pi \in \Pi}\sum_{l\neq j}\left(P_{il}f_{l\pi}P_{kj}f_{j\tilde \pi}+P_{ij}f_{j\tilde\pi}P_{kl}f_{l\pi}\right)}\\ 
    &\leq \frac{\sum_{\pi \in \Pi}\sum_{\tilde\pi \in \Pi}\sum_{l\neq j}\left|P_{ij}f_{j\tilde\pi}P_{kl}f_{l\pi} - P_{kj}f_{j\tilde\pi}P_{il}f_{l\pi}\right|}{\sum_{\pi \in \Pi}\sum_{\tilde\pi \in \Pi}\sum_{l\neq j}\left(P_{il}f_{l\pi}P_{kj}f_{j\tilde \pi}+P_{ij}f_{j\tilde\pi}P_{kl}f_{l\pi}\right)} \\
    &\leq \frac{\sum_{\pi \in \Pi}\sum_{\tilde\pi \in \Pi}\sum_{l\neq j}\left\{\left(P_{il}f_{l\pi}P_{kj}f_{j\tilde \pi}+P_{ij}f_{j\tilde\pi}P_{kl}f_{l\pi}\right)\max_{l\neq j}\frac{\left|P_{ij}P_{kl} - P_{il}P_{kj}\right|}{\left(P_{ij}P_{kl} + P_{il}P_{kj}\right)}\right\}}{\sum_{\pi \in \Pi}\sum_{\tilde\pi \in \Pi}\sum_{l\neq j}\left(P_{il}f_{l\pi}P_{kj}f_{j\tilde \pi}+P_{ij}f_{j\tilde\pi}P_{kl}f_{l\pi}\right)} \\
    &= \max_{l\neq j}\frac{\left|P_{ij}P_{kl} - P_{il}P_{kj}\right|}{\left(P_{ij}P_{kl} + P_{il}P_{kj}\right)}.
\end{align*}
Recall that
\begin{align*}
    \zeta(P^{t_1,t_1+1}) &= \max_{1\leq i,k \leq K} \sum_{j: (P^{t_1,t_1+1})_{ij} > (P^{t_1,t_1+1})_{kj}} \left((P^{t_1,t_1+1})_{ij}-(P^{t_1,t_1+1})_{kj}\right) \\
    &= \max_{1\leq i,k \leq K} \sum_{j: (P^{t_1,t_1+1})_{ij} > (P^{t_1,t_1+1})_{kj}} \left|(P^{t_1,t_1+1})_{ij}-(P^{t_1,t_1+1})_{kj}\right| \\
    &= \frac{1}{2}\max_{1\leq i,k \leq K} \sum_{j} \left|(P^{t_1,t_1+1})_{ij}-(P^{t_1,t_1+1})_{kj}\right|,
\end{align*}
where the last line follows from the fact that $\sum_{j}(P^{t_1,t_1+1})_{kj} = \sum_j(P^{t_1,t_1+1})_{ij} = 1$. This implies that
\begin{align*}
    \zeta(P^{t_1,t_1+1}) &\leq  \frac{1}{2} \max_{1\leq i,k\leq K}\sum_{j=1}^{K}\max_{l\neq j}\frac{\left|P_{ij}P_{kl} - P_{il}P_{kj}\right|}{\left(P_{ij}P_{kl} + P_{il}P_{kj}\right)} \leq \phi,
\end{align*}
under Assumption~\ref{cond:approximationerror}. As a result, we get
\begin{equation*}
    |P(Z_t=j|Y_0,\ldots,Y_{t+s}) - m_{j,\theta}(Y_{t-s},\ldots,Y_{t+s})| \leq \phi^s,
\end{equation*}
which, combined with \eqref{eq: approximation truncate future}, implies the desired result that
\begin{equation*}
    |P(Z_t = j | Y_0,\ldots,Y_T) - m_{j,\theta}(Y_{t-s},\ldots,Y_{t+s})| \leq 3\phi^s.
\end{equation*}

Finally, the approximation error of $m_{ij,\theta}$ can be upper bounded in a similar fashion. Specifically, recall that
\begin{equation*}
    m_{ij,\theta}(Y_{t-s}^{t+s}) =  P_\theta(Z_{t-1}=i,Z_t=j | Z_{t-s}=1, Y_{t-s},\ldots,Y_{t+s}).
\end{equation*}
And the difference between $m_{ij,\theta}(Y_{t-s}^{t+s})$ and the intermediate approximation $P(Z_{t-1}=i, Z_t = j | Y_0,\ldots,Y_{t+s})$ can be written as follows.
\begin{align*}
    &\quad P(Z_{t-1}=i, Z_t = j | Y_0,\ldots,Y_{t+s}) - P(Z_{t-1}=i,Z_t=j | Z_{t-s}=1, Y_{t-s},\ldots,Y_{t+s}) \\
    &= \sum_{k=1}^K P(Z_{t-1}=i,Z_t=j |Z_{t-s}=k, Y_0,Y_1,\ldots,Y_{t+s})P(Z_{t-s}=k | Y_0,Y_1,\ldots,Y_{t+s}) \\
    &\quad - P(Z_{t-1}=i,Z_t=j | Z_{t-s}=1, Y_{t-s},\ldots,Y_{t+s})  \\
    &= \sum_{k=1}^K P(Z_{t-1}=i,Z_t=j |Z_{t-s}=k, Y_{t-s},\ldots,Y_{t+s})P(Z_{t-s}=k | Y_0,Y_1,\ldots,Y_{t+s}) \\
    &\quad - P(Z_{t-1}=i,Z_t=j | Z_{t-s}=1, Y_{t-s},\ldots,Y_{t+s}) \\
    &= \sum_{k=1}^K \Big\{P(Z_t=j |Z_{t-1}=i, Y_{t-1},\ldots,Y_{t+s})P(Z_{t-1}=i |Z_{t-s}=k, Y_{t-s},\ldots,Y_{t+s}) \\
    &\quad \quad \times P(Z_{t-s}=k | Y_0,Y_1,\ldots,Y_{t+s})\Big\} \\
    &\quad - P(Z_t=j |Z_{t-1}=i, Y_{t-1},\ldots, Y_{t+s})P(Z_{t-1}=i | Z_{t-s}=1, Y_{t-s},\ldots,Y_{t+s}) \\
    &= P(Z_t=j |Z_{t-1}=i, Y_{t-1},\ldots, Y_{t+s}) \times \Bigg[\sum_{k=1}^K \Big\{P(Z_{t-1}=i |Z_{t-s}=k, Y_{t-s},\ldots,Y_{t+s}) - \\
    &\qquad \qquad P(Z_{t-1}=i |Z_{t-s}=1, Y_{t-s},\ldots,Y_{t+s})\Big\}P(Z_{t-s}=k | Y_0,Y_1,\ldots,Y_{t+s})\Bigg] \\
    &\leq \max_k \{(P^{t-s,t-1})_{ki} - (P^{t-s,t-1})_{1i}\} \leq \psi(P^{t-s,t-1}) \\
    &\leq \phi^{s-1},
\end{align*}
since $\zeta(P^{t_1,t_1+1}) \leq \phi$ for generic $t_1$ as we have shown. Combining with \eqref{eq: approximation truncate future double}, we get the desired result.

\end{proof}

\section{Proof of Lemma~\ref{RE}}\label{app:proofRE}
\begin{proof}[Proof of Lemma~\ref{RE}]
Recall that in this lemma, we aim to show that for all $j \in \{1,\ldots,K\}$ and all $v \in \mathbb{R}^{d^2}$, 
\begin{equation*}
    v^\top \left[\frac{1}{T}\sum_{t=1}^T \textnormal{Id}_d \otimes \left\{m_{j,\theta}(Y_{t-s}^{t+s}) Y_{t-1}Y_{t-1}^\top\right\}\right] v \geq \alpha \|v\|_2^2 - \tau_{RE}\|v\|_1^2,
\end{equation*}
uniformly over $\theta \in \Theta$, with high probability. We first give an outline of the proof. We start by controlling the tail behavior of $v^\top \{m_{j,\theta}(Y_{t-s}^{t+s}) Y_{t-1}Y_{t-1}^\top\} v$, which will enable us to obtain a uniform concentration results over $\theta$ over a collection of sparse vectors $v$, if we had independent and identically distributed data. Under $\beta$-mixing, this translates to a uniform concentration result for our original time series data, which leads to a lower bound on $v^\top \left[\frac{1}{T}\sum_{t=1}^T \textnormal{Id}_d \otimes \left\{m_{j,\theta}(Y_{t-s}^{t+s}) Y_{t-1}Y_{t-1}^\top\right\}\right] v$ for sparse $v$. We will then show that this lower bound for sparse vectors implies a lower bound for all vectors $v$. In the following, we present the complete proof.

First, Lemma B.1 in \citet{basu2015regularized} implies that it suffices to show that 
\begin{equation*}
    v^\top \left[\frac{1}{T}\sum_{t=1}^T m_{j,\theta}(Y_{t-s}^{t+s}) Y_{t-1}Y_{t-1}^\top\right] v \geq \alpha \|v\|_2^2 - \tau_{RE}\|v\|_1^2,
\end{equation*}
for all $v \in \mathbb{R}^d$ uniformly over $\theta$. We will prove the above statement under the assumptions of Lemma~\ref{RE} in the following steps.

\textbf{Step I: control the tail behavior of $v^\top \{m_{j,\theta}(Y_{t-s}^{t+s}) Y_{t-1}Y_{t-1}^\top\} v$.} To start, we recall that under the stationary distribution, $Y_t$ is a sub-Gaussian random vector and we define
\begin{equation*}
    K_Y = \sup_{v \in \mathbb{R}^d,\|v\|_2=1}\sup_{q \geq 1}(E|v^\top Y_t|^q)^{1/q}q^{-1/2}.
\end{equation*}
Therefore,
\begin{equation}\label{subgaussianvector}
    E|v^\top Y_t|^q \leq K_Y^q q^{q/2}, \quad \forall q \geq 1, \ \forall v\in\mathbb{R}^d,\  \|v\|_2=1.
\end{equation}
This implies that
\begin{align*}
    E\left[\left\{m_{j,\theta}(Y_{t-s}^{t+s}) v^\top  Y_{t-1}Y_{t-1}^\top v\right\}^q\right] &= E\left[\left\{m_{j,\theta}(Y_{t-s}^{t+s})\right\}^q\left\{ v^\top  Y_{t-1}Y_{t-1}^\top v\right\}^q\right] \\
    &\leq E\left[\left\{ v^\top  Y_{t-1}Y_{t-1}^\top v\right\}^q\right] \\
    &= E\left[\left|v^\top  Y_{t-1}\right|^{2q}\right] \\
    &\leq K_Y^{2q} (2q)^q,
\end{align*}
and therefore
\begin{equation*}
    \left(E\left[\left\{m_{j,\theta}(Y_{t-s}^{t+s}) v^\top  Y_{t-1}Y_{t-1}^\top v\right\}^q\right]\right)^{1/q} \leq 2K_Y^2 q.
\end{equation*}

\textbf{Step II: uniform concentration for i.i.d data with sparse $v$.} For a vector $v \in \mathbb{R}^d$, let $\textnormal{supp}(v) \subseteq \{1,\ldots,d\}$ denote the support of $v$, that is, $\textnormal{supp}(v) = \{i: v_i \neq 0\}$. Let $\mathbb{S}(2b) = \{v \in \mathbb{R}^d: |\textnormal{supp}(v)| \leq 2b, \|v\|_2 \leq 1\}$ denote the set of $2b$-sparse vectors in the $d$-dimensional unit ball. We will specify the exact value of $b$ in a later step. For any subset $\tilde{S}$ of $\{1,\ldots,d\}$ with cardinality $2b$, let $\mathbb{S}_{\tilde{S}}$ denote the subset of $\mathbb{S}(2b)$ supported on $\tilde S$. Let $\mathbb{K}_{\tilde S}$ be a $1/10$-cover of $\mathbb{S}_{\tilde S}$, and define $\mathbb{K} = \cup_{\tilde S} \mathbb{K}_{\tilde S}$. Then, $\mathbb{K}$ is a $1/10$-cover of $\mathbb{S}(2b)$. Since the $\epsilon$ covering number of a $2b$-dimensional unit ball is upper bounded by $(3/\epsilon)^{2b}$, we have $|\mathbb{K}| \leq \binom{d}{2b} 30^{2b}$. Here, the binomial coefficient $\binom{d}{2b}$ arises from the fact that $v \in \mathbb{S}(2b)$ is supported on one of the $\binom{d}{2b}$ subsets of cardinality $2b$ of $\{1,\ldots,d\}$.

Define a function $f_{v,\theta}^j$ as
\begin{equation*}
    f_{v,\theta}^j(Y_{t-s}^{t+s}) = m_{j,\theta}(Y_{t-s}^{t+s})v^\top Y_{t-1}Y_{t-1}^\top v,
\end{equation*}
and define a function class
\begin{equation*}
    \mathcal{F} = \left\{f_{v,\theta}^j: j \in \{1,\ldots,K\}, v \in \mathbb{K}, \theta \in \Theta \right\}.
\end{equation*}
We now establish a uniform concentration result over the function class $\mathcal{F}$ for i.i.d. data. To this end, for a fixed constant $\tilde{c}$, let $N = T/\{\tilde{c}\log T\}$. Let $\{\tilde{Y}_{n-s}^{n+s}\}_{n=1}^N$ be an i.i.d. sample where the marginal distribution of $\tilde{Y}_{n-s}^{n+s}$ is the same as the marginal distribution of $Y_{t-s}^{t+s}$. For the ease of notation, let $X_n = \tilde{Y}_{n-s}^{n+s}$. Note that the sample size of this i.i.d. sample is smaller than the sample size of the original time series by a log factor, and the reason for this shall become clear in later steps of the proof. We study the following tail probability
\begin{equation*}
    P\left(\sup_{f \in \mathcal{F}} \left|\frac{1}{N}\sum_{n=1}^N f(X_n) - E[f(X_n)] \right| > \frac{\rho_{\min}}{243}\right),
\end{equation*}
where $\rho_{\min}$ is defined in Assumption~\ref{cond:minmaxeigenvalue}.

\textbf{Step II-I: symmetrization.} We start with a symmetrization argument using Theorem~\ref{symmetrization}. To apply this theorem, we first need to find an upper bound for the $L_2$ norm of functions in $\mathcal{F}$. For any fixed function $f \in \mathcal{F}$, there exist $j \in \{1,\ldots,K\}$, $v \in \mathbb{K}$ and $\theta \in \Theta$ such that
\begin{align*}
    \|f\|_2^2 &= E\left[f(X_n)^2\right] = E\left[m_{j,\theta}(Y_{t-s}^{t+s})^2 \left\{v^\top Y_{t-1}Y_{t-1}^\top v \right\}^2 \right] \\
    &\leq E\left[\left\{v^\top Y_{t-1}Y_{t-1}^\top v \right\}^2 \right] \\
    &\leq 16 K_Y^4,
\end{align*}
where the second line follows from the fact that $m_{j,\theta}$ is uniformly upper bounded by 1, and the third line follows from our result from Step I. The upper bound above holds for any $v$, $j$ and $\theta$. Thus, we can take $R = 4K_Y^2$ in Theorem~\ref{symmetrization}. Now to apply this theorem, we only need that $T \geq 72\tilde{c}(\log T)(972K_Y^2/\rho_{\min})^2$. Note that this holds for sufficiently large $T$. Under this condition, Theorem~\ref{symmetrization} implies that
\begin{equation*}
    P\left(\sup_{f \in \mathcal{F}} \left|\frac{1}{N}\sum_{n=1}^N f(X_n) - E[f(X_n)] \right| > \frac{\rho_{\min}}{243}\right) \leq 4P\left(\sup_{f \in \mathcal{F}} \left|\frac{1}{N}\sum_{n=1}^N W_n f(X_n)\right| \geq \frac{\rho_{\min}}{972}\right),
\end{equation*}
for i.i.d. Rademacher random variables $W_n$.

\textbf{Step II-II: Control the empirical norm of functions in $\mathcal{F}$.} To control the probability on the right-hand side of the above display, we first condition on $X_1,\ldots, X_N$. Define the event $\mathcal{A}$ such that $I_{\mathcal{A}}\{X_1,\ldots,X_N\}= 1$ if and only if
\begin{equation*}
    \sup_{f \in \mathcal{F}} \frac{1}{N}\sum_{n=1}^N f(X_n)^2 \leq 16K_Y^4 + 1.
\end{equation*}
We now study the probability of the event $\mathcal{A}$. For a fixed vector $v \in \mathbb{K}$, define a function class
\begin{equation*}
    \mathcal{F}_v = \left\{f_{v,\theta}^j: j \in \{1,\ldots,K\}, \theta \in \Theta\right\}.
\end{equation*}
With this definition, the function class $\mathcal{F}$ can be written as $\mathcal{F} = \cup_{v \in \mathbb{K}}\mathcal{F}_v$.
\begin{align*}
    \sup_{f \in \mathcal{F}_v} \frac{1}{N}\sum_{n=1}^N f(X_n)^2 &= \max_{j}\sup_{\theta \in \Theta} \frac{1}{N}\sum_{n=1}^N m_{j,\theta}(\tilde Y_{n-s}^{n+s})^2 \left\{v^\top \tilde Y_{n-1}\tilde Y_{n-1}^\top v \right\}^2 \\
    &\leq \max_{j}\sup_{\theta \in \Theta} \frac{1}{N}\sum_{n=1}^N \left\{v^\top \tilde Y_{n-1}\tilde Y_{n-1}^\top v \right\}^2 \\
    &= \frac{1}{N}\sum_{n=1}^N \left\{v^\top \tilde Y_{n-1}\tilde Y_{n-1}^\top v \right\}^2.
\end{align*}
As shown in Step I, for any fixed vector $v$, the random variable $v^\top \tilde Y_{n-1}\tilde Y_{n-1}^\top v$ is sub-weibull$(1)$ with sub-weibull norm $2K_Y^2$. By Lemma 6 in \citet{wong2020lasso}, $\{v^\top \tilde Y_{n-1}\tilde Y_{n-1}^\top v\}^2$ is sub-weibull$(1/2)$ with sub-weibull norm $16K_Y^4$. Now by Lemma~\ref{keyconcentration},
\begin{multline*}
    P\left(\left|\frac{1}{N}\sum_{n=1}^N \left\{v^\top \tilde Y_{n-1}\tilde Y_{n-1}^\top v \right\}^2 - E\left[\left\{v^\top \tilde Y_{n-1}\tilde Y_{n-1}^\top v \right\}^2\right]\right| > 1\right) \\
    \leq N\exp\left\{-\frac{N^{1/2}}{4K_Y^2\tilde C_1}\right\} + \exp\left\{-\frac{N}{256K_Y^8\tilde C_2}\right\},
\end{multline*}
for some constants $\tilde C_1$ and $\tilde C_2$. As we have shown, $E[\{v^\top \tilde Y_{n-1}\tilde Y_{n-1}^\top v \}^2] \leq 16 K_Y^4$. Together with the above display, this implies that
\begin{equation*}
    P\left(\frac{1}{N}\sum_{n=1}^N \left\{v^\top \tilde Y_{n-1}\tilde Y_{n-1}^\top v \right\}^2 > 16K_Y^4 + 1\right)\leq N\exp\left\{-\frac{N^{1/2}}{4K_Y^2\tilde C_1}\right\} + \exp\left\{-\frac{N}{256K_Y^8\tilde C_2}\right\},
\end{equation*}
and therefore
\begin{equation*}
    P\left(\sup_{f \in \mathcal{F}_v} \frac{1}{N}\sum_{n=1}^N f(X_n)^2 > 16K_Y^4 + 1 \right) \leq N\exp\left\{-\frac{N^{1/2}}{4K_Y^2\tilde C_1}\right\} + \exp\left\{-\frac{N}{256K_Y^8\tilde C_2}\right\},
\end{equation*}
for a fixed $v \in \mathbb{K}$. Applying a union bound over $\mathbb{K}$, we have that
\begin{multline*}
    P\left(\sup_{f \in \mathcal{F}} \frac{1}{N}\sum_{n=1}^N f(X_n)^2 > 16K_Y^4 + 1 \right) \leq \binom{d}{2b} 30^{2b}N\exp\left\{-\frac{N^{1/2}}{4K_Y^2\tilde C_1}\right\}  \\
    +  \binom{d}{2b}30^{2b}\exp\left\{-\frac{N}{256K_Y^8\tilde C_2}\right\}.
\end{multline*}
This provides a way to control the empirical norm of functions in the class $\mathcal{F}$ uniformly with high probability.

\textbf{Step II-III: condition on $X_1,\ldots, X_N$.} We now condition on $X_1,\ldots,X_N$ and study the probability
\begin{equation*}
    P\left(\sup_{f \in \mathcal{F}} \left|\frac{1}{N}\sum_{n=1}^N W_n f(X_n)\right| \geq \frac{\rho_{\min}}{972} \Bigr\rvert X_1=x_1,\ldots, X_N=x_N, I_\mathcal{A}\left\{X_1,\ldots,X_N\right\} =1 \right).
\end{equation*}
for a set of values $x_1,\ldots, x_N$ such that $I_{\mathcal{A}}\{x_1,\ldots,x_N\} = 1$. Our main tool to control the probability above is Corollary 8.3 in \citet{geer2000empirical}, which requires controlling the entropy of the function class $\mathcal{F}$.

By a Sudakov minoration argument similar to the proof of Proposition~\ref{errororder}, we have
\begin{equation*}
    \sqrt{\log N_c(\epsilon,\Theta,\|\cdot\|_2)} \leq C_1r \sqrt{|S|(\log K + \log d)}/\epsilon,
\end{equation*}
for some constant $C_1$. Now consider the function class $\mathcal{F}_v^j = \{f_{v,\theta}^j: \theta \in \Theta\}$. We aim to relate the entropy of $\mathcal{F}_v^j$ to that of $\Theta$, by showing that functions in $\mathcal{F}_v^j$ are Lipschitz with respect to $\theta$. Let $Q_n(x_1,\ldots,x_N)$ denote the empirical distribution that puts mass $1/N$ at each value $x_n$. We will often omit in the notation its dependence on $(x_1,\ldots,x_N)$ and write $Q_n$ for simplicity. For a function $f$, define its norm under $Q_n$, $\|f\|_{Q_n}$, such that $\|f\|_{Q_n}^2 = \int f^2dQ_n = \sum_{n=1}^N f^2(x_n)/N$.
\begin{align*}
    &\quad \left\|f_{v,\theta_1}^j - f_{v,\theta_2}^j \right\|_{Q_n}^2 \\
    &= \frac{1}{N}\sum_{n=1}^N \left\{f_{v,\theta_1}^j(\tilde Y_{n-s}^{n+s}) - f_{v,\theta_2}^j(\tilde Y_{n-s}^{n+s})\right\}^2 \\
    &= \frac{1}{N}\sum_{n=1}^N \left\{m_{j,\theta_1}(\tilde Y_{n-s}^{n+s})v^\top \tilde Y_{n-1} \tilde Y_{n-1}^\top v - m_{j,\theta_2}(\tilde Y_{n-s}^{n+s})v^\top \tilde Y_{n-1} \tilde Y_{n-1}^\top v\right\}^2 \\
    &= \frac{1}{N}\sum_{n=1}^N \left\{v^\top \tilde Y_{n-1} \tilde Y_{n-1}^\top v\right\}^2  \left\{\int_0^1 \frac{\partial m_{j,\theta}(\tilde{Y}_{n-s}^{n+s})}{\partial \theta}\rvert_{\theta = u\theta_1 + (1-u)\theta_2}^\top (\theta_2 - \theta_1)du\right\}^2 \\
    &\leq \left[\frac{1}{N}\sum_{n=1}^N \left\{v^\top \tilde Y_{n-1} \tilde Y_{n-1}^\top v\right\}^4\right]^{1/2} \left[\frac{1}{N}\sum_{n=1}^N \left\{\int_0^1 \frac{\partial m_{j,\theta}(\tilde{Y}_{n-s}^{n+s})}{\partial \theta}\rvert_{\theta = u\theta_1 + (1-u)\theta_2}^\top (\theta_2 - \theta_1)du\right\}^4\right]^{1/2},
\end{align*}
where the last line follows from Cauchy-Schwarz inequality. The first term in the last line is upper bounded by some constant uniformly in $v$ with high probability. To see this, we apply Lemma~6 in \citet{wong2020lasso} again and get that $\{v^\top \tilde Y_{n-1} \tilde Y_{n-1}^\top v\}^4$ is sub-weibull$(1/4)$ with sub-weibull norm $K_{4,Y} = 2^4 (16 K_Y^4)^2$. Applying Lemma~\ref{keyconcentration}, we get the following concentration result:
\begin{multline*}
    P\left( \left|\frac{1}{N}\sum_{n=1}^N \left\{v^\top \tilde Y_{n-1} \tilde Y_{n-1}^\top v\right\}^4 - E\left[\left\{v^\top \tilde Y_{n-1} \tilde Y_{n-1}^\top v\right\}^4\right]\right| > 1\right) \leq N\exp\left(-\frac{N^{1/4}}{K_{4,Y}^{1/4}\tilde C_1}\right) \\
    + \exp\left(-\frac{N}{K_{4,Y}^2\tilde C_2}\right),
\end{multline*}
for $N \geq 4$. Recall that we have shown in step I that $E[\{v^\top \tilde Y_{n-1} \tilde Y_{n-1}^\top v\}^4] \leq 8^4K_Y^8$ for all $v \in \mathbb{K}$. Combined with the above concentration result, we have that
\begin{equation*}
    P\left( \frac{1}{N}\sum_{n=1}^N \left\{v^\top \tilde Y_{n-1} \tilde Y_{n-1}^\top v\right\}^4 > 1 + 8^4K_Y^8\right) \leq N\exp\left(-\frac{N^{1/4}}{K_{4,Y}^{1/4}\tilde C_1}\right) + \exp\left(-\frac{N}{K_{4,Y}^2\tilde C_2}\right),
\end{equation*}
for any fixed $v \in \mathbb{K}$. Applying a union bound over $\mathbb{K}$, we have that
\begin{multline*}
    P\left( \sup_{v \in \mathbb{K}} \frac{1}{N}\sum_{n=1}^N \left\{v^\top \tilde Y_{n-1} \tilde Y_{n-1}^\top v\right\}^4 > 1 + 8^4K_Y^8\right) \leq \binom{d}{2b}30^{2b} N\exp\left(-\frac{N^{1/4}}{K_{4,Y}^{1/4}\tilde C_1}\right) \\
    + \binom{d}{2b}30^{2b} \exp\left(-\frac{N}{K_{4,Y}^2\tilde C_2}\right).
\end{multline*}
We now turn to the second term, 
\begin{align*}
    &\quad \left[\frac{1}{N}\sum_{n=1}^N \left\{\int_0^1 \frac{\partial m_{j,\theta}(\tilde{Y}_{n-s}^{n+s})}{\partial \theta}\rvert_{\theta = u\theta_1 + (1-u)\theta_2}^\top (\theta_2 - \theta_1)du\right\}^4\right]^{1/2} \\
    &\leq \left[\frac{1}{N}\sum_{n=1}^N \int_0^1 \left\{\frac{\partial m_{j,\theta}(\tilde{Y}_{n-s}^{n+s})}{\partial \theta}\rvert_{\theta = u\theta_1 + (1-u)\theta_2}^\top (\theta_2 - \theta_1)\right\}^4 du \right]^{1/2} \\
    &= \left[\frac{1}{N}\sum_{n=1}^N \int_0^1 \left\{(\theta_2 - \theta_1)^\top\frac{\partial m_{j,\theta}(\tilde{Y}_{n-s}^{n+s})}{\partial \theta}\rvert_{\theta = u\theta_1 + (1-u)\theta_2}\frac{\partial m_{j,\theta}(\tilde{Y}_{n-s}^{n+s})}{\partial \theta}\rvert_{\theta = u\theta_1 + (1-u)\theta_2}^\top (\theta_2 - \theta_1) \right\}^2 du \right]^{1/2} \\
    &= \left[\int_0^1 \frac{1}{N}\sum_{n=1}^N\left\{(\theta_2 - \theta_1)^\top\frac{\partial m_{j,\theta}(\tilde{Y}_{n-s}^{n+s})}{\partial \theta}\rvert_{\theta = u\theta_1 + (1-u)\theta_2}\frac{\partial m_{j,\theta}(\tilde{Y}_{n-s}^{n+s})}{\partial \theta}\rvert_{\theta = u\theta_1 + (1-u)\theta_2}^\top (\theta_2 - \theta_1) \right\}^2 du \right]^{1/2} \\
    &\leq \left[\int_0^1 \frac{1}{N}\sum_{n=1}^N \left\|\frac{\partial m_{j,\theta}(\tilde{Y}_{n-s}^{n+s})}{\partial \theta}\rvert_{\theta = u\theta_1 + (1-u)\theta_2}\frac{\partial m_{j,\theta}(\tilde{Y}_{n-s}^{n+s})}{\partial \theta}\rvert_{\theta = u\theta_1 + (1-u)\theta_2}^\top\right\|_2^2 \|\theta_2 - \theta_1\|_2^4 du\right]^{1/2} \\
    &= \left[\int_0^1 \frac{1}{N}\sum_{n=1}^N \left\|\frac{\partial m_{j,\theta}(\tilde{Y}_{n-s}^{n+s})}{\partial \theta}\rvert_{\theta = u\theta_1 + (1-u)\theta_2}\frac{\partial m_{j,\theta}(\tilde{Y}_{n-s}^{n+s})}{\partial \theta}\rvert_{\theta = u\theta_1 + (1-u)\theta_2}^\top\right\|_2^2 du\right]^{1/2} \|\theta_2 - \theta_1\|_2^2 \\
    &\leq \sup_{\tilde\theta \in \mathcal{B}(r,\theta^*)} \left[ \frac{1}{N}\sum_{n=1}^N \left\|\frac{\partial m_{j,\theta}(\tilde{Y}_{n-s}^{n+s})}{\partial \theta}\rvert_{\theta = \tilde\theta}\frac{\partial m_{j,\theta}(\tilde{Y}_{n-s}^{n+s})}{\partial \theta}\rvert_{\theta = \tilde\theta}^\top\right\|_2^2 \right]^{1/2}\|\theta_2 - \theta_1\|_2^2.
\end{align*}
Define a Lipschitz constant $L(x_1^N)$ such that
\begin{equation*}
    L^2(x_1^N) = \left(1 + 64K_Y^4\right) \max_j\sup_{\tilde\theta \in \mathcal{B}(r,\theta^*)} \left[ \frac{1}{N}\sum_{n=1}^N \left\|\frac{\partial m_{j,\theta}(\tilde{Y}_{n-s}^{n+s})}{\partial \theta}\rvert_{\theta = \tilde\theta}\frac{\partial m_{j,\theta}(\tilde{Y}_{n-s}^{n+s})}{\partial \theta}\rvert_{\theta = \tilde\theta}^\top\right\|_2^2 \right]^{1/2},
\end{equation*}
and then we have
\begin{equation*}
    \left\|f_{v,\theta_1}^j - f_{v,\theta_1}^j \right\|_{Q_n}^2 \leq L^2(x_1^N)\|\theta_2 - \theta_1\|_2^2.
\end{equation*}
Therefore, we can construct an $\epsilon L(x_1^N)$-cover of the function class $\mathcal{F}_v^j$ from an $\epsilon$-cover of $\Theta$. As a result,
\begin{equation*}
    \sqrt{\log N_c(\epsilon L(x_1^N),\mathcal{F}_v^j,\|\cdot\|_{Q_n})} \leq \sqrt{\log N_c(\epsilon,\Theta,\|\cdot\|_2)} \leq \frac{C_1r \sqrt{|S|(\log K + \log d)}}{\epsilon},
\end{equation*}
and 
\begin{equation*}
    \sqrt{\log N_c(\epsilon,\mathcal{F}_v^j,\|\cdot\|_{Q_n})} \leq \frac{C_1r L(x_1^N) \sqrt{|S|(\log K + \log d)}}{\epsilon}.
\end{equation*}
As $\mathcal{F} = \cup_{j,v}\mathcal{F}_v^j$, we have that
\begin{equation*}
    \sqrt{\log N_c(\epsilon,\mathcal{F},\|\cdot\|_{Q_n})} \leq \sqrt{\log K + \log \binom{d}{2b} + 2b\log 30} + \sqrt{\log N_c(\epsilon,\mathcal{F}_v^j,\|\cdot\|_{Q_n})}.
\end{equation*}

We are now ready to apply Theorem~\ref{corollary83}. If $I_{A}\{x_1,\ldots,x_N\}=1$, we can take $R^2 = \max\{16K_Y^4+1,(\rho_{\min}/972)^2\}$ which guarantees that $R \geq \rho_{\min}/972$. We now compute the entropy integral,
\begin{align*}
    \int_{(\rho_{\min}/972)/8}^R \sqrt{\log N_c(\epsilon,\mathcal{F},\|\cdot\|_{Q_n})} d\epsilon &\leq C_2\sqrt{\log K + \log \binom{d}{2b} + 2b\log 30} \\
    &\quad + C_1 r L(x_1^N) \sqrt{|S|(\log K + \log d)} \int_{(\rho_{\min}/972)/8}^R \frac{1}{\epsilon}d\epsilon \\
    &\leq  C_2\sqrt{\log K + \log \binom{d}{2b} + 2b\log 30} \\
    &\quad + C_3 r L(x_1^N) \sqrt{|S|(\log K + \log d)}.
\end{align*}
When the following holds,
\begin{multline}\label{corollary83rec}
    \sqrt{N}\geq 1944C \rho_{\min}^{-1}\max\Bigg\{\rho_{\min}/972, 4K_Y^2+1, \\
    C_2\sqrt{\log K + \log \binom{d}{2b} + 2b\log 30} + C_3 r L(x_1^N) \sqrt{|S|(\log K + \log d)}\Bigg\},
\end{multline}
Theorem~\ref{corollary83} implies that
\begin{equation*}
    P\left(\sup_{f \in \mathcal{F}} \left|\frac{1}{N}\sum_{n=1}^N W_n f(x_n)\right| \geq \frac{\rho_{\min}}{972}\right) \leq C\exp\left\{-\left(\frac{\rho_{\min}}{972}\right)^2\frac{N}{4C^2\max\{16K_Y^4+1,(\rho_{\min}/972)^2\}}\right\}.
\end{equation*}

\textbf{Step II-IV: marginalize over $X_1,\ldots,X_N$.} Let $\mathcal{A}_1$ denote the event that 
\begin{equation*}
    \sqrt{N}\geq 3888 C \rho_{\min}^{-1}C_3 r L(x_1^N) \sqrt{|S|(\log K + \log d)}.
\end{equation*}
Suppose that $T$ is such that
\begin{equation*}
    \sqrt{N}\geq 1944C \rho_{\min}^{-1}\max\left\{\rho_{\min}/972, 4K_Y^2+1, 2C_2\sqrt{\log K + \log \binom{d}{2b} + 2b\log 30}\right\},
\end{equation*}
and that $\mathcal{A}_1$ happens, then \eqref{corollary83rec} is met. In the previous step, we derived an upper bound on the tail probability of interest, conditioned on $(x_1,\ldots, x_N)$ such that $I_{\mathcal{A}} =1$ and $I_{\mathcal{A}_1} =1$. In this step, we marginalize over $(X_1,\ldots,X_N)$. For two events $E_1$ and $E_2$, let $E_1 \wedge E_2$ denote the event that $E_1$ and $E_2$ both happen, and $E_1 \vee E_2$ denote the event that at least one of $E_1$ and $E_2$ happens. Then, we have
\begin{align*}
    &\quad P\left(\sup_{f \in \mathcal{F}} \left|\frac{1}{N}\sum_{n=1}^N W_n f(X_n)\right| \geq \frac{\rho_{\min}}{972}\right) \\
    &\leq P\left(\sup_{f \in \mathcal{F}} \left|\frac{1}{N}\sum_{n=1}^N W_n f(X_n)\right| \geq \frac{\rho_{\min}}{972} \wedge \mathcal{A} \wedge \mathcal{A}_1 \right) + P(\mathcal{A}^C) + P(\mathcal{A}_1^C) \\
    &\leq C\exp\left\{-\left(\frac{\rho_{\min}}{972}\right)^2\frac{N}{4C^2\max\{16K_Y^4+1,(\rho_{\min}/972)^2\}}\right\} \\
    &\quad + \binom{d}{2b} 30^{2b}N\exp\left\{-\frac{N^{1/2}}{4K_Y^2\tilde C_1}\right\} +  \binom{d}{2b}30^{2b}\exp\left\{-\frac{N}{256K_Y^8\tilde C_2}\right\} \\
    &\quad + \binom{d}{2b}30^{2b} N\exp\left(-\frac{N^{1/4}}{K_{4,Y}^{1/4}\tilde C_1}\right) + \binom{d}{2b}30^{2b} \exp\left(-\frac{N}{K_{4,Y}^2\tilde C_2}\right) \\
    &\quad + \tilde{u}(N,d) \\
    &:= \tilde{u}_{RE}(N,d).
\end{align*}
Combined with the symmetrization result we obtained from step II-I, we have that
\begin{equation*}
    P\left(\sup_{f \in \mathcal{F}} \left|\frac{1}{N}\sum_{n=1}^N f(X_n) - E[f(X_n)] \right| > \frac{\rho_{\min}}{243}\right) \leq 4\tilde{u}_{RE}(N,d).
\end{equation*}

\textbf{Step III: uniform concentration for $\beta$-mixing process with sparse $v$.} Now we extend the above uniform concentration results to the process $\{Y_{t-s}^{t+s}\}$ by applying Theorem~\ref{karandikarthm}. The detailed argument is the same as in the proof of Proposition~\ref{errororder}.
\begin{equation*}
    P\left(\sup_{f \in \mathcal{F}} \left|\frac{1}{T}\sum_{t=1}^T f(Y_{t-s}^{t+s}) - E[f(Y_{t-s}^{t+s})] \right| > \frac{\rho_{\min}}{243}\right) \leq C_3(\log T)\tilde{u}_{RE}(N,d) + \frac{2}{T^{1/2}},
\end{equation*}
for some constant $C_3$.

\textbf{Step IV: extension to all vectors $v$.} For every $v \in \mathbb{S}(2b)$, there exists some $\tilde{v}(v) \in \mathbb{K}$ such that $\|v-\tilde{v}(v)\| \leq 1/10$ and that $v$ and $\tilde{v}(v)$ have the same support. For ease of notation, define 
\begin{equation*}
    \Delta_{\Sigma,\theta}^j := \frac{1}{T}\sum_{t=1}^{T}  m_{j,\theta}(Y_{t-s}^{t+s})Y_{t-1}Y_{t-1}^\top - E\left[\frac{1}{T}\sum_{t=1}^{T}  m_{j,\theta}(Y_{t-s}^{t+s})Y_{t-1}Y_{t-1}^\top\right].
\end{equation*}
Then, for a fixed $\theta$,
\begin{align*}
V &:= \sup_{v \in \mathbb{S}(2b)} |v^\top \Delta_{\Sigma,\theta}^j v| \\
&= \sup_{v \in \mathbb{S}(2b)} |(v-\tilde{v}(v))^\top \Delta_{\Sigma,\theta}^j (v-\tilde{v}(v)) + \tilde{v}(v)^\top \Delta_{\Sigma,\theta}^j \tilde{v}(v) + 2(v-\tilde{v}(v))^\top \Delta_{\Sigma,\theta}^j \tilde{v}(v)| \\
&\leq \sup_{v \in \mathbb{S}(2b)} |(v-\tilde{v}(v))^\top \Delta_{\Sigma,\theta}^j (v-\tilde{v}(v))| + \max_{\tilde{v}\in\mathbb{K}} |\tilde{v}^\top \Delta_{\Sigma,\theta}^j \tilde{v}| + 2\sup_{v \in \mathbb{S}(2b)}|(v-\tilde{v}(v))^\top \Delta_{\Sigma,\theta}^j \tilde{v}(v)|.
\end{align*}
The first term in the last line is upper bounded by $V/100$ as $10(v-\tilde{v}(v)) \in \mathbb{S}(2b)$ for all $v \in \mathbb{S}(2b)$. Now consider the third term.
First we note that
\begin{multline*}
    2(v-\tilde{v}(v))^\top \Delta_{\Sigma,\theta}^j (v-\tilde{v}(v)) = \frac{1}{10}[\{\tilde{v}(v)+10(v-\tilde{v}(v))\}^\top \Delta_{\Sigma,\theta}^j \{\tilde{v}(v)+10(v-\tilde{v}(v))\} \\
    -\tilde{v}(v)^\top \Delta_{\Sigma,\theta}^j \tilde{v}(v) - 10(v-\tilde{v}(v))^\top \Delta_{\Sigma,\theta}^j \{10(v-\tilde{v}(v))\} ].
\end{multline*}
And so,
\begin{align*}
    &\quad 2\sup_{v \in \mathbb{S}(2b)}|(v-\tilde{v}(v))^\top \Delta_{\Sigma,\theta}^j \tilde{v}(v)| \\
    &\leq \frac{1}{10}\sup|\{\tilde{v}(v)+10(v-\tilde{v}(v))\}^\top \Delta_{\Sigma,\theta}^j \{\tilde{v}(v)+10(v-\tilde{v}(v))\}| \\
    &\ \ + \frac{1}{10}\sup|\tilde{v}(v)^\top \Delta_{\Sigma,\theta}^j \tilde{v}(v)| + \frac{1}{10}\sup|10(v-\tilde{v}(v))^\top \Delta_{\Sigma,\theta}^j \{10(v-\tilde{v}(v))\}| \\
    &\leq \frac{4}{10}\sup\left|\left\{\frac{\tilde{v}(v)+10(v-\tilde{v}(v))}{2}\right\}^\top \Delta_{\Sigma,\theta}^j \left\{\frac{\tilde{v}(v)+10(v-\tilde{v}(v))}{2}\right\}\right| \\
    &\ \ + \frac{1}{10}\sup|\tilde{v}(v)^\top \Delta_{\Sigma,\theta}^j \tilde{v}(v)| + \frac{1}{10}\sup|10(v-\tilde{v}(v))^\top \Delta_{\Sigma,\theta}^j \{10(v-\tilde{v}(v))\}| \\
    &\leq \frac{4}{10}V + \frac{1}{10}V + \frac{1}{10}V,
\end{align*}
where the last inequality above follows from the fact that $\{\tilde{v}(v)+10(v-\tilde{v}(v))\}/2$, $\tilde{v}(v)$ and $10(v-\tilde{v}(v))$ all have norm at most 1 and hence lie in $\mathbb{S}(2b)$. Re-arranging terms, we have that
\begin{equation*}
    V \leq \frac{100}{39} \max_{\tilde{v}\in\mathbb{K}} |\tilde{v}^\top \Delta_{\Sigma,\theta}^j \tilde{v}| \leq 3\max_{\tilde{v}\in\mathbb{K}} |\tilde{v}^\top \Delta_{\Sigma,\theta}^j \tilde{v}|.
\end{equation*}
The above argument holds for any $\theta \in \Theta$ and $j$, and thus
\begin{equation*}
    \max_j\sup_{\theta \in \Theta}\sup_{v \in \mathbb{S}(2b)} |v^\top \Delta_{\Sigma,\theta}^j v| \leq 3\max_j\sup_{\theta \in \Theta}\max_{\tilde{v}\in\mathbb{K}} |\tilde{v}^\top \Delta_{\Sigma,\theta}^j \tilde{v}| = 3\sup_{f \in \mathcal{F}} \left|\frac{1}{T}\sum_{t=1}^T f(Y_{t-s}^{t+s}) - E[f(Y_{t-s}^{t+s})] \right|.
\end{equation*}
Therefore, with probability at least 
\begin{equation*}
    1-C_3(\log T)\tilde{u}_{RE}(N,d) - \frac{2}{T^{1/2}},
\end{equation*}
we have that
\begin{equation*}
    \left|\frac{1}{T}\sum_{t=1}^{T}  m_{j,\theta}(Y_{t-s}^{t+s})v^\top Y_{t-1}Y_{t-1}^\top v - E\left[  m_{j,\theta}(Y_{t-s}^{t+s})v^\top Y_{t-1}Y_{t-1}^\top v\right]\right| \leq \frac{\rho_{\min}}{81},
\end{equation*}
for all $j \in \{1,\ldots,K\}$, $\theta \in\Theta$ and $v \in \mathbb{S}^{2b}$.

Next, Lemma 12 in \citet{loh2012high} implies that, with probability at least
\begin{equation*}
    1-C_3(\log T)\tilde{u}_{RE}(N,d) - \frac{2}{T^{1/2}},
\end{equation*}
we have that
\begin{equation*}
    \left|\frac{1}{T}\sum_{t=1}^{T}  m_{j,\theta}(Y_{t-s}^{t+s})v^\top Y_{t-1}Y_{t-1}^\top v - E\left[  m_{j,\theta}(Y_{t-s}^{t+s})v^\top Y_{t-1}Y_{t-1}^\top v\right]\right| \leq \frac{\rho_{\min}}{3}\left(\|v\|_2^2 + \|v\|_1^2/b\right),
\end{equation*}
for all $v \in \mathbb{R}^d$, $j \in \{1,\ldots,K\}$ and $\theta \in\Theta$. The above display implies that
\begin{align*}
    v^\top \left[\frac{1}{T}\sum_{t=1}^{T}  m_{j,\theta}(Y_{t-s}^{t+s})Y_{t-1}Y_{t-1}^\top\right]v &\geq v^\top E\left[\frac{1}{T}\sum_{t=1}^{T}  m_{j,\theta}(Y_{t-s}^{t+s})Y_{t-1}Y_{t-1}^\top\right]v-\frac{\rho_{\min}}{3}(\|v\|_2^2 + \|v\|_1^2/b)
\end{align*}
Next we relate $m_{j,\theta}$ back to $w_{j,\theta}$. In particular
\begin{align*}
    &\quad\left|v^\top E\left[\frac{1}{T}\sum_{t=1}^{T}  m_{j,\theta}(Y_{t-s}^{t+s})Y_{t-1}Y_{t-1}^\top\right]v - v^\top E\left[\frac{1}{T}\sum_{t=1}^{T}  w_{j,\theta}(Y_0^T)Y_{t-1}Y_{t-1}^\top\right]v\right| \\
    &= \left|v^\top E\left[\frac{1}{T}\sum_{t=1}^{T}  \left\{ m_{j,\theta}(Y_{t-s}^{t+s}) - w_{j,\theta}(Y_0^T) \right\} Y_{t-1}Y_{t-1}^\top\right]v\right| \\
    &\leq \|v\|_2^2 \left\|E\left[\frac{1}{T}\sum_{t=1}^{T}  \left\{ m_{j,\theta}(Y_{t-s}^{t+s}) - w_{j,\theta}(Y_0^T) \right\} Y_{t-1}Y_{t-1}^\top\right]\right\|_2 \\
    &\leq 3\phi^s \|v\|_2^2 \|E[Y_{t-1}Y_{t-1}^\top]\|_2 \\
    & \leq 3\phi^s \rho_{\max}\|v\|_2^2.
\end{align*}
Hence,
\begin{align*}
    v^\top E\left[\frac{1}{T}\sum_{t=1}^{T}  m_{j,\theta}(Y_{t-s}^{t+s})Y_{t-1}Y_{t-1}^\top\right]v &\geq v^\top E\left[\frac{1}{T}\sum_{t=1}^{T}  w_{j,\theta}(Y_0^T)Y_{t-1}Y_{t-1}^\top\right]v  - 3\phi^s \rho_{\max}\|v\|_2^2 \\
    &\geq \left(\rho_{\min}-3\phi^s\rho_{\max}\right)\|v\|_2^2 \geq 2\rho_{\min}\|v\|_2^2/3,
\end{align*}
when $s \asymp \log T$ and $T > \exp\{\log(\rho_{\min}/(9\rho_{\max}))/\log\phi\}$. This implies that
\begin{equation*}
    v^\top \left[\frac{1}{T}\sum_{t=1}^{T}  m_{j,\theta}(Y_{t-s}^{t+s})Y_{t-1}Y_{t-1}^\top\right]v \geq \frac{\rho_{\min}}{3} \|v\|_2^2 -\frac{\rho_{\min}}{3b}\|v\|_1^2,    
\end{equation*}
for all $v$, $j$ and $\theta$.

\textbf{Step V: choose $b$ and conclusion of the proof.} So far, we have shown that 
\begin{equation*}
    v^\top \left[\frac{1}{T}\sum_{t=1}^{T}  m_{j,\theta}(Y_{t-s}^{t+s})Y_{t-1}Y_{t-1}^\top\right]v \geq \frac{\rho_{\min}}{3} \|v\|_2^2 -\frac{\rho_{\min}}{3b}\|v\|_1^2,    
\end{equation*}
for all $v \in \mathbb{R}^d$, $j \in \{1,\ldots,K\}$ and $\theta \in \Theta$, with probability at least
\begin{align*}
    1 - \frac{2}{T^{1/2}} - C_3(\log T)&\Bigg\{ C\exp\left\{-\left(\frac{\rho_{\min}}{972}\right)^2\frac{N}{4C^2\max\{16K_Y^4+1,(\rho_{\min}/972)^2\}}\right\} \\
    &\quad + \binom{d}{2b} 30^{2b}N\exp\left\{-\frac{N^{1/2}}{4K_Y^2\tilde C_1}\right\} +  \binom{d}{2b}30^{2b}\exp\left\{-\frac{N}{256K_Y^8\tilde C_2}\right\} \\
    &\quad + \binom{d}{2b}30^{2b} N\exp\left(-\frac{N^{1/4}}{K_{4,Y}^{1/4}\tilde C_1}\right) + \binom{d}{2b}30^{2b} \exp\left(-\frac{N}{K_{4,Y}^2\tilde C_2}\right) \\
    &\quad + \tilde{u}(N,d)\Bigg\}.
\end{align*}
Using the upper bound that $\binom{d}{2b} \leq (\frac{ed}{2b})^{2b}$, and set $b = T^{1/5}/(\log d)$, the above probability is lower bounded by
\begin{align*}
    u_{RE}(T,d) &= 1-\frac{2}{T^{1/2}} - C_3 C\exp\left\{\log\log T-\left(\frac{\rho_{\min}}{972}\right)^2\frac{T}{4\tilde c(\log T)C^2\max\{16K_Y^4+1,(\rho_{\min}/972)^2\}}\right\} \\
    &\quad -C_3(\log T)\tilde{u}(T/(\tilde c\log T),d) \\
    &\quad - C_3\tilde{c}^{-1}\exp\left\{C_4T^{1/5}- \frac{(T/(\tilde{c}\log T))^{1/4}}{\max\{4K_Y^2\tilde C_1,256K_Y^8\tilde C_2,K_{4,Y}^{1/4}\tilde C_1,K_{4,Y}^2\tilde C_2\}}\right\},
\end{align*}
which converges to 1. Therefore, with probability at least $u_{RE}(T,d)$,
\begin{equation*}
    v^\top \left[\frac{1}{T}\sum_{t=1}^{T}  m_{j,\theta}(Y_{t-s}^{t+s})Y_{t-1}Y_{t-1}^\top\right]v \geq \frac{\rho_{\min}}{3} \|v\|_2^2 -\frac{\rho_{\min}(\log d)}{3T^{1/5}}\|v\|_1^2,    
\end{equation*}
for all $v \in \mathbb{R}^d$, $j \in \{1,\ldots,K\}$ and $\theta \in \Theta$, for $T$ sufficiently large such that
\begin{equation*}
    T \geq 72\tilde{c}(\log T)(972K_Y^2/\rho_{\min})^2,
\end{equation*}
and
\begin{multline*}
    \sqrt{T/(\tilde{c}\log T)}\geq 1944C \rho_{\min}^{-1}\max\Bigg\{\rho_{\min}/972, 4K_Y^2+1, \\
    2C_2\sqrt{\log K + 2T^{1/5}\{1+(1+\log 30)/\log d\}}\Bigg\}.
\end{multline*}
\end{proof}

\section{Proof of Theorem~\ref{estimationerror} and Proposition~\ref{errororder}}\label{app:proofmainthm}
\begin{proof}[Proof of Theorem~\ref{estimationerror}]
We prove the theorem in two major steps. In the first step, we focus on one iteration of the EM algorithm. We show that with appropriate choice of $\lambda$, the estimation error of the updated parameter estimate can be upper bounded in terms of $\lambda$. Moreover, we give explicit requirements that the value of $\lambda$ needs to satisfy to establish this upper bound. In the second step, we choose a specific sequence of values for $\lambda$ over the iterations. We use induction to show that in each iteration, our chosen $\lambda$ value satisfies the requirements in the first step, and hence our upper bound on the estimation error holds in each iteration.

\textbf{Step I: estimation error in one iteration when $\lambda$ is chosen appropriately.} We first focus on the $q$-th iteration of the EM algorithm. Let $\theta^{(q-1)}$ denote the parameter estimate prior to the $q$-th iteration, and let $\theta^{(q)}$ be the updated parameter estimate after the $q$-th iteration. For the ease of notation, in this proof, we will often write $\theta^{(q-1)}$ as $\theta = (\beta^\top,p^\top,\sigma^\top)^\top$ and $\theta^{(q)}$ as $\hat\theta = (\hat\beta^\top,\hat p^\top,\hat\sigma^\top)^\top$.

\textbf{Step I-I: estimation error of $\beta$.} Recall that in the M-step, we have
\begin{align*}
    \hat\beta &= \argmin_{\tilde\beta} \frac{1}{T}\sum_{t=1}^{T}
    \left(\sum_{j=1}^K m_{j,\theta}(Y_{t-s}^{t+s})\left\|Y_t - (\textnormal{Id}_d \otimes Y_{t-1})^\top \tilde\beta_j \right\|_2^2\right)+ \lambda\sum_{j=1}^K \|\tilde\beta_j\|_1  \\
    &= \argmin_{\tilde\beta} \frac{1}{T}\sum_{t=1}^{T}
    \left(\sum_{j=1}^K m_{j,\theta}(Y_{t-s}^{t+s})\left\|Y_t - \tilde B_j^\top Y_{t-1} \right\|_2^2\right)+ \lambda\sum_{j=1}^K \|\tilde\beta_j\|_1,
\end{align*}
where we recall that $\tilde\beta_j = \vectorize{\tilde B_j}$. Therefore, by definition,
\begin{equation*}
    \frac{1}{T}\sum_{t=1}^T\sum_{j=1}^K m_{j,\theta}(Y_{t-s}^{t+s})\left\|Y_t -\hat B_j^\top Y_{t-1} \right\|_2^2 + \lambda\|\hat\beta\|_1 \leq \frac{1}{T}\sum_{t=1}^T\sum_{j=1}^K m_{j,\theta}(Y_{t-s}^{t+s})\left\|Y_t - (B_j^*)^\top Y_{t-1} \right\|_2^2 + \lambda\|\beta^*\|_1.
\end{equation*}
Re-arranging terms, we have that
\begin{align*}
    &\frac{1}{T}\sum_{t=1}^{T}\sum_{j=1}^K m_{j,\theta}(Y_{t-s}^{t+s})\|(B_j^*-\hat B_j)^\top Y_{t-1}\|_2^2 \\
    &\ \ \leq \underbrace{\frac{2}{T}\sum_{t=1}^{T}\sum_{j=1}^K m_{j,\theta}(Y_{t-s}^{t+s})\left\{Y_t - (B_j^*)^\top Y_{t-1}\right\}^\top(\hat B_j -B_j^*)^\top Y_{t-1}}_{\text{term 1}} + \underbrace{\lambda \left(\|\beta^*\|_1 - \|\hat\beta\|_1\right)}_{\text{term 2}}.
\end{align*}
We proceed by studying the two terms on the right-hand side of the above inequality separately.

Term 2 is easier to study. Recall that $S$ denotes the support of $\beta^*$. Then we have
\begin{align*}
    \|\hat\beta\|_1 - \|\beta^*\|_1 &= \|(\hat\beta-\beta^*+\beta^*)_S\|_1 +\|(\hat\beta-\beta^*+\beta^*)_{S^C}\|_1  - \|\beta^*\|_1 \\
    &= \|(\hat\beta-\beta^*)_S+\beta^*_S\|_1 +\|(\hat\beta-\beta^*)_{S^C}\|_1  - \|\beta^*\|_1 \\
    &\geq \|\beta^*_S\|_1 - \|(\hat\beta-\beta^*)_S\|_1 +\|(\hat\beta-\beta^*)_{S^C}\|_1  - \|\beta^*\|_1 \\
    &= \|(\hat\beta-\beta^*)_{S^C}\|_1 - \|(\hat\beta-\beta^*)_S\|_1.
\end{align*}
Thus, 
\begin{equation*}
    \lambda \left(\|\beta^*\|_1 - \|\hat\beta\|_1\right) \leq \lambda\left(\|(\hat\beta-\beta^*)_S\|_1-\|(\hat\beta-\beta^*)_{S^C}\|_1\right).
\end{equation*}

Next, we study term 1. Note that in the setting of penalized regression without regime switching, $\{Y_t-(B_j^*)^\top Y_{t-1}\}$ will be replaced by the error $\epsilon$ and thus we can apply concentration inequalities directly. However, in our setting, we need to further decompose it. Let $\hat\beta_{ji}$ denote the $i$-th column of $\hat B_j$ and $\beta_{ji}^*$ denote the $i$-th column of $B_j^*$. Then
\begin{align*}
    \text{term 1} &= \frac{2}{T}\sum_{t=1}^{T}\sum_{i=1}^{d}\sum_{j=1}^{K}\left[(\hat\beta_{ji}-\beta_{ji}^*)^\top Y_{t-1}(Y_{ti}-\beta_{ji}^{*\top}Y_{t-1})m_{j,\theta}\right] \\
    &= 2\underbrace{\sum_{i=1}^{d}\sum_{j=1}^{K}\left[(\hat\beta_{ji}-\beta_{ji}^*)^\top \left\{\frac{1}{T}\sum_{t=1}^{T} Y_{t-1}(Y_{ti}-\beta_{ji}^{*\top}Y_{t-1})m_{j,\theta}-E\left[Y_{t-1}(Y_{ti}-\beta_{ji}^{*\top}Y_{t-1})m_{j,\theta}\right]\right\}\right]}_{\text{term 1.1}} \\
    &\ \ + 2\underbrace{\sum_{i=1}^{d}\sum_{j=1}^{K}\left\{(\hat\beta_{ji}-\beta_{ji}^*)^\top E\left[\frac{1}{T}\sum_{t=1}^{T}Y_{t-1}(Y_{ti}-\beta_{ji}^{*\top}Y_{t-1})\left(m_{j,\theta}-w_{j,\theta}\right)\right]\right\}}_{\text{term 1.2}} \\
    &\ \ + 2\underbrace{\sum_{i=1}^{d}\sum_{j=1}^{K}\left\{(\hat\beta_{ji}-\beta_{ji}^*)^\top E\left[\frac{1}{T}\sum_{t=1}^{T}Y_{t-1}(Y_{ti}-\beta_{ji}^{*\top}Y_{t-1})w_{j,\theta}\right]\right\}}_{\text{term 1.3}}.
\end{align*}
Now, to control term 1.1, we have
\begin{align*}
    \text{term 1.1} &\leq \sum_{i=1}^{d}\sum_{j=1}^{K}\left\|\hat\beta_{ji}-\beta_{ji}^*\right\|_1 \left\|\frac{1}{T}\sum_{t=1}^{T} Y_{t-1}(Y_{ti}-\beta_{ji}^{*\top}Y_{t-1})m_{j,\theta}-E\left[Y_{t-1}(Y_{ti}-\beta_{ji}^{*\top}Y_{t-1})m_{j,\theta}\right]\right\|_{\infty} \\
    &\leq \max_{i,j}\left\|\frac{1}{T}\sum_{t=1}^{T} Y_{t-1}(Y_{ti}-\beta_{ji}^{*\top}Y_{t-1})m_{j,\theta}-E\left[Y_{t-1}(Y_{ti}-\beta_{ji}^{*\top}Y_{t-1})m_{j,\theta}\right]\right\|_{\infty}\sum_{i=1}^{d}\sum_{j=1}^{K}\left\|\hat\beta_{ji}-\beta_{ji}^*\right\|_1 \\
    &= \|\hat\beta-\beta^*\|_1 \max_{i,j,k}\left|\frac{1}{T}\sum_{t=1}^{T} Y_{t-1,k}(Y_{ti}-\beta_{ji}^{*\top}Y_{t-1})m_{j,\theta}-E\left[Y_{t-1,k}(Y_{ti}-\beta_{ji}^{*\top}Y_{t-1})m_{j,\theta}\right]\right| \\
    &\leq \Delta\|\hat\beta-\beta^*\|_1,
\end{align*}
where the last line holds with high probability under Assumption~\ref{DB}. Note that Assumption~\ref{DB} is a uniform concentration results when $\theta$ vary over $\Theta$ where the vector of regression coefficients is approximately sparse. As we will show later, this approximate sparsity can indeed be achieved by choosing $\lambda$ appropriately. Therefore, $\theta$ will indeed lie in $\Theta$ after the first iteration. However, in the first iteration, $\theta^{(0)}$ may not belong to $\Theta$ as discussed below Theorem~\ref{estimationerror}. In fact, when $\theta^{(0)}$ is chosen randomly in $\mathcal{B}(r;\theta^*)$ independent of the observed data, concentration results similar to those in Assumption~\ref{DB} is expected. To handle such random initialization, we only need such concentration results to hold pointwise in $\theta$, which is considerably easier to establish. Indeed, we can condition on the random initialization without changing the distribution of the observed data to upper bound the conditional probability, and then marginalize over the random initialization.

Next, for term 1.3, we first define another parameter estimate $\tilde\beta$ as
\begin{equation*}
    \tilde\beta = \argmin_{\beta} E\left[\frac{1}{T}\sum_{t=1}^{T}\sum_{j=1}^K w_{j,\theta}\left\|Y_t-\left(\textnormal{Id}_d \otimes Y_{t-1}\right)^\top \beta_j\right\|_2^2\right],
\end{equation*}
which can be regarded as one iteration in the population EM algorithm. Let $\tilde\beta_{ji}$ denote the sub-vector corresponding to the $i$-th column of $\tilde B_j$, which satisfies the first order condition
\begin{equation*}
    E\left[\frac{1}{T}\sum_{t=1}^{T}w_{j,\theta}Y_{t-1}(Y_{ti}-\tilde\beta_{ji}^\top Y_{t-1}) \right] = 0.
\end{equation*}
Hence, term 1.3 can be written as
\begin{align*}
    &\sum_{i=1}^{d}\sum_{j=1}^{K}\left[(\hat\beta_{ji}-\beta_{ji}^*)^\top E\left[\frac{1}{T}\sum_{t=1}^{T}Y_{t-1}Y_{t-1}^{\top}w_{j,\theta}\right](\tilde\beta_{ji}-\beta_{ji}^*)\right] \\
    &= (\hat\beta-\beta^*)^\top \Sigma (\tilde\beta-\beta^*) \\
    &\leq \rho_{\max}\|\hat\beta-\beta^*\|_2 \|\tilde\beta-\beta^*\|_2 \\
\end{align*}
where $\Sigma$ is a block-diagonal matrix with the $K$ diagonal blocks given by $\textnormal{Id}_d \otimes E\left[\frac{1}{T}\sum_{t=1}^{T}Y_{t-1}Y_{t-1}^{\top}w_{j,\theta}\right]$. In the above display, the third line holds under Assumption~\ref{cond:minmaxeigenvalue}. Next, we work with term 1.2. Note that $Y_{ti} = \sum_{j=1}^K I\{Z_t = j\}Y_{t-1}^\top \beta_{ji}^*+\tilde{\epsilon}_{ti}$ where $\tilde{\epsilon}_{ti} = \sum_{j=1}^K I\{Z_t = j\}\sigma_j\epsilon_{ti}$, and therefore we have
\begin{align*}
    \text{term 1.2} &= \sum_{j=1}^{K}\sum_{i=1}^{d}\left\{(\hat\beta_{ji}-\beta_{ji}^*)^\top E\left[\frac{1}{T}\sum_{t=1}^{T}Y_{t-1}\left(Y_{t-1}^\top \sum_{l\neq j}I\{Z_t = l\} (\beta_{li}^*-\beta_{ji}^*)+\tilde{\epsilon}_{ti}\right) \left(m_{j,\theta}-w_{j,\theta}\right)\right]\right\}.
\end{align*}
Recall that we have defined the vector $D_\beta = (D_{\beta,1}^\top,\ldots,D_{\beta,K}^\top)^\top$, $D_{\beta,j} = (\beta_1^{*\top}-\beta_j^{*\top},\ldots,\beta_{j-1}^{*\top}-\beta_j^{*\top},\beta_{j+1}^{*\top}-\beta_j^{*\top},\ldots,\beta_{K}^{*\top}-\beta_j^{*\top})^\top$. With these notations, we can further split the above expression into two terms with the first one being
\begin{align*}
    &\sum_{j=1}^{K}\sum_{i=1}^{d}(\hat\beta_{ji}-\beta_{ji}^*)^\top E\left[\frac{1}{T}\sum_{t=1}^{T}\left\{Y_{t-1}\left(m_{j,\theta}-w_{j,\theta}\right)Y_{t-1}^\top\sum_{l \neq j}I\{Z_t=l\} (\beta_{li}^*-\beta_{ji}^*) \right\} \right] \\
    &= \sum_{l \neq j, 1\leq l \leq K} \sum_{j=1}^{K}\sum_{i=1}^{d}(\hat\beta_{ji}-\beta_{ji}^*)^\top E\left[\frac{1}{T}\sum_{t=1}^{T} Y_{t-1}\left(m_{j,\theta}-w_{j,\theta}\right)I\{Z_t=l\}Y_{t-1}^\top \right](\beta_{li}^*-\beta_{ji}^*)  \\
    &= \sum_{l \neq j, 1\leq l \leq K} \sum_{j=1}^{K} (\hat\beta_{j}-\beta_{j}^*)^\top \left\{\textnormal{Id}_d \otimes E\left[\frac{1}{T}\sum_{t=1}^{T} Y_{t-1}\left(m_{j,\theta}-w_{j,\theta}\right)I\{Z_t=l\}Y_{t-1}^\top \right] \right\} (\beta_{l}^*-\beta_{j}^*) \\
    &= \sum_{j=1}^K \left\{\boldsymbol{e}_{K-1} \otimes (\hat\beta_j - \beta_j^*)\right\}^\top \left(\textnormal{Id}_{K-1} \otimes M_{1j} \right)D_{\beta,j} \\
    &= v_\beta M_1 D_\beta 
     \leq \|v_\beta\|_2 \|M_1D_\beta\|_2 \\
    &\leq (K-1)^{1/2}\|\hat\beta - \beta^*\|_2 \|M_1D_\beta\|_2,
\end{align*}
where $M_{1j}$ is the block diagonal matrix $\textnormal{Id}_d \otimes E\left[\frac{1}{T}\sum_{t=1}^{T}Y_{t-1}I\{Z_t=l\} \left(m_{j,\theta}-w_{j,\theta}\right)Y_{t-1}^\top \right]$, $M_1$ is a block diagomal matrix with the $j$-th diagonal block being $\textnormal{Id}_{K-1} \otimes M_{1j}$, and $v_\beta$ is obtained by concatenating the vectors $\boldsymbol{e}_{K-1} \otimes (\hat\beta_j - \beta_j^*)$ with $\boldsymbol{e}_{K-1}$ being a vector of 1 of length $K-1$. Note that the last line in the above display holds since $\|v_\beta\|_2 \leq (K-1)^{1/2}\|\hat\beta - \beta^*\|_2$. The operator norm of $M_1$ is the same as the maximum of the operator norms of $E\left[\frac{1}{T}\sum_{t=1}^{T}Y_{t-1}I\{Z_t=l\} \left(m_{j,\theta}-w_{j,\theta}\right)Y_{t-1}^\top \right]$. In particular,
\begin{align*}
    &\quad \left\|E\left[\frac{1}{T}\sum_{t=1}^{T}Y_{t-1}I\{Z_t=l\} \left(m_{j,\theta}-w_{j,\theta}\right)Y_{t-1}^\top \right]\right\|_2 \\
    &= \max\Bigg\{\left|\lambda_{\min}\left(E\left[\frac{1}{T}\sum_{t=1}^{T}Y_{t-1}I\{Z_t=l\} \left(m_{j,\theta}-w_{j,\theta}\right)Y_{t-1}^\top \right]\right)\right|,\\
    &\quad \left|\lambda_{\max}\left(E\left[\frac{1}{T}\sum_{t=1}^{T}Y_{t-1}I\{Z_t=l\} \left(m_{j,\theta}-w_{j,\theta}\right)Y_{t-1}^\top \right]\right)\right|\Bigg\} \\
    &= \max_{\|u\|_2=1} \left|E\left[u^\top\frac{1}{T}\sum_{t=1}^{T}Y_{t-1}I\{Z_t=l\} \left(m_{j,\theta}-w_{j,\theta}\right)Y_{t-1}^\top u \right]\right| \\
    &\leq \max_{\|u\|_2=1} E\left[\frac{1}{T}\sum_{t=1}^{T}u^\top Y_{t-1}Y_{t-1}^\top u \left|I\{Z_t=l\} \left(m_{j,\theta}-w_{j,\theta}\right)\right| \right] \\
    &\leq 3\phi^s \max_{\|u\|_2=1} E\left[\frac{1}{T}\sum_{t=1}^{T}u^\top Y_{t-1}Y_{t-1}^\top u \right] \\
    &\leq 3\phi^s\rho_{\max},
\end{align*}
by Lemma~\ref{lemma:approximationerror} and Assumption~\ref{cond:minmaxeigenvalue}. Thus, we have that
\begin{equation*}
    (K-1)^{1/2}\|\hat\beta - \beta^*\|_2 \|M_1D_\beta\|_2 \leq 3\phi^s \rho_{\max}(K-1)^{1/2}\|\hat\beta - \beta^*\|_2\|D_\beta\|_2.
\end{equation*}
The second term in term 1.2 is given by
\begin{align*}
    &\sum_{j=1}^{K}\sum_{i=1}^{d}\left\{(\hat\beta_{ji}-\beta_{ji}^*)^\top E\left[\frac{1}{T}\sum_{t=1}^{T}Y_{t-1}\tilde{\epsilon}_{ti} \left(m_{j,\theta}-w_{j,\theta}\right)\right]\right\} \\
    &\quad \leq \|\hat\beta-\beta^*\|_1\max_{i,j,k}\left|E\left[\frac{1}{T}\sum_{t=1}^{T}Y_{t-1,k}\tilde{\epsilon}_{ti} \left(m_{j,\theta}-w_{j,\theta}\right)\right]\right| \\
    &\quad \leq 3\phi^s \|\hat\beta-\beta^*\|_1\max_{i,j,k}E\left|Y_{t-1,k}\tilde{\epsilon}_{ti}\right| \\
    &\quad \leq \phi^sC_1\|\hat\beta-\beta^*\|_1,
\end{align*}
for some constant $C_1$, where we note that $E|Y_{t-1,k}\tilde{\epsilon}_{ti}| = E|Y_{t-1,k}\sum_{j=1}^K I\{Z_t = j\}\sigma_j\epsilon_{ti}| \leq \sigma_{\max}E|Y_{t-1,k}\epsilon_{ti}|$, and $E|Y_{t-1,k}\epsilon_{ti}| \leq (E[Y_{t-1,k}^2\epsilon_{ti}^2])^{1/2} \leq (E[Y_{t-1,k}^2]E[\epsilon_{ti}^2])^{1/2}$, which is upper bounded uniformly in $i$ and $k$ as $Y_t$ is a sub-Gaussian random vector.

Up to now, we have decomposed $\frac{1}{T}\sum_{t=1}^{T}\sum_{j=1}^K m_{j,\theta}(Y_{t-s}^{t+s})\|(B_j^*-\hat B_j)^\top Y_{t-1}\|_2^2$ into different terms and bounded each term. Putting these upper bounds together, we have that
\begin{align}\label{preRE}
    &\frac{1}{T}\sum_{t=1}^{T}\sum_{j=1}^K m_{j,\theta}(Y_{t-s}^{t+s})\|(B_j^*-\hat B_j)^\top Y_{t-1}\|_2^2 \nonumber \\
    &\ \ \leq \Delta\|\hat\beta-\beta^*\|_1 + \rho_{\max} \|\hat\beta-\beta^*\|_2 \|\tilde\beta-\beta^*\|_2 + 3\phi^s \rho_{\max}(K-1)^{1/2}\|\hat\beta - \beta^*\|_2\|D_\beta\|_2 +\phi^s C_1\|\hat\beta-\beta^*\|_1 \nonumber \\ 
    &\ \ + \lambda\left(\|(\hat\beta-\beta^*)_S\|_1-\|(\hat\beta-\beta^*)_{S^C}\|_1\right) \nonumber \\
    &\ \ \leq \frac{3\lambda}{2}\|(\hat\beta-\beta^*)_S\|_1-\frac{\lambda}{2}\|(\hat\beta-\beta^*)_{S^C}\|_1 + (\rho_{\max}\|\tilde\beta-\beta^*\|_2 + 3\phi^s\rho_{\max}(K-1)^{1/2}\|D_\beta\|_2 )\|\hat\beta-\beta^*\|_2,
\end{align}
where the last line follows provided that $\lambda$ is chosen such that $\lambda\geq 2\Delta+2C_1\phi^s$. Thus,
\begin{align*}
    \lambda\|(\hat\beta-\beta^*)_{S^C}\|_1 &\leq 3\lambda\|(\hat\beta-\beta^*)_S\|_1 + 2(\rho_{\max}\|\tilde\beta-\beta^*\|_2  + 3\phi^s\rho_{\max}(K-1)^{1/2}\|D_\beta\|_2 )\|\hat\beta-\beta^*\|_2 \\
    &\leq \left\{3\lambda \sqrt{|S|}+2(\rho_{\max}\|\tilde\beta-\beta^*\|_2  + 3\phi^s\rho_{\max}(K-1)^{1/2}\|D_\beta\|_2 )\right\}\|\hat\beta-\beta^*\|_2 \\
    &\leq 4\lambda \sqrt{|S|} \|\hat\beta-\beta^*\|_2,
\end{align*}
when $2(\rho_{\max}\|\tilde\beta-\beta^*\|_2  + 3\phi^s\rho_{\max}(K-1)^{1/2}\|D_\beta\|_2) \leq \lambda\sqrt{|S|}$. The above display would also imply that $\|\hat\beta-\beta^*\|_1 \leq 5\sqrt{|S|}\|\hat\beta-\beta^*\|_2$. Note that we need to choose $\lambda$ to be such that
\begin{align}
    \lambda &\geq 2\Delta + 2C_1\phi^s; \label{lambdareq1}\\
    \lambda\sqrt{|S|} &\geq 2(\rho_{\max}\|\tilde\beta-\beta^*\|_2  + 3\phi^s\rho_{\max}(K-1)^{1/2}\|D_\beta\|_2), \label{lambdareq2}
\end{align}
to ensure that $\hat\beta - \beta^*$ is approximately sparse in the sense that $\|\hat\beta-\beta^*\|_1 \leq 5\sqrt{|S|}\|\hat\beta-\beta^*\|_2$, which will be important when we apply the restricted eigenvalue condition later. Now, with such choice of $\lambda$, by the restricted eigenvalue condition in Lemma~\ref{RE}, we have that
\begin{align*}
    &\frac{1}{T}\sum_{t=1}^{T}\sum_{j=1}^K m_{j,\theta}(Y_{t-s}^{t+s})\|(B_j^*-\hat B_j)^\top Y_{t-1}\|_2^2 \\
    &\geq \alpha \|\hat\beta-\beta^*\|_2^2 - \tau_{RE}\|\hat\beta-\beta^*\|_1^2 \\
    &\geq \alpha \|\hat\beta-\beta^*\|_2^2 -25\tau_{RE} |S| \|\hat\beta-\beta^*\|_2^2 \\ 
    &\geq \alpha/2 \|\hat\beta-\beta^*\|_2^2,
\end{align*}
for sufficiently large $T$. Combined with \eqref{preRE}, we have that
\begin{multline*}
    \frac{\alpha}{2} \|\hat\beta-\beta^*\|_2^2 \leq \frac{1}{T}\sum_{t=1}^{T}\sum_{j=1}^K m_{j,\theta}(Y_{t-s}^{t+s})\|(B_j^*-\hat B_j)^\top Y_{t-1}\|_2^2 \\
    \leq \frac{3\lambda}{2}\sqrt{|S|}\|\hat\beta-\beta^*\|_2 + (\rho_{\max}\|\tilde\beta-\beta^*\|_2 + 3\phi^s\rho_{\max}(K-1)^{1/2}\|D_\beta\|_2 )\|\hat\beta-\beta^*\|_2.
\end{multline*}
This in turn implies that
\begin{equation*}
    \|\hat\beta-\beta^*\|_2 \leq \frac{3\lambda\sqrt{|S|}}{\alpha} + \frac{2}{\alpha} \left(\rho_{\max}\|\tilde\beta-\beta^*\|_2 +  3\phi^s\rho_{\max}(K-1)^{1/2}\|D_\beta\|_2 \right).
\end{equation*}

\textbf{Step I-II: estimation error of $p$.} Next, we consider the estimation of the transition probabilities $p_{ij}$. For the ease of notation, we define $m_{i\cdot,\theta} = \sum_{j=1}^K m_{ij,\theta}$ and $w_{i\cdot,\theta} = \sum_{j=1}^K w_{ij,\theta}$. Note that the update has a closed form solution
\begin{equation*}
    \hat p_{ij} = \left\{\frac{1}{T}\sum_{t=1}^T m_{ij,\theta}(Y_{t-s}^{t+s})\right\}\left\{\frac{1}{T}\sum_{t=1}^T m_{i\cdot,\theta}(Y_{t-s}^{t+s})\right\}^{-1},
\end{equation*}
and therefore,
\begin{align*}
    \hat p_{ij} - p_{ij}^* &=  \left\{\frac{1}{T}\sum_{t=1}^T m_{ij,\theta}(Y_{t-s}^{t+s})\right\}\left\{\frac{1}{T}\sum_{t=1}^T m_{i\cdot,\theta}(Y_{t-s}^{t+s})\right\}^{-1} - p_{ij}^* \\
    &= \left\{\frac{1}{T}\sum_{t=1}^T m_{ij,\theta}(Y_{t-s}^{t+s})\right\}\left\{\frac{1}{T}\sum_{t=1}^T m_{i\cdot,\theta}(Y_{t-s}^{t+s})\right\}^{-1} \\
    &\qquad \quad -\left\{E\left[\frac{1}{T}\sum_{t=1}^T m_{ij,\theta}(Y_{t-s}^{t+s})\right]\right\}\left\{E\left[\frac{1}{T}\sum_{t=1}^T m_{i\cdot,\theta}(Y_{t-s}^{t+s})\right]\right\}^{-1} \\
    &\quad + \left\{E\left[\frac{1}{T}\sum_{t=1}^T m_{ij,\theta}(Y_{t-s}^{t+s})\right]\right\}\left\{E\left[\frac{1}{T}\sum_{t=1}^T m_{i\cdot,\theta}(Y_{t-s}^{t+s})\right]\right\}^{-1} \\
    &\qquad \quad - \left\{E\left[\frac{1}{T}\sum_{t=1}^T w_{ij,\theta}(Y_{0}^T)\right]\right\}\left\{E\left[\frac{1}{T}\sum_{t=1}^T w_{i\cdot,\theta}(Y_0^T)\right]\right\}^{-1} \\
    &\quad + \left\{E\left[\frac{1}{T}\sum_{t=1}^T w_{ij,\theta}(Y_{0}^T)\right]\right\}\left\{E\left[\frac{1}{T}\sum_{t=1}^T w_{i\cdot,\theta}(Y_0^T)\right]\right\}^{-1} - p_{ij}^*.
\end{align*}
We denote the three differences as term $p.1$, term $p.2$, and term $p.3$. Similarly, the vector form $\hat p - p^*$ can be decomposed into 3 terms, which we denote as $(\hat p -p^*)_1$, $(\hat p -p^*)_2$ and $(\hat p -p^*)_3$. In particular, the third term $(\hat p -p^*)_3$ corresponds to the difference between the updated parameter value in a population EM algorithm and the true parameter value, that is, $M_p(\theta) - p^*$. Under Assumption~\ref{DB}, term $p.1$ is upper bounded by $\Delta_p$ with high probability, and thus $\|(\hat p - p^*)_1\|_2 \leq K\Delta_p$ with high probability. Finally, by Lemma~\ref{lemma:approximationerror}, $\|(\hat p -p^*)_2\|_2$ is upper bounded by $C_2 \phi^s$ for some constant $C_2$ for sufficiently large $T$. To see this, we note that term $p.2$ can be split into two differences, with the first one being
\begin{multline*}
    \left\{E\left[\frac{1}{T}\sum_{t=1}^T m_{ij,\theta}(Y_{t-s}^{t+s})\right]\right\}\left\{E\left[\frac{1}{T}\sum_{t=1}^T m_{i\cdot,\theta}(Y_{t-s}^{t+s})\right]\right\}^{-1} - \\ \left\{E\left[\frac{1}{T}\sum_{t=1}^T m_{ij,\theta}(Y_{t-s}^{t+s})\right]\right\}\left\{E\left[\frac{1}{T}\sum_{t=1}^T w_{i\cdot,\theta}(Y_0^T)\right]\right\}^{-1},
\end{multline*}
whose absolute value is bounded by
\begin{equation*}
    \left|\left\{E\left[\frac{1}{T}\sum_{t=1}^T m_{i\cdot,\theta}(Y_{t-s}^{t+s})\right]\right\}^{-1} - \left\{E\left[\frac{1}{T}\sum_{t=1}^T w_{i\cdot,\theta}(Y_0^T)\right]\right\}^{-1}\right|,
\end{equation*}
as $m_{ij,\theta}$ is upper bounded by 1 uniformly. The above absolute value is in turn upper bounded by
\begin{equation*}
    \left|\left\{E\left[\frac{1}{T}\sum_{t=1}^T m_{i\cdot,\theta}(Y_{t-s}^{t+s})\right]\right\}^{-1}\right| \left|\left\{E\left[\frac{1}{T}\sum_{t=1}^T w_{i\cdot,\theta}(Y_0^T)\right]\right\}^{-1}\right|E\left[\frac{1}{T}\sum_{t=1}^T\left| m_{i\cdot,\theta}(Y_{t-s}^{t+s}) - w_{i\cdot,\theta}(Y_0^T)\right|\right].
\end{equation*}
The third factor is upper bounded by a constant multiple of $\phi^s$, as the difference between $m_{i\cdot,\theta}$ and $w_{i\cdot,\theta}$ is upper bounded by $\phi^{s-1}$ which can be shown in a similar fashion as in the proof of Lemma~\ref{lemma:approximationerror}. By Assumption~\ref{denomtransition}, the second factor is upper bounded by $\iota^{-1}$, and for sufficiently large value of $T$ and hence $s$, the first factor is upper bounded by $2\iota^{-1}$. The second difference we need to consider is
\begin{equation*}
   E\left[\frac{1}{T}\sum_{t=1}^T \left|m_{ij,\theta}(Y_{t-s}^{t+s})-w_{ij,\theta}(Y_0^T)\right|\right]\left\{E\left[\frac{1}{T}\sum_{t=1}^T w_{i\cdot,\theta}(Y_0^T)\right]\right\}^{-1},
\end{equation*}
which is upper bounded by a constant multiple of $\iota^{-1}\phi^s$ under Assumption~\ref{denomtransition} and by Lemma~\ref{lemma:approximationerror}. Therefore, $\|(\hat p - p^*)_2\|_2$ is upper bounded by $C_2\phi^s$ for some constant $C_2$ that only depends on $\iota$. Combining the upper bounds we have derived, we have an upper bound on the estimation error of the transition probabilities
\begin{equation*}
    \|\hat p - p\|_2 \leq K \Delta_p + C_2\phi^s + \|M_p(\theta) - p^*\|_2.
\end{equation*}

\textbf{Step I-III: estimation error of $\sigma$.} Next we consider the update of $\sigma^2$. Recall that for each regime $1 \leq j \leq K$, the update for $\sigma_j^2$ has a closed-form expression given by
\begin{equation*}
    \hat\sigma_j^2 = \left\{\frac{1}{T}\sum_{t=1}^T \frac{1}{d} m_{j,\theta} \left\|Y_{t} - \left(\textnormal{Id}_d \otimes Y_{t-1}\right)^\top \hat\beta_j\right\|_2^2\right\}\left\{\frac{1}{T}\sum_{t=1}^T m_{j,\theta} \right\}^{-1}.
\end{equation*}
Then, using this expression, we can decompose the estimation error as follows
\begin{align*}
    \hat\sigma_j^2 - (\sigma_j^2)^* &=\left\{\frac{1}{T}\sum_{t=1}^T \frac{1}{d} m_{j,\theta} \left\|Y_{t} - \left(\textnormal{Id}_d \otimes Y_{t-1}\right)^\top \hat\beta_j\right\|_2^2\right\}
    \left\{\frac{1}{T}\sum_{t=1}^T m_{j,\theta} \right\}^{-1} \\
    &\qquad \quad - 
    E\left[ \frac{1}{T}\sum_{t=1}^T \frac{1}{d} m_{j,\theta} \left\|Y_{t} - \left(\textnormal{Id}_d \otimes Y_{t-1}\right)^\top \hat\beta_j\right\|_2^2\right] 
    E\left[\frac{1}{T}\sum_{t=1}^T m_{j,\theta} \right]^{-1} \\
    &\quad + E\left[ \frac{1}{T}\sum_{t=1}^T \frac{1}{d} m_{j,\theta} \left\|Y_{t} - \left(\textnormal{Id}_d \otimes Y_{t-1}\right)^\top \hat\beta_j\right\|_2^2\right] 
    E\left[\frac{1}{T}\sum_{t=1}^T m_{j,\theta} \right]^{-1} \\
    &\qquad \quad -
    E\left[ \frac{1}{T}\sum_{t=1}^T \frac{1}{d} w_{j,\theta} \left\|Y_{t} - \left(\textnormal{Id}_d \otimes Y_{t-1}\right)^\top \hat\beta_j\right\|_2^2\right] 
    E\left[\frac{1}{T}\sum_{t=1}^T w_{j,\theta} \right]^{-1}
    \\
    &\quad + E\left[ \frac{1}{T}\sum_{t=1}^T \frac{1}{d} w_{j,\theta} \left\|Y_{t} - \left(\textnormal{Id}_d \otimes Y_{t-1}\right)^\top \hat\beta_j\right\|_2^2\right] 
    E\left[\frac{1}{T}\sum_{t=1}^T w_{j,\theta} \right]^{-1} \\
    &\qquad \quad - 
    E\left[ \frac{1}{T}\sum_{t=1}^T \frac{1}{d} w_{j,\theta} \left\|Y_{t} - \left(\textnormal{Id}_d \otimes Y_{t-1}\right)^\top \tilde\beta_j\right\|_2^2\right] 
    E\left[\frac{1}{T}\sum_{t=1}^T w_{j,\theta} \right]^{-1}\\
    &\quad E\left[ \frac{1}{T}\sum_{t=1}^T \frac{1}{d} w_{j,\theta} \left\|Y_{t} - \left(\textnormal{Id}_d \otimes Y_{t-1}\right)^\top \tilde\beta_j\right\|_2^2\right] 
    E\left[\frac{1}{T}\sum_{t=1}^T w_{j,\theta} \right]^{-1}
    - (\sigma_j^2)^*. 
\end{align*}
We refer to these four differences as term $\sigma.1$ to term $\sigma.4$ according to the order in the above display. Under Assumption~\ref{DB}, term $\sigma.1$ is upper bounded by $\Delta_\sigma$ with high probability. Term $\sigma.4$ is equivalent to $|(M_\sigma(\theta))_j^2 - (\sigma_j^*)^2|$, which is the error of the population EM estimate. Term $\sigma.2$ involves the approximation error of $m_{j,\theta}$ in both the numerator and denominator. To upper bound this term, we further decompose it into 2 terms, with the first one corresponding to the approximation error of $m_{j,\theta}$ in the numerator,
\begin{multline*}
    E\left[ \frac{1}{T}\sum_{t=1}^T \frac{1}{d} m_{j,\theta} \left\|Y_{t} - \left(\textnormal{Id}_d \otimes Y_{t-1}\right)^\top \hat\beta_j\right\|_2^2\right] 
    E\left[\frac{1}{T}\sum_{t=1}^T w_{j,\theta} \right]^{-1} - \\
    E\left[ \frac{1}{T}\sum_{t=1}^T \frac{1}{d} w_{j,\theta} \left\|Y_{t} - \left(\textnormal{Id}_d \otimes Y_{t-1}\right)^\top \hat\beta_j\right\|_2^2\right] 
    E\left[\frac{1}{T}\sum_{t=1}^T w_{j,\theta} \right]^{-1},
\end{multline*}
whose absolute value is upper bounded, by Assumption~\ref{denomtransition}, by 
\begin{equation*}
    \iota^{-1} E\left[\frac{1}{T}\sum_{t=1}^T \frac{1}{d} \left|m_{j,\theta}(Y_{t-s}^{t+s}) -w_{j,\theta}(Y_0^T)\right| \left\|Y_{t} - \left(\textnormal{Id}_d \otimes Y_{t-1}\right)^\top \hat\beta_j\right\|_2^2\right].
\end{equation*}
By Lemma~\ref{lemma:approximationerror}, this is upper bounded by
\begin{equation*}
    3\iota^{-1} \phi^s E\left[ \frac{1}{d} \left\|Y_{t} - \left(\textnormal{Id}_d \otimes Y_{t-1}\right)^\top \hat\beta_j\right\|_2^2\right] = 3\iota^{-1} \phi^s E\left[ \frac{1}{d}\sum_{i=1}^d \left(Y_{ti} - \hat\beta_{ji}^\top Y_{t-1}\right)^2\right],
\end{equation*}
which is upper bounded by $C_{3,1} \phi^s$ for some constant $C_{3,1}$. To see this, we note that the expectation can be upper bounded by some constant due to the fact that $Y_t$ is a sub-Gaussian random vector and that the estimation error of $\hat\beta$ is upper bounded by some constant as we will show later in the proof. The second term in term $\sigma.2$ corresponds to the approximation error of $m_{j,\theta}$ in the denominator,
\begin{align*}
    &\quad E\left[ \frac{1}{T}\sum_{t=1}^T \frac{1}{d} m_{j,\theta} \left\|Y_{t} - \left(\textnormal{Id}_d \otimes Y_{t-1}\right)^\top \hat\beta_j\right\|_2^2\right] 
    E\left[\frac{1}{T}\sum_{t=1}^T m_{j,\theta} \right]^{-1} \\
    &\qquad - E\left[ \frac{1}{T}\sum_{t=1}^T \frac{1}{d} m_{j,\theta} \left\|Y_{t} - \left(\textnormal{Id}_d \otimes Y_{t-1}\right)^\top \hat\beta_j\right\|_2^2\right] 
    E\left[\frac{1}{T}\sum_{t=1}^T w_{j,\theta} \right]^{-1} \\
    &= E\left[ \frac{1}{T}\sum_{t=1}^T \frac{1}{d} m_{j,\theta} \left\|Y_{t} - \left(\textnormal{Id}_d \otimes Y_{t-1}\right)^\top \hat\beta_j\right\|_2^2\right]
    E\left[\frac{1}{T}\sum_{t=1}^T m_{j,\theta} \right]^{-1}
    E\left[\frac{1}{T}\sum_{t=1}^T w_{j,\theta} \right]^{-1} \\
    &\qquad \times \left\{E\left[\frac{1}{T}\sum_{t=1}^T w_{j,\theta} \right] - E\left[\frac{1}{T}\sum_{t=1}^T m_{j,\theta} \right]\right\}.
\end{align*}
To upper bound this term, we note that
\begin{equation*}
    \left|E\left[\frac{1}{T}\sum_{t=1}^T w_{j,\theta} \right] - E\left[\frac{1}{T}\sum_{t=1}^T m_{j,\theta} \right]\right| \leq 3\phi^s
\end{equation*}
by Lemma~\ref{lemma:approximationerror}. Combined with the fact that $E[\sum_{t=1}^T w_{j,\theta}/T] \geq \iota$ from Assumption~\ref{denomtransition}, we have $E[\sum_{t=1}^T m_{j,\theta}/T] \geq \iota/2$ for sufficiently large $T$ (hence sufficiently large $s$). Therefore, this second term in term $\sigma.2$ is upper bounded by $C_{3,2}\phi^s$ for some constant $C_{3,2}$. Finally, we study term $\sigma.3$ which corresponds to the estimation error resulting from the estimation error of $\beta$, which is
\begin{equation*}
    E\left[\frac{1}{T}\sum_{t=1}^T w_{j,\theta}\right]^{-1}
    E\left[\frac{1}{T}\sum_{t=1}^T \frac{1}{d} w_{j,\theta}\sum_{i=1}^d \left\{ \left(Y_{ti} - \hat\beta_{ji}^\top Y_{t-1}\right)^2 -\left(Y_{ti} - \tilde\beta_{ji}^\top Y_{t-1}\right)^2 \right\}\right].
\end{equation*}
The first factor is upper bounded by $\iota^{-1}$ and the second factor is equivalent to
\begin{multline*}
    \frac{2}{d}\sum_{i=1}^d E\left[\frac{1}{T}\sum_{t=1}^T w_{j,\theta} \left(Y_{ti} - \tilde\beta_{ji}^\top Y_{t-1}\right)Y_{t-1}^\top(\hat\beta_{ji}-\tilde\beta_{ji}) \right]
    + \\
    \frac{1}{d}\sum_{i=1}^d E\left[\frac{1}{T}\sum_{t=1}^T  w_{j,\theta} (\hat\beta_{ji} - \tilde\beta_{ji})^\top Y_{t-1}Y_{t-1}^\top (\hat\beta_{ji} - \tilde\beta_{ji})\right],
\end{multline*}
and by the definition of $\tilde\beta$ the first term is 0. The second term is upper bounded by $\rho_{\max}\|\hat\beta - \tilde\beta\|_2^2/d$ by Assumption~\ref{cond:minmaxeigenvalue}, which is in turn upper bounded by $2\rho_{\max}(\|\hat\beta - \beta^*\|_2^2 + \|\tilde\beta - \beta^*\|_2^2)/d$ by triangle inequality. Combining all these upper bounds, we have that for the subvector corresponding to $\sigma_j^2$ across all regimes, 
\begin{align*}
    &\quad \|\hat\sigma^2 - (\sigma^*)^2\|_2 \\
    &\leq \sqrt{K}\Delta_\sigma + \|(M_\sigma(\theta))^2 - (\sigma^*)^2\|_2 + C_3\phi^s + \frac{2\sqrt{K}\rho_{\max}}{d}\left(\|\hat\beta - \beta^*\|_2^2 + \|\tilde\beta - \beta^*\|_2^2 \right) \\
    &\leq \sqrt{K}\Delta_\sigma + \|(M_\sigma(\theta))^2 - (\sigma^*)^2\|_2 + C_3\phi^s + \frac{2\sqrt{K}\rho_{\max}}{d}\|\hat\beta - \beta^*\|_2^2 + \frac{2\sqrt{K}\rho_{\max}r}{d}\|\tilde\beta - \beta^*\|_2,
\end{align*}
where the third line follows as $\|M(\theta) - \theta^*\| \leq \|\theta - \theta^*\| \leq r$ and $C_3$ absorbs $\sqrt{K}$.

\textbf{Step I-IV: estimation error of $\theta$ and requirements on $\lambda$.} Combining the estimation error for $\beta$, $p$ and $\sigma$, we get that
\small
\begin{align*}
    &\quad \|\hat\theta - \theta^*\|_2 \\ 
    &= \left\{\|\hat\beta - \beta^*\|_2^2 + \|\hat p - p^*\|_2^2 + \|\hat\sigma^2 - (\sigma^*)^2\|_2^2\right\}^{1/2} \\
    &\leq \Bigg\{ \left[\frac{3\lambda\sqrt{|S|}}{\alpha} + \frac{2}{\alpha} \left(\rho_{\max}\|\tilde\beta-\beta^*\|_2 +  3\phi^s\rho_{\max}(K-1)^{1/2}\|D_\beta\|_2 \right)\right]^2 \\
    &\quad + \left[K \Delta_p + C_2\phi^s + \|M_p(\theta) - p^*\|_2\right]^2 \\
    &\quad + \left[\sqrt{K}\Delta_\sigma + \|(M_\sigma(\theta))^2 - (\sigma^*)^2\|_2 + C_3\phi^s + \frac{2\sqrt{K}\rho_{\max}}{d}\|\hat\beta - \beta^*\|_2^2 + \frac{2\sqrt{K}\rho_{\max}r}{d}\|\tilde\beta - \beta^*\|_2\right]^2\Bigg\}^{1/2} \\
    &\leq \frac{3\lambda\sqrt{|S|}}{\alpha} + \frac{2}{\alpha} 3\phi^s\rho_{\max}(K-1)^{1/2}\|D_\beta\|_2 + K \Delta_p + C_2\phi^s + \sqrt{K} \Delta_\sigma + C_3\phi^s + \frac{2\sqrt{K}\rho_{\max}}{d}\|\hat\beta - \beta^*\|_2^2 \\
    &\quad + \left\{\left(\frac{2\rho_{\max}}{\alpha}\right)^2 \|\tilde\beta - \beta^*\|_2^2 + \|M_p(\theta) - p^*\|_2^2 + 2 \|(M_\sigma(\theta))^2 - (\sigma^*)^2\|_2^2 + 2\left(\frac{2\sqrt{K}\rho_{\max}r}{d}\right)^2\|\tilde\beta - \beta^*\|_2^2\right\}^{1/2} \\
    &\leq \frac{3\lambda\sqrt{|S|}}{\alpha} + \frac{2}{\alpha} 3\phi^s\rho_{\max}(K-1)^{1/2}\|D_\beta\|_2 + K \Delta_p + C_2\phi^s + \sqrt{K} \Delta_\sigma + C_3\phi^s + \frac{2\sqrt{K} \rho_{\max}}{d}\|\hat\beta - \beta^*\|_2^2 \\
    &\quad + \max\left\{\left(\frac{2\rho_{\max}}{\alpha}\right)^2 + 2K\left(\frac{2\rho_{\max}r}{d}\right)^2, 2\right\}^{1/2} \left\{\|\tilde\beta - \beta^*\|_2^2 + \|M_p(\theta) - p^*\|_2^2 + \|(M_\sigma(\theta))^2 - (\sigma^*)^2\|^2 \right\}^{1/2}  \\
    &= \frac{3\lambda\sqrt{|S|}}{\alpha} + \frac{2}{\alpha} 3\phi^s\rho_{\max}(K-1)^{1/2}\|D_\beta\|_2 + K \Delta_p + C_2\phi^s + \sqrt{K} \Delta_\sigma + C_3\phi^s + \frac{2\sqrt{K}\rho_{\max}}{d}\|\hat\beta - \beta^*\|_2^2 \\
    &\quad + \max\left\{\left(\frac{2\rho_{\max}}{\alpha}\right)^2 + 2K \left(\frac{2\rho_{\max}r}{d}\right)^2, 2\right\}^{1/2} \|M(\theta) - \theta^*\|_2 \\
    &\leq \frac{3\lambda\sqrt{|S|}}{\alpha} + \frac{6}{\alpha} \phi^s\rho_{\max}(K-1)^{1/2}\|D_\beta\|_2 + K \Delta_p + C_2\phi^s + \sqrt{K} \Delta_\sigma + C_3\phi^s + \frac{2\sqrt{K}\rho_{\max}}{d}\|\hat\beta - \beta^*\|_2^2 \\
    &\quad + \eta\kappa \|\theta - \theta^*\|_2,
\end{align*}
\normalsize
where the last line follows from Lemma~\ref{populationcontraction} and the definition of $\eta$. Now suppose we choose $\lambda$ in a way that
\begin{equation}\label{lambdareq3}
    \frac{\lambda \sqrt{|S|}}{\alpha} \geq \frac{6}{\alpha} \phi^s\rho_{\max}(K-1)^{1/2}\|D_\beta\|_2 + K \Delta_p + C_2\phi^s + \sqrt{K} \Delta_\sigma + C_3\phi^s  + \eta\kappa \|\theta - \theta^*\|_2,
\end{equation}
we would have the following upper bound on the estimation error,
\begin{equation*}
    \|\hat\theta - \theta^*\|_2 \leq \frac{4\lambda\sqrt{|S|}}{\alpha} + \frac{2\sqrt{K}\rho_{\max}}{d}\|\hat\beta - \beta^*\|_2^2.
\end{equation*}

Note that the requirement on $\lambda$ in \eqref{lambdareq3} is stronger than the one in \eqref{lambdareq2}, because
\begin{align*}
    &\quad \frac{6}{\alpha} \phi^s\rho_{\max}(K-1)^{1/2}\|D_\beta\|_2 + K \Delta_p + C_2\phi^s + \sqrt{K} \Delta_\sigma + C_3\phi^s  + \eta\kappa \|\theta - \theta^*\|_2 \\
    &\geq \frac{6}{\alpha} \phi^s\rho_{\max}(K-1)^{1/2}\|D_\beta\|_2 + \eta \|M(\theta) - \theta^*\|_2 \\
    &\geq \frac{6}{\alpha} \phi^s\rho_{\max}(K-1)^{1/2}\|D_\beta\|_2 + \eta \|\tilde\beta-\beta^*\|_2 \\
    &\geq \frac{6}{\alpha}\phi^s\rho_{\max}(K-1)^{1/2}\|D_\beta\|_2 + \frac{2}{\alpha}\rho_{\max}\|\tilde\beta-\beta^*\|_2.
\end{align*}
Therefore, if $\lambda$ is chosen such that \eqref{lambdareq3} is satisfied, this $\lambda$ will also satisfy \eqref{lambdareq2}. Recall that we have the following upper bound for the estimation error of $\beta$ when $\lambda$ satisfies \eqref{lambdareq1} and \eqref{lambdareq2},
\begin{equation*}
    \|\hat\beta-\beta^*\|_2 \leq \frac{3\lambda\sqrt{|S|}}{\alpha} + \frac{2}{\alpha} \left(\rho_{\max}\|\tilde\beta-\beta^*\|_2 +  3\phi^s\rho_{\max}(K-1)^{1/2}\|D_\beta\|_2 \right),
\end{equation*}
and therefore when $\lambda$ satisfies \eqref{lambdareq1} and \eqref{lambdareq3},
\begin{equation*}
    \|\hat\beta-\beta^*\|_2 \leq \frac{4\lambda\sqrt{|S|}}{\alpha}.
\end{equation*}
Consequently, we have the following upper bound for the overall estimation error
\begin{equation*}
    \|\hat\theta - \theta^*\|_2 \leq \frac{4\lambda\sqrt{|S|}}{\alpha} + \frac{2\sqrt{K}\rho_{\max}}{d}\left(\frac{4\lambda\sqrt{|S|}}{\alpha}\right)^2.
\end{equation*}
Finally, suppose that $\lambda$ is such that
\begin{equation}\label{lambdareq4}
    \frac{4\lambda\sqrt{|S|}}{\alpha} \leq 4.1\eta\kappa r,
\end{equation}
the estimation error is upper bounded by
\begin{equation}\label{onesteperror}
    \|\hat\theta - \theta^*\|_2 \leq \left(1 + \frac{8.2\sqrt{K}\rho_{\max}\eta\kappa r}{d}\right) \frac{4\lambda\sqrt{|S|}}{\alpha}.
\end{equation}
We note again that to achieve the error bound in \eqref{onesteperror}, $\lambda$ needs to be chosen such that
\begin{align*}
    \frac{4\lambda\sqrt{|S|}}{\alpha} &\leq 4.1\eta\kappa r, \\
    \lambda &\geq 2\Delta + 2C_1\phi^s, \\
    \frac{\lambda \sqrt{|S|}}{\alpha} &\geq \frac{6}{\alpha} \phi^s\rho_{\max}(K-1)^{1/2}\|D_\beta\|_2 + K \Delta_p + C_2\phi^s + \sqrt{K} \Delta_\sigma + C_3\phi^s  + \eta\kappa \|\theta - \theta^*\|_2,
\end{align*}
which are the requirements in \eqref{lambdareq1}, \eqref{lambdareq3} and \eqref{lambdareq4}.

\textbf{Step II: induction.} Now we study our specified choice of $\lambda$. We use induction to show that in each iteration, the requirements for $\lambda$ are met with our choice, and hence the estimation error in each iteration is upper bounded according to \eqref{onesteperror}. To start, recall that we choose $\lambda$ such that
\begin{align*}
    \lambda^{(q)} &= \tau^q \frac{\alpha}{4\sqrt{|S|}}\left(1 + \frac{8.2\sqrt{K}\rho_{\max}\eta\kappa r}{d}\right)^{-1}\|\theta^{(0)} - \theta^*\|_2 + \frac{1-\tau^q}{1-\tau} \max\Bigg\{2\Delta + 2C_1\phi^s, \\
    &\qquad \frac{\alpha}{\sqrt{|S|}}\left(\frac{6}{\alpha} \phi^s\rho_{\max}(K-1)^{1/2}\|D_\beta\|_2 + K \Delta_p + C_2\phi^s + \sqrt{K}\Delta_\sigma + C_3\phi^s  \right)\Bigg\},
\end{align*}
where $\tau = 4\kappa \eta \left(1 + \frac{8.2\sqrt{K}\rho_{\max}\eta\kappa r}{d}\right)$. 

\textbf{Step II-I: $q=1$.} When $q = 1$, we have
\begin{align*}
    \lambda^{(1)} &\geq \frac{1-\tau}{1-\tau} (2\Delta + 2C_1\phi^s) = 2\Delta + 2C_1\phi^,
\end{align*}
and thus \eqref{lambdareq1} is met. Moreover,
\begin{align*}
    \lambda^{(1)} &\geq \frac{1-\tau}{1-\tau}\frac{\alpha}{\sqrt{|S|}}\left(\frac{6}{\alpha} \phi^s\rho_{\max}(K-1)^{1/2}\|D_\beta\|_2 + K \Delta_p + C_2\phi^s + \sqrt{K}\Delta_\sigma + C_3\phi^s  \right) \\
    &\quad + \tau \frac{\alpha}{4\sqrt{|S|}}\left(1 + \frac{8.2\sqrt{K}\rho_{\max}\eta\kappa r}{d}\right)^{-1}\|\theta^{(0)} - \theta^*\|_2 \\
    &\geq \frac{\alpha}{\sqrt{|S|}}\left\{\left(\frac{6}{\alpha} \phi^s\rho_{\max}(K-1)^{1/2}\|D_\beta\|_2 + K \Delta_p + C_2\phi^s + \sqrt{K}\Delta_\sigma + C_3\phi^s  \right) + \kappa \eta \|\theta^{(0)} - \theta^*\|_2\right\},
\end{align*}
and therefore \eqref{lambdareq3} is met. Finally, we have that
\begin{align*}
    &\quad \frac{4\lambda^{(1)}\sqrt{|S|}}{\alpha} \\
    &= \frac{4\sqrt{|S|}}{\alpha}\max\left\{2\Delta + 2C_1\phi^s, \frac{\alpha}{\sqrt{|S|}}\left(\frac{6}{\alpha} \phi^s\rho_{\max}(K-1)^{1/2}\|D_\beta\|_2 + K \Delta_p + C_2\phi^s + \sqrt{K}\Delta_\sigma + C_3\phi^s  \right)\right\} \\ 
    &\quad + \tau \left(1 + \frac{8.2\sqrt{K}\rho_{\max}\eta\kappa r}{d}\right)^{-1}\|\theta^{(0)} - \theta^*\|_2 \\
    &\leq \frac{4\sqrt{|S|}}{\alpha}\max\left\{2\Delta + 2C_1\phi^s, \frac{\alpha}{\sqrt{|S|}}\left(\frac{6}{\alpha} \phi^s\rho_{\max}(K-1)^{1/2}\|D_\beta\|_2 + K \Delta_p + C_2\phi^s + \sqrt{K}\Delta_\sigma + C_3\phi^s  \right)\right\} \\ 
    &\quad + 4\eta\kappa r,
\end{align*}
and therefore \eqref{lambdareq4} is met if 
\begin{equation*}
    \frac{8\sqrt{|S|}}{\alpha}(\Delta + C_1\phi^s) \leq 0.1 \eta\kappa r,
\end{equation*}
and
\begin{equation*}
    4\left(\frac{6}{\alpha} \phi^s\rho_{\max}(K-1)^{1/2}\|D_\beta\|_2 + K \Delta_p + C_2\phi^s + \sqrt{K}\Delta_\sigma + C_3\phi^s  \right) \leq 0.1 \eta\kappa r.
\end{equation*}
With the choice that $s = -\log T/(2\log\phi)$, we have $\phi^s = T^{-1/2}$, and therefore $\phi^s\sqrt{|S|}$ will be approaching 0 if $|S| = o(T)$. Moreover, we have assumed that $\sqrt{|S|}\Delta$, $\Delta_p$ and $\Delta_\sigma$ are all $o(1)$ when $T$ approaches infinity. Therefore, the quantities $\frac{8\sqrt{|S|}}{\alpha}(\Delta + C_1\phi^s)$ and $ 4(\frac{6}{\alpha} \phi^s\rho_{\max}(K-1)^{1/2}\|D_\beta\|_2 + K \Delta_p + C_2\phi^s + \sqrt{K}\Delta_\sigma + C_3\phi^s)$ can be made arbitrarily small when $T$ is sufficiently large, and in particular, they will be smaller than $0.1\eta\kappa r$ for sufficiently large $T$.

\textbf{Step II-II: $q>1$.} Next we show that if $\lambda^{(q)}$ satisfies \eqref{lambdareq1}, \eqref{lambdareq3} and \eqref{lambdareq4} so that
\begin{equation*}
    \|\theta^{(q)} - \theta^*\|_2 \leq \left(1 + \frac{8.2\sqrt{K}\rho_{\max}\eta\kappa r}{d}\right) \frac{4\lambda^{(q)}\sqrt{|S|}}{\alpha},
\end{equation*}
then $\lambda^{(q+1)}$ satisfies \eqref{lambdareq1}, \eqref{lambdareq3} and \eqref{lambdareq4} so that
\begin{equation*}
    \|\theta^{(q+1)} - \theta^*\|_2 \leq \left(1 + \frac{8.2\sqrt{K}\rho_{\max}\eta\kappa r}{d}\right) \frac{4\lambda^{(q+1)}\sqrt{|S|}}{\alpha}.
\end{equation*}
To start, we note that in fact \eqref{lambdareq1} is always met. Indeed, as $\tau < 1$, we have
\begin{equation*}
    \lambda^{(q+1)} \geq \frac{1-\tau^{(q+1)}}{1-\tau} (2\Delta + 2C_1\phi^s) \geq 2\Delta + 2C_1\phi^s.
\end{equation*}
Also, \eqref{lambdareq4} is always met when $T$ is sufficiently large. Indeed, one term in $4\lambda^{(q+1)}\sqrt{|S|}/\alpha$ is upper bounded by
\begin{equation*}
    \tau^{q+1}\left(1+\frac{8.2\sqrt{K}\rho_{\max}\eta\kappa r}{d}\right)^{-1} \|\theta^{(0)} - \theta^*\|_2 \leq 4\eta\kappa r \tau^q \leq 4\eta\kappa r,
\end{equation*}
as $\tau < 1$. Thus, \eqref{lambdareq4} is satisfied if 
\begin{multline*}
    \frac{4\sqrt{|S|}}{\alpha}\frac{1-\tau^{q+1}}{1-\tau} \max\Bigg\{2\Delta + 2C_1\phi^s, \\
    \frac{\alpha}{\sqrt{|S|}}\left(\frac{6}{\alpha} \phi^s\rho_{\max}(K-1)^{1/2}\|D_\beta\|_2 + K \Delta_p + C_2\phi^s + \sqrt{K}\Delta_\sigma + C_3\phi^s  \right)\Bigg\} \leq 0.1\eta\kappa r,
\end{multline*}
which would be the case if 
\begin{multline*}
    \frac{4\sqrt{|S|}}{\alpha}\frac{1}{1-\tau} \max\Bigg\{2\Delta + 2C_1\phi^s, \\
    \frac{\alpha}{\sqrt{|S|}}\left(\frac{6}{\alpha} \phi^s\rho_{\max}(K-1)^{1/2}\|D_\beta\|_2 + K \Delta_p + C_2\phi^s + \sqrt{K}\Delta_\sigma + C_3\phi^s  \right)\Bigg\} \leq 0.1\eta\kappa r,
\end{multline*}
which is indeed the case for sufficiently large $T$. The argument is essentially the same as establishing \eqref{lambdareq4} for the case $q=1$. Therefore, it remains to establish \eqref{lambdareq3}, which in this case is equivalent to
\begin{equation*}
    \frac{\alpha}{\sqrt{|S|}}\left\{\frac{6}{\alpha} \phi^s\rho_{\max}(K-1)^{1/2}\|D_\beta\|_2 + K \Delta_p + C_2\phi^s + \sqrt{K} \Delta_\sigma + C_3\phi^s  + \eta\kappa \|\theta^{(q)} - \theta^*\|_2\right\} \leq \lambda^{(q+1)}.
\end{equation*}
To show this, we start with the left-hand side quantity, and apply the bound we have for $\|\theta^{(q)} - \theta^*\|_2$,
\begin{align*}
    &\quad \frac{\alpha}{\sqrt{|S|}}\left\{\frac{6}{\alpha} \phi^s\rho_{\max}(K-1)^{1/2}\|D_\beta\|_2 + K \Delta_p + C_2\phi^s + \sqrt{K} \Delta_\sigma + C_3\phi^s  + \eta\kappa \|\theta^{(q)} - \theta^*\|_2\right\} \\
    &\leq \frac{\alpha}{\sqrt{|S|}}\left\{\frac{6}{\alpha} \phi^s\rho_{\max}(K-1)^{1/2}\|D_\beta\|_2 + K \Delta_p + C_2\phi^s + \sqrt{K} \Delta_\sigma + C_3\phi^s  \right\}  + \tau \lambda^{(q)} \\
    &\leq \max\left\{2\Delta + 2C_1\phi^s, \frac{\alpha}{\sqrt{|S|}}\left(\frac{6}{\alpha} \phi^s\rho_{\max}(K-1)^{1/2}\|D_\beta\|_2 + K \Delta_p + C_2\phi^s + \sqrt{K} \Delta_\sigma + C_3\phi^s  \right)\right\} \\
    &\quad + \tau \frac{1-\tau^q}{1-\tau} \max\left\{2\Delta + 2C_1\phi^s, \frac{\alpha}{\sqrt{|S|}}\left(\frac{6}{\alpha} \phi^s\rho_{\max}(K-1)^{1/2}\|D_\beta\|_2 + K \Delta_p + C_2\phi^s + \sqrt{K} \Delta_\sigma + C_3\phi^s  \right)\right\} \\
    &\quad + \tau^{q+1} \frac{\alpha}{4\sqrt{|S|}}\left(1 + \frac{8.2\sqrt{K}\rho_{\max}\eta\kappa r}{d}\right)^{-1}\|\theta^{(0)} - \theta^*\|_2 \\
    &\leq \frac{1-\tau^{q+1}}{1-\tau} \max\left\{2\Delta + 2C_1\phi^s, \frac{\alpha}{\sqrt{|S|}}\left(\frac{6}{\alpha} \phi^s\rho_{\max}(K-1)^{1/2}\|D_\beta\|_2 + K \Delta_p + C_2\phi^s + \sqrt{K} \Delta_\sigma + C_3\phi^s  \right)\right\}\\
    &\quad + \tau^{q+1} \frac{\alpha}{4\sqrt{|S|}}\left(1 + \frac{8.2\sqrt{K}\rho_{\max}\eta\kappa r}{d}\right)^{-1}\|\theta^{(0)} - \theta^*\|_2 \\
    &= \lambda^{(q+1)}.
\end{align*}
where we have made use of the explicit expression of our choice of $\lambda^{(q)}$ to obtain the second inequality and that of $\lambda^{(q+1)}$ in the final equality. 

We have shown that $\lambda^{(q+1)}$ satisfies \eqref{lambdareq1}, \eqref{lambdareq3} and \eqref{lambdareq4}. Thus, we have
\begin{equation*}
    \|\theta^{(q+1)} - \theta^*\|_2 \leq \left(1 + \frac{8.2\sqrt{K}\rho_{\max}\eta\kappa r}{d}\right) \frac{4\lambda^{(q+1)}\sqrt{|S|}}{\alpha},
\end{equation*}
and using the explicit expression for $\lambda^{(q+1)}$, we have
\begin{align*}
    \|\theta^{(q+1)} - \theta^*\|_2 &\leq \left(1 + \frac{8.2\sqrt{K}\rho_{\max}\eta\kappa r}{d}\right) \frac{4\sqrt{|S|}}{\alpha}\frac{1-\tau^{q+1}}{1-\tau} \max\Bigg\{2\Delta + 2C_1\phi^s, \\
    &\qquad \quad \frac{\alpha}{\sqrt{|S|}}\left(\frac{6}{\alpha} \phi^s\rho_{\max}(K-1)^{1/2}\|D_\beta\|_2 + K \Delta_p + C_2\phi^s + \sqrt{K} \Delta_\sigma + C_3\phi^s  \right)\Bigg\} \\
    &\quad +  \tau^{q+1} \|\theta^{(0)} - \theta^*\|_2.
\end{align*}
This shows that for sufficiently large $T$, $\theta^{(q)}$ indeed remains in $\mathcal{B}(r;\theta^*)$ for all $q$, as the first term in the upper bound above can be made arbitrarily small with sufficiently large $T$. This finishes the induction.

\end{proof}

\begin{proof}[Proof of Proposition~\ref{errororder}]
For the ease of notation, define the following functions
\begin{align*}
    h^{ijk}(Y_{t-1}^t) &= Y_{t-1,k}(Y_{ti}-\beta_{ji}^{*\top}Y_{t-1}); \\
    f_\theta^{ijk}(Y_{t-s}^{t+s}) &= Y_{t-1,k}(Y_{ti}-\beta_{ji}^{*\top}Y_{t-1})m_{j,\theta}(Y_{t-s}^{t+s}) = h^{ijk}(Y_{t-1}^t)m_{j,\theta}(Y_{t-s}^{t+s}).
\end{align*}
Define the set of parameter values $\Theta$ such that
\begin{equation*}
    \Theta = \{\theta = (\beta^\top,p^\top,\sigma^\top)^\top: \|\theta - \theta^*\|_2 \leq r, \|(\beta - \beta^*)_{S^C}\|_1 \leq 4\sqrt{|S|}\|\beta-\beta^*\|_2\}
\end{equation*}
Define the function class $\mathcal{G}$ as
\begin{equation*}
    \mathcal{G} = \bigcup_{i,j,k} \left\{f_{\theta}^{ijk}: \theta \in \Theta \right\}.
\end{equation*}
For
\begin{equation*}
    \delta = C \sqrt{\frac{|S|l(T,d)(\log T)^3 (\log K + \log d) + (\log T)^4}{T}},
\end{equation*}
we will show that as $T \rightarrow 0$,
\begin{equation*}
    P\left(\max_{i,j,k}\sup_{\theta\in\Theta}\left|\frac{1}{T}\sum_{t=1}^Tf_\theta^{ijk}(Y_{t-s}^{t+s}) -E\left[f_\theta^{ijk}(Y_{t-s}^{t+s})\right]\right| \geq \delta \right) \rightarrow 0,
\end{equation*}
that is,
\begin{equation}\label{uniformConcentration}
    P\left( \sup_{g \in \mathcal{G}} \left|\frac{1}{T}\sum_{t=1}^T g(Y_{t-s}^{t+s}) - E\left[g(Y_{t-s}^{t+s})\right]\right| \geq \delta \right) \leq u(T,d),
\end{equation}
for some function $u(T,d)$ that converges to 0 as $T$ approaches $\infty$. 

We shall prove this claim in three steps. In step I, we control the tail behavior of the random variable $h^{ijk}(Y_{t-1}^t)$, which will be useful for the concentration result we establish later. In step II, we establish a uniform concentration result over the function class $\mathcal{G}$ using entropy argument, which holds for independent and identically distributed observations. In step III, we establish a uniform concentration result for the original time series data that is $\beta$-mixing, based on the concentration result for i.i.d. data from step II.

\textbf{Step I: Control the tail of $h^{ijk}$}. 
To start, we recall that under the stationary distribution, $Y_t$ is a sub-Gaussian random vector and we define $$K_Y = \sup_{v \in \mathbb{R}^d,\|v\|=1}\sup_{q \geq 1}(E|v^\top Y_t|^q)^{1/q}q^{-1/2}.$$ Therefore,
\begin{equation}\label{subgaussianvector2}
    (E|v^\top Y_t|^q)^{1/q} \leq K_Yq^{1/2}, \quad \forall q \geq 1, \ \forall v\in\mathbb{R}^d, \|v\|=1.
\end{equation}
Next, consider the random variable $h^{ijk}(Y_{t-1}^t) = Y_{t-1,k}(Y_{ti}-\beta_{ji}^{*\top}Y_{t-1})$, and fix an integer $q \geq 1$.
\begin{align*}
    \left(E|h^{ijk}(Y_{t-1}^t)|^q\right)^{1/q} &= \left(E\left|Y_{t-1,k}(Y_{ti}-\beta_{ji}^{*\top}Y_{t-1})\right|^q\right)^{1/q} \\
    &\leq \left(E\left|Y_{t-1,k}Y_{ti}\right|^q\right)^{1/q} + \left(E\left|Y_{t-1,k}\beta_{ji}^{*\top}Y_{t-1}\right|^q\right)^{1/q} \\
    &\leq \left(E|Y_{t-1,k}|^{2q} E|Y_{t,i}|^{2q}\right)^{1/(2q)} + \left(E|Y_{t-1,k}|^{2q} E|\beta_{ji}^{*\top}Y_{t-1}|^{2q}\right)^{1/(2q)} \\
    &= \left(E|Y_{t-1,k}|^{2q}\right)^{1/(2q)} \left(E|Y_{t,i}|^{2q}\right)^{1/(2q)} + \left(E|Y_{t-1,k}|^{2q}\right)^{1/(2q)} \left(E|\beta_{ji}^{*\top}Y_{t-1}|^{2q}\right)^{1/(2q)} \\
    &\leq K_Y (2q)^{1/2}K_Y (2q)^{1/2} + K_Y (2q)^{1/2}K_Y (2q)^{1/2} \\
    &= 4K_Y^2q,
\end{align*}
where the second line follows from triangle inequality, the third line follows from Cauchy-Schwarz inequality, and the fifth line follows from \eqref{subgaussianvector2}. In particular, the argument above holds for any $q \geq 1$ and for any $(i,j,k)$, and therefore, $(E|h^{ijk}(Y_{t-1}^t)|^q)^{1/q} \leq 4K_Y^2q$ for all $q \geq 1$ and $(i,j,k)$. This implies that $h^{ijk}(Y_{t-1}^t)$ is sub-exponential and sub-weibull(1) for any $(i,j,k)$. Moreover, when setting $q = 2$, we have that
\begin{equation*}
    E\left[\left\{h^{ijk}(Y_{t-1}^t)\right\}^2\right] \leq 64 K_Y^4, \quad \forall i,j,k.
\end{equation*}

\textbf{Step II: uniform concentration for i.i.d data.} For any fixed constant $\tilde{c}$, let $N = T/\{\tilde{c}\log T\}$. Let $\{\tilde{Y}_{n-s}^{n+s}\}_{n=1}^N$ be an i.i.d. sample where the marginal distribution of $\tilde{Y}_{n-s}^{n+s}$ is the same as the marginal distribution of $Y_{t-s}^{t+s}$. For the ease of notation, let $X_n = \tilde{Y}_{n-s}^{n+s}$. Note that the sample size of this i.i.d. sample is smaller than the sample size of the original time series by a log factor, and the reason for this shall become clear in step III. To establish the desired upper bound on the tail probability with time series data, we first derive an upper bound for the following analogous tail probability with i.i.d data:
\begin{equation}\label{iidtruncate}
    P\left(\sup_{g \in \mathcal{G}}\left|\frac{1}{N}\sum_{n=1}^N g(X_n) - E[g(X_n)]\right| \geq \delta \right).
\end{equation}

\textbf{Step II-I: symmetrization.} We start by a symmetrization argument.
\begin{theorem}[Corollary 3.4 in \citet{geer2000empirical}, symmetrization]\label{symmetrization}
Suppose $\sup_{g\in\mathcal{G}}\|g\| \leq R$. Then for $N \geq 72R^2/\delta^2$,
\begin{equation*}
    P\left(\sup_{g\in\mathcal{G}}\left|\frac{1}{N}\sum_{n=1}^N g(X_n) - E[g(X_n)]\right| \geq \delta \right) \leq 4P\left( \sup_{g\in\mathcal{G}}\left|\frac{1}{N}\sum_{n=1}^N W_ng(X_n)\right| \geq \delta/4 \right),
\end{equation*}
where $\{W_1,\ldots,W_N\}$ is a sequence of i.i.d. Rademacher random variables, that is, $P(W_n=1) = P(W_n=-1)=1/2$, independent of $\{X_1,\ldots,X_N\}$.
\end{theorem}
To apply this theorem, we need to find the value of $R$ for the function class $\mathcal{G}$ we have defined. Note that fix a function $g \in \mathcal{G}$, there exist $(i,j,k)$ and $\theta$ such that,
\begin{align*}
    \|g\|^2 = E[g(X_n)^2] &= E\left[\left\{\tilde{Y}_{n-1,k}\left(\tilde{Y}_{n,i} - \beta_{ji}^{*\top}\tilde{Y}_{n-1}\right)m_{j,\theta}(\tilde{Y}_{n-s}^{n+s})\right\}^2\right] \\
    &\leq E\left[\left\{\tilde{Y}_{n-1,k}\left(\tilde{Y}_{n,i} - \beta_{ji}^{*\top}\tilde{Y}_{n-1}\right)\right\}^2\right] \\
    &= E\left[\left\{h^{ijk}(\tilde{Y}_{n-1}^n)\right\}^2\right] \\
    &\leq 64 K_Y^4,
\end{align*}
where the second line follows as $m_{j,\theta}$ is upper bounded by 1 for any $\theta$. The fourth line follows from step I and the fact that the marginal distribution of $\tilde{Y}_{n-1}^n$ is the same as the marginal distribution of $Y_{t-1}^t$. The above display holds for any $(i,j,k)$ and $\theta$, and hence holds for any $g \in \mathcal{G}$. Therefore, $\sup_{g \in \mathcal{G}}\|g\| \leq 8K_Y^2$, and we can take $R = 8K_Y^2$ which is a constant. To apply Theoerem~\ref{symmetrization} we only need to check $N\delta^2 \geq 8R^2$. Given the definition of $\delta$, we have that $N\delta^2$ approaches infinity as $T$ approaches infinity, and hence it will be larger than $8R^2$ for sufficiently large $T$.

\textbf{Step II-II: Control the empirical norm of functions in $\mathcal{G}$}. Given step II-I, it now suffices to control the probability 
\begin{equation*}
    P\left( \sup_{g\in\mathcal{G}}\left|\frac{1}{N}\sum_{n=1}^N W_ng(X_n)\right| \geq \delta/4 \right).
\end{equation*}
To do this, we condition on $X_1,X_2,\ldots,X_N$. Define the event $\mathcal{A}$ such that $I_{\mathcal{A}}\{X_1,\ldots,X_N\} = 1$ if and only if
\begin{equation*}
    \sup_{g \in \mathcal{G}}\frac{1}{N}\sum_{n=1}^N g(X_n)^2 \leq 64K_Y^4 + 1
\end{equation*}
We now study the probability of the event $\mathcal{A}$. Define the function class $\mathcal{G}^{ijk} = \{f^{ijk}_\theta: \theta \in \Theta\}$, and note that $\mathcal{G} = \cup_{i,j,k} \mathcal{G}^{ijk}$. 
\begin{align*}
    \sup_{g \in \mathcal{G}^{ijk}}\frac{1}{N}\sum_{n=1}^N g(X_n)^2 &= \sup_{\theta \in \Theta}\frac{1}{N}\sum_{n=1}^N f^{ijk}_\theta(X_n)^2 \\
    &= \sup_{\theta \in \Theta} \frac{1}{N}\sum_{n=1}^N \left\{\tilde{Y}_{n-1,k}\left(\tilde{Y}_{n,i} - \beta_{ji}^{*\top}\tilde{Y}_{n-1}\right)m_{j,\theta}(\tilde{Y}_{n-s}^{n+s})\right\}^2  \\
    &\leq \sup_{\theta \in \Theta} \frac{1}{N}\sum_{n=1}^N \left\{\tilde{Y}_{n-1,k}\left(\tilde{Y}_{n,i} - \beta_{ji}^{*\top}\tilde{Y}_{n-1}\right)\right\}^2 \\
    &= \frac{1}{N}\sum_{n=1}^N \left\{\tilde{Y}_{n-1,k}\left(\tilde{Y}_{n,i} - \beta_{ji}^{*\top}\tilde{Y}_{n-1}\right)\right\}^2 \\
    &= \frac{1}{N}\sum_{n=1}^N\left\{h^{ijk}(\tilde{Y}_{n-1}^n)\right\}^2.
\end{align*}
As shown in step I, the random variable $h^{ijk}(\tilde{Y}_{n-1}^n)$ is sub-weibull(1) with sub-weibull norm $4K_Y^2$. By Lemma 6 in \citet{wong2020lasso}, $\{h^{ijk}(\tilde{Y}_{n-1}^n)\}^2$ is sub-weibull($1/2$) with sub-weibull norm $64K_Y^4$. By Lemma 13 in \citet{wong2020lasso},
\begin{equation*}
    P\left(\left|\frac{1}{N}\sum_{n=1}^N \{h^{ijk}(\tilde{Y}_{n-1}^n)\}^2 - E\left[\{h^{ijk}(\tilde{Y}_{n-1}^n)\}^2\right]\right| > 1\right) \leq N\exp\left\{-\frac{N^{1/2}}{8K_Y^2 \tilde{C}_1}\right\} + \exp\left\{-\frac{N}{\tilde{C}_2 (64K_Y^4)^2}\right\},
\end{equation*}
for some constants $\tilde{C}_1$ and $\tilde{C}_2$. As we have shown, $E[\{h^{ijk}(\tilde{Y}_{n-1}^n)\}^2] \leq 64K_Y^4$. Together with the above display, this implies that 
\begin{equation*}
    P\left(\frac{1}{N}\sum_{n=1}^N \{h^{ijk}(\tilde{Y}_{n-1}^n)\}^2 > 64K_Y^4+ 1\right) \leq N\exp\left\{-\frac{N^{1/2}}{8K_Y^2\tilde{C}_1}\right\} + \exp\left\{-\frac{N}{\tilde{C}_2 (64K_Y^4)^2}\right\},
\end{equation*}
and therefore
\begin{equation*}
    P\left(\sup_{g \in \mathcal{G}^{ijk}}\frac{1}{N}\sum_{n=1}^N g(X_n)^2 > 64K_Y^4+ 1\right) \leq N\exp\left\{-\frac{N^{1/2}}{8K_Y^2\tilde{C}_1}\right\} + \exp\left\{-\frac{N}{\tilde{C}_2 (64K_Y^4)^2}\right\},
\end{equation*}
for any $(i,j,k)$. Applying a union bound, we have that
\begin{equation}\label{empiricalnorm}
    P\left(\sup_{g \in \mathcal{G}}\frac{1}{N}\sum_{n=1}^N g(X_n)^2 > 64K_Y^4+ 1\right) \leq Kd^2 N\exp\left\{-\frac{N^{1/2}}{8K_Y^2\tilde{C}_1}\right\} + Kd^2 \exp\left\{-\frac{N}{\tilde{C}_2 (64K_Y^4)^2}\right\}.
\end{equation}
This provides a way to control the empirical norm of functions in the class $\mathcal{G}$. Specifically, the empirical norm is upper bounded by $ 64K_Y^4+ 1$, uniformly over $G$, with high probability. For notation convenience, define $u_1(N,d)$ such that
\begin{equation}\label{defineu1}
    u_1(N,d) = Kd^2 N\exp\left\{-\frac{N^{1/2}}{8K_Y^2\tilde{C}_1}\right\} + Kd^2 \exp\left\{-\frac{N}{\tilde{C}_2 (64K_Y^4)^2}\right\}.
\end{equation}

\textbf{Step II-III: condition on $X_1,X_2,\ldots,X_N$.} We now condition on $X_1,X_2,\ldots,X_N$ and study the probability
\begin{multline*}
    P\left( \sup_{g\in\mathcal{G}}\left|\frac{1}{N}\sum_{n=1}^N W_ng(X_n)\right| \geq \delta/4 \mid X_1=x_1,\ldots, X_N=x_N, I_\mathcal{A}\left\{X_1,\ldots,X_N\right\} =1 \right) \\
    = P\left( \sup_{g\in\mathcal{G}}\left|\frac{1}{N}\sum_{n=1}^N W_ng(x_n)\right| \geq \delta/4 \right),
\end{multline*}
for a set of values $x_1,\ldots, x_N$ such that $I_{\mathcal{A}}\{x_1,\ldots,x_N\} = 1$. Our main tool to control the probability above is Corollary 8.3 in \citet{geer2000empirical}, which requires controlling the entropy of the function class $\mathcal{G}$.

Let $Q_n(x_1,\ldots,x_N)$ denote the empirical distribution that puts mass $1/N$ at each value $x_n$. We will often omit in the notation its dependence on $(x_1,\ldots,x_N)$ and write $Q_n$ for simplicity. For a function $g$, define its norm under $Q_n$, $\|g\|_{Q_n}$, such that $\|g\|_{Q_n}^2 = \int g^2dQ_n = \sum_{n=1}^N g^2(x_n)/N$. 

Recall that we define a subset of the parameter space $\Theta \subseteq \mathbb{R}^{Kd^2 + K^2}$ as $\Theta = \{\theta = (\beta^\top,p^\top,\sigma^\top)^\top: \|\theta - \theta^*\|_2 \leq r, \|(\beta - \beta^*)_{S^C}\|_1 \leq 4\sqrt{|S|}\|\beta-\beta^*\|_2\}$. We first derive an upper bound on the entropy of $\Theta$, which will be used later to upper bound the entropy of $\mathcal{G}$.

\begin{lemma}[Sudakov Minoration]
Let $A \sim N(0,\textnormal{Id}_{\tilde{d}})$. For any $\Theta \subseteq \mathbb{R}^{\tilde{d}}$ and any $\epsilon>0$,
\begin{equation*}
    \epsilon\sqrt{\log M(\epsilon,\Theta,\|\cdot\|_2)} \leq c E\left[\sup_{\theta \in \Theta} \langle \theta, A \rangle\right],
\end{equation*}
for some constant $c$, where $M(\epsilon,\Theta,\|\cdot\|_2)$ denotes the $\epsilon$-packing number of $\Theta$.
\end{lemma}
\noindent Here, $\tilde{d} = Kd^2 + K^2$. Let $\tilde S = S \cup \{Kd^2+1 ,\ldots, \tilde{d}\}$ denote the support of $\theta^*$, that is, the positions of non-zero components of $\theta^*$. Then,
\begin{align*}
    \langle\theta-\theta^*,A\rangle &= \sum_{n=1}^{\tilde{d}} (\theta_i-\theta^*_i) A_i = \sum_{i \in \tilde{S}} (\theta_i-\theta^*_i)A_i + \sum_{i \in \tilde{S}^C} (\theta_i-\theta^*_i) A_i \\
    &\leq \left\|(\theta-\theta^*)_{\tilde{S}}\right\|_2 \left\|A_{\tilde{S}}\right\|_2 + \left\|(\theta-\theta^*)_{\tilde{S}^C}\right\|_1 \left\|A_{\tilde{S}^C}\right\|_\infty \\
    &= \left\|(\theta-\theta^*)_{\tilde{S}}\right\|_2 \left\|A_{\tilde{S}}\right\|_2 + \left\|(\beta-\beta^*)_{S^C}\right\|_1 \left\|A_{S^C}\right\|_\infty \\
    &\leq \left\|\theta-\theta^*\right\|_2 \left\|A_{\tilde{S}}\right\|_2 + 4\sqrt{|S|}\|\theta-\theta^*\|_2 \left\|A_{S^C}\right\|_\infty \\
    &\leq r\left(\left\|A_{\tilde{S}}\right\|_2 + 4\sqrt{|S|}\left\|A_{S^C}\right\|_\infty\right).
\end{align*}
Therefore,
\begin{align*}
    E\left[\sup_{\theta \in \Theta}\langle\theta,A\rangle \right] &= E\left[\sup_{\theta\in\Theta}\langle\theta-\theta^*,A\rangle + \langle\theta^*,A\rangle\right] = E\left[\sup_{\theta\in\Theta}\langle\theta-\theta^*,A\rangle\right] \\
    &\leq rE\left[\left\|A_{\tilde{S}}\right\|_2 + 4\sqrt{|S|}\left\|A_{S^C}\right\|_\infty\right] \\
    &\leq r\left(4\sqrt{|S|}\sqrt{2\log K+ 4\log d} + \frac{\sqrt{2}\Gamma(\frac{|S|+K^2+1}{2})}{\Gamma(\frac{|S|+K^2}{2})}\right),
\end{align*}
where we have used the fact that the components of $A$ are i.i.d. standard normal random variables and $\Gamma(\cdot)$ here is the Gamma function. The interesting case is when $|S|$ approaches infinity asymptotically, and in this case the ratio between the two gamma functions in the last line is of order $\sqrt{|S|}$. Thus, there exists some constant $C_1$ such that $E\left[\sup_{\theta \in \Theta}\langle\theta,A\rangle \right] \leq C_1 r \sqrt{|S|(\log K + \log d)}$. As a result,
\begin{equation*}
    \sqrt{\log N_c(\epsilon,\Theta,\|\cdot\|_2)} \leq \sqrt{\log M(\epsilon,\Theta,\|\cdot\|_2)} \leq cC_1r \sqrt{|S|(\log K + \log d)}/\epsilon,
\end{equation*}
where $N_c(\epsilon,\Theta,\|\cdot\|_2)$ is the $\epsilon$-covering number of $\Theta$, which is upper bounded by the $\epsilon$-packing number of $\Theta$.

Next, consider the function class $\mathcal{G}^{ijk} = \{g^{ijk}_\theta: g^{ijk}_\theta(\tilde{Y}_{n-s}^{n+s}) = h^{ijk}(\tilde{Y}_{n-1}^n)m_{j,\theta}(\tilde{Y}_{n-s}^{n+s}), \theta\in\Theta \}$. To link the entropy of $\mathcal{G}^{ijk}$ with the entropy of $\Theta$, we show that functions in $\mathcal{G}^{ijk}$ are Lipschitz in $\theta$.
\begin{align*}
    &\quad \|g^{ijk}_{\theta_1} - g^{ijk}_{\theta_2}\|_{Q_n}^2 \\
    &= \frac{1}{N}\sum_{n=1}^N \left\{g^{ijk}_{\theta_1}(\tilde{Y}_{n-s}^{n+s})-g^{ijk}_{\theta_2}(\tilde{Y}_{n-s}^{n+s})\right\}^2 \\
    &= \frac{1}{N}\sum_{n=1}^N \left\{h^{ijk}(\tilde{Y}_{n-1}^n)m_{j,\theta_1}(\tilde{Y}_{n-s}^{n+s}) - h^{ijk}(\tilde{Y}_{n-1}^n)m_{j,\theta_2}(\tilde{Y}_{n-s}^{n+s})\right\}^2 \\
    &= \frac{1}{N}\sum_{n=1}^N \left\{h^{ijk}(\tilde{Y}_{n-1}^n)\right\}^2  \left\{\int_0^1 \frac{\partial m_{j,\theta}(\tilde{Y}_{n-s}^{n+s})}{\partial \theta}\rvert_{\theta = u\theta_1 + (1-u)\theta_2}^\top (\theta_2 - \theta_1)du\right\}^2 \\
    &\leq \frac{1}{N}\sum_{n=1}^N \Bigg[\left\{h^{ijk}(\tilde{Y}_{n-1}^n)\right\}^2 \times \\
    &\quad \int_0^1 (\theta_2 - \theta_1)^\top \left\{\frac{\partial m_{j,\theta}(\tilde{Y}_{n-s}^{n+s})}{\partial \theta}\rvert_{\theta = u\theta_1 + (1-u)\theta_2}\right\}\left\{\frac{\partial m_{j,\theta}(\tilde{Y}_{n-s}^{n+s})}{\partial \theta}\rvert_{\theta = u\theta_1 + (1-u)\theta_2}\right\}^\top(\theta_2 - \theta_1)du\Bigg] \\
    &= \int_0^1 \frac{1}{N}\sum_{n=1}^N \Bigg[\left\{h^{ijk}(\tilde{Y}_{n-1}^n)\right\}^2 \times \\
    &\quad  (\theta_2 - \theta_1)^\top \left\{\frac{\partial m_{j,\theta}(\tilde{Y}_{n-s}^{n+s})}{\partial \theta}\rvert_{\theta = u\theta_1 + (1-u)\theta_2}\right\}\left\{\frac{\partial m_{j,\theta}(\tilde{Y}_{n-s}^{n+s})}{\partial \theta}\rvert_{\theta = u\theta_1 + (1-u)\theta_2}\right\}^\top(\theta_2 - \theta_1)\Bigg]du \\
    &= \int_0^1 (\theta_2 - \theta_1)^\top \Bigg[\frac{1}{N}\sum_{n=1}^N \left\{\frac{\partial m_{j,\theta}(\tilde{Y}_{n-s}^{n+s})}{\partial \theta}\rvert_{\theta = u\theta_1 + (1-u)\theta_2}\right\}\left\{h^{ijk}(\tilde{Y}_{n-1}^n)\right\}^2 \times \\
    &\quad \left\{\frac{\partial m_{j,\theta}(\tilde{Y}_{n-s}^{n+s})}{\partial \theta}\rvert_{\theta = u\theta_1 + (1-u)\theta_2}\right\}^\top\Bigg](\theta_2 - \theta_1)du. \\
\end{align*}
Define a matrix $M_{\theta_1, \theta_2}^{ijk}(u,\tilde{Y}_{n-s}^{n+s})$ such that
\begin{equation*}
    M_{\theta_1, \theta_2}^{ijk}(u,\tilde{Y}_{n-s}^{n+s}) = \left\{\frac{\partial m_{j,\theta}(\tilde{Y}_{n-s}^{n+s})}{\partial \theta}\rvert_{\theta = u\theta_1 + (1-u)\theta_2}\right\}\left\{h^{ijk}(\tilde{Y}_{n-1}^n)\right\}^2 \left\{\frac{\partial m_{j,\theta}(\tilde{Y}_{n-s}^{n+s})}{\partial \theta}\rvert_{\theta = u\theta_1 + (1-u)\theta_2}\right\}^\top,
\end{equation*}
then we have
\begin{align*}
    \|g^{ijk}_{\theta_1} - g^{ijk}_{\theta_2}\|_{Q_n}^2 &\leq \int_0^1 (\theta_2 - \theta_1)^\top \left\{\frac{1}{N}\sum_{n=1}^N M_{\theta_1, \theta_2}^{ijk}(u,\tilde{Y}_{n-s}^{n+s})\right\}(\theta_2 - \theta_1) du \\
    &\leq \int_0^1 \left\|\frac{1}{N}\sum_{n=1}^N M_{\theta_1, \theta_2}^{ijk}(u,\tilde{Y}_{n-s}^{n+s})\right\|_2 \left\|\theta_2 - \theta_1 \right\|_2^2 du \\
    &= \left\|\theta_2 - \theta_1 \right\|_2^2 \int_0^1 \left\|\frac{1}{N}\sum_{n=1}^N M_{\theta_1, \theta_2}^{ijk}(u,\tilde{Y}_{n-s}^{n+s})\right\|_2 du
\end{align*}
Now for $\tilde{\theta} \in \mathcal{B}(r,\theta^*)$, define a matrix $M(\tilde{\theta},Y_{n-s}^n)$ as
\begin{equation*}
    M_{\tilde{\theta}}^{ijk}(\tilde{Y}_{n-s}^{n+s}) = \left\{\frac{\partial m_{j,\theta}(\tilde{Y}_{n-s}^{n+s})}{\partial \theta}\rvert_{\theta = \tilde{\theta}}\right\}\left\{h^{ijk}(\tilde{Y}_{n-1}^n)\right\}^2 \left\{\frac{\partial m_{j,\theta}(\tilde{Y}_{n-s}^{n+s})}{\partial \theta}\rvert_{\theta = \tilde{\theta}}\right\}^\top.
\end{equation*}
We then have
\begin{align*}
    \|g^{ijk}_{\theta_1} - g^{ijk}_{\theta_2}\|_{Q_n}^2 &\leq \|\theta_2 - \theta_1\|_2^2 \int_0^1 \sup_{\tilde\theta \in \mathcal{B}(r,\theta^*)}\left\|\frac{1}{N}\sum_{n=1}^N M_{\tilde{\theta}}^{ijk}(\tilde{Y}_{n-s}^{n+s})\right\|_2 du \\
    &\leq \left\{\sup_{\tilde\theta \in \mathcal{B}(r,\theta^*)}\left\|\frac{1}{N}\sum_{n=1}^N M_{\tilde{\theta}}^{ijk}(\tilde{Y}_{n-s}^{n+s})\right\|_2\right\} \|\theta_2 - \theta_1\|_2^2.
\end{align*}
Define a Lipschitz constant $L_{ijk}(X_1^N)$ such that
\begin{equation*}
    L_{ijk}^2(X_1^N) = L_{ijk}^2(\tilde{Y}_{1-s}^{1-s},\ldots,\tilde{Y}_{n-s}^{n+s}) = \sup_{\tilde\theta \in \mathcal{B}(r,\theta^*)}\left\|\frac{1}{N}\sum_{n=1}^N M_{\tilde{\theta}}^{ijk}(\tilde{Y}_{n-s}^{n+s})\right\|_2,
\end{equation*}
and a Lipschitz constant $L(X_1^N)$ such that
\begin{align*}
    L^2(X_1^N) = L^2(\tilde{Y}_{1-s}^{1-s},\ldots,\tilde{Y}_{n-s}^{n+s}) &= \max\left\{\sup_{\tilde\theta \in \mathcal{B}(r,\theta^*)}\max_{i,j,k}\left\|\frac{1}{N}\sum_{n=1}^N M_{\tilde{\theta}}^{ijk}(\tilde{Y}_{n-s}^{n+s})\right\|_2, 1 \right\} \\
    &= \max\left\{\sup_{\tilde\theta \in \mathcal{B}(r,\theta^*)}\left\|\frac{1}{N}\sum_{n=1}^N M_{\tilde{\theta}}(\tilde{Y}_{n-s}^{n+s})\right\|_2, 1\right\},
\end{align*}
where the matrix $M_{\tilde{\theta}}$ is a block-diagonal matrix with the diagonal blocks given by $M_{\tilde{\theta}}^{ijk}$. 

We give an alternative definition of the Lipschitz constant. 
\begin{align*}
    &\quad \|g^{ijk}_{\theta_1} - g^{ijk}_{\theta_2}\|_{Q_n}^2 \\
    &= \frac{1}{N}\sum_{n=1}^N \left\{g^{ijk}_{\theta_1}(\tilde{Y}_{n-s}^{n+s})-g^{ijk}_{\theta_2}(\tilde{Y}_{n-s}^{n+s})\right\}^2 \\
    &= \frac{1}{N}\sum_{n=1}^N \left\{h^{ijk}(\tilde{Y}_{n-1}^n)m_{j,\theta_1}(\tilde{Y}_{n-s}^{n+s}) - h^{ijk}(\tilde{Y}_{n-1}^n)m_{j,\theta_2}(\tilde{Y}_{n-s}^{n+s})\right\}^2 \\
    &= \frac{1}{N}\sum_{n=1}^N \left\{h^{ijk}(\tilde{Y}_{n-1}^n)\right\}^2  \left\{\int_0^1 \frac{\partial m_{j,\theta}(\tilde{Y}_{n-s}^{n+s})}{\partial \theta}\rvert_{\theta = u\theta_1 + (1-u)\theta_2}^\top (\theta_2 - \theta_1)du\right\}^2 \\
    &\leq \left[\frac{1}{N}\sum_{n=1}^N \left\{h^{ijk}(\tilde{Y}_{n-1}^n)\right\}^4\right]^{1/2} \left[\frac{1}{N}\sum_{n=1}^N \left\{\int_0^1 \frac{\partial m_{j,\theta}(\tilde{Y}_{n-s}^{n+s})}{\partial \theta}\rvert_{\theta = u\theta_1 + (1-u)\theta_2}^\top (\theta_2 - \theta_1)du\right\}^4\right]^{1/2},
\end{align*}
where the last line follows by Cauchy-Schwarz inequality. We now show that the first term in the last line of the above display is upper bounded by some constant, uniformly in $(i,j,k)$ with high probability. To this end, first we recall that we have shown $\{h^{ijk}(\tilde{Y}_{n-1}^n)\}^2$ is sub-weibull($1/2$) with sub-weibull norm $64K_Y^4$ for any $(i,j,k)$. Applying Lemma~6 in \citet{wong2020lasso} again, we get that $\{h^{ijk}(\tilde{Y}_{n-1}^n)\}^4$ is sub-weibull($1/4$) with sub-weibull norm $K_4 = 2^4(64K_Y^4)^2$. Now applying Lemma~13 in \citet{wong2020lasso}, we get the following concentration result:
\begin{equation*}
    P\left(\left|\frac{1}{N}\sum_{n=1}^N \left\{h^{ijk}(\tilde{Y}_{n-1}^n)\right\}^4 - E\left[\left\{h^{ijk}(\tilde{Y}_{n-1}^n)\right\}^4\right]\right| > 1 \right) \leq N \exp\left(-\frac{N^{1/4}}{K_4^{1/4}\tilde C_1}\right) + \exp\left(-\frac{N}{K_4^2\tilde C_2}\right),
\end{equation*}
for $N > 4$. Recall that $h^{ijk}(\tilde{Y}_{n-1}^n)$ is sub-exponential, and therefore $(E|h^{ijk}(\tilde{Y}_{n-1}^n)|^q)^{1/q} \leq 4K_Y^2q$ for all $q \geq 1$. In particular, this implies that
\begin{equation*}
    E\left[\left\{h^{ijk}(\tilde{Y}_{n-1}^n)\right\}^4\right] \leq 16^4K_Y^8, \quad \forall (i,j,k).
\end{equation*}
Combined with our earlier concentration result, we have that 
\begin{equation*}
    P\left(\frac{1}{N}\sum_{n=1}^N \left\{h^{ijk}(\tilde{Y}_{n-1}^n)\right\}^4  > 16^4 K_Y^8 + 1 \right) \leq N \exp\left(-\frac{N^{1/4}}{K_4^{1/4}\tilde C_1}\right) + \exp\left(-\frac{N}{K_4^2\tilde C_2}\right).
\end{equation*}
Applying a union bound over $(i,j,k)$, we get that
\begin{equation*}
    P\left(\max_{i,j,k}\frac{1}{N}\sum_{n=1}^N \left\{h^{ijk}(\tilde{Y}_{n-1}^n)\right\}^4  > 16^4 K_Y^8 + 1 \right)
    \leq Kd^2 N \exp\left(-\frac{N^{1/4}}{K_4^{1/4}\tilde C_1}\right) + Kd^2 \exp\left(-\frac{N}{K_4^2\tilde C_2}\right).
\end{equation*}
Therefore, $[\sum_{n=1}^N \{h^{ijk}(\tilde{Y}_{n-1}^n)\}^4/N]^{1/2}$ is upper bounded by $16^2K_Y^4 + 1$, uniformly in $(i,j,k)$, with high probability. We now turn to the second term, 
\begin{align*}
    &\quad \left[\frac{1}{N}\sum_{n=1}^N \left\{\int_0^1 \frac{\partial m_{j,\theta}(\tilde{Y}_{n-s}^{n+s})}{\partial \theta}\rvert_{\theta = u\theta_1 + (1-u)\theta_2}^\top (\theta_2 - \theta_1)du\right\}^4\right]^{1/2} \\
    &\leq \left[\frac{1}{N}\sum_{n=1}^N \int_0^1 \left\{\frac{\partial m_{j,\theta}(\tilde{Y}_{n-s}^{n+s})}{\partial \theta}\rvert_{\theta = u\theta_1 + (1-u)\theta_2}^\top (\theta_2 - \theta_1)\right\}^4 du \right]^{1/2} \\
    &= \left[\frac{1}{N}\sum_{n=1}^N \int_0^1 \left\{(\theta_2 - \theta_1)^\top\frac{\partial m_{j,\theta}(\tilde{Y}_{n-s}^{n+s})}{\partial \theta}\rvert_{\theta = u\theta_1 + (1-u)\theta_2}\frac{\partial m_{j,\theta}(\tilde{Y}_{n-s}^{n+s})}{\partial \theta}\rvert_{\theta = u\theta_1 + (1-u)\theta_2}^\top (\theta_2 - \theta_1) \right\}^2 du \right]^{1/2} \\
    &= \left[\int_0^1 \frac{1}{N}\sum_{n=1}^N\left\{(\theta_2 - \theta_1)^\top\frac{\partial m_{j,\theta}(\tilde{Y}_{n-s}^{n+s})}{\partial \theta}\rvert_{\theta = u\theta_1 + (1-u)\theta_2}\frac{\partial m_{j,\theta}(\tilde{Y}_{n-s}^{n+s})}{\partial \theta}\rvert_{\theta = u\theta_1 + (1-u)\theta_2}^\top (\theta_2 - \theta_1) \right\}^2 du \right]^{1/2} \\
    &\leq \left[\int_0^1 \frac{1}{N}\sum_{n=1}^N \left\|\frac{\partial m_{j,\theta}(\tilde{Y}_{n-s}^{n+s})}{\partial \theta}\rvert_{\theta = u\theta_1 + (1-u)\theta_2}\frac{\partial m_{j,\theta}(\tilde{Y}_{n-s}^{n+s})}{\partial \theta}\rvert_{\theta = u\theta_1 + (1-u)\theta_2}^\top\right\|_2^2 \|\theta_2 - \theta_1\|_2^4 du\right]^{1/2} \\
    &= \left[\int_0^1 \frac{1}{N}\sum_{n=1}^N \left\|\frac{\partial m_{j,\theta}(\tilde{Y}_{n-s}^{n+s})}{\partial \theta}\rvert_{\theta = u\theta_1 + (1-u)\theta_2}\frac{\partial m_{j,\theta}(\tilde{Y}_{n-s}^{n+s})}{\partial \theta}\rvert_{\theta = u\theta_1 + (1-u)\theta_2}^\top\right\|_2^2 du\right]^{1/2} \|\theta_2 - \theta_1\|_2^2 \\
    &\leq \sup_{\tilde\theta \in \mathcal{B}(r,\theta^*)} \left[ \frac{1}{N}\sum_{n=1}^N \left\|\frac{\partial m_{j,\theta}(\tilde{Y}_{n-s}^{n+s})}{\partial \theta}\rvert_{\theta = \tilde\theta}\frac{\partial m_{j,\theta}(\tilde{Y}_{n-s}^{n+s})}{\partial \theta}\rvert_{\theta = \tilde\theta}^\top\right\|_2^2 \right]^{1/2}\|\theta_2 - \theta_1\|_2^2.
\end{align*}
Define a Lipschitz constant $L(x_1^N)$ such that
\begin{equation*}
    L^2(x_1^N) = \max_{i,j,k}\left[\frac{1}{N}\sum_{n=1}^N \left\{h^{ijk}(\tilde{Y}_{n-1}^n)\right\}^4\right]^{1/2} \sup_{\tilde\theta \in \mathcal{B}(r,\theta^*)} \left[ \frac{1}{N}\sum_{n=1}^N \left\|\frac{\partial m_{j,\theta}(\tilde{Y}_{n-s}^{n+s})}{\partial \theta}\rvert_{\theta = \tilde\theta}\frac{\partial m_{j,\theta}(\tilde{Y}_{n-s}^{n+s})}{\partial \theta}\rvert_{\theta = \tilde\theta}^\top\right\|_2^2 \right]^{1/2}.
\end{equation*}
For some sequence $l(N,d)$, let $\tilde{u}_j(N,d)$ denote the following probability
\begin{equation*}
    \tilde{u}_j(N,d) = P\left(\sup_{\tilde\theta \in \mathcal{B}(r,\theta^*)} \left[ \frac{1}{N}\sum_{n=1}^N \left\|\frac{\partial m_{j,\theta}(\tilde{Y}_{n-s}^{n+s})}{\partial \theta}\rvert_{\theta = \tilde\theta}\frac{\partial m_{j,\theta}(\tilde{Y}_{n-s}^{n+s})}{\partial \theta}\rvert_{\theta = \tilde\theta}^\top\right\|_2^2 \right]^{1/2} > l(N,d)\right),
\end{equation*}
then $L^2(x_1^N) \leq l(N,d)(16^2K_Y^4 + 1)$ with probability at least 
\begin{equation*}
    1 - \sum_{j=1}^K \tilde{u}_j(N,d) - Kd^2 N \exp\left(-\frac{N^{1/4}}{K_4^{1/4}\tilde C_1}\right) - Kd^2 \exp\left(-\frac{N}{K_4^2\tilde C_2}\right).
\end{equation*}

With the defined Lipschitz constant, we have $\|g^{ijk}_{\theta_1} - g^{ijk}_{\theta_2}\|_{Q_n} \leq L(x_1^N) \|\theta_1-\theta_2\|_2$. Therefore, one can construct an $\epsilon L(x_1^N)$-cover of the function class $\mathcal{G}^{ijk}$ from an $\epsilon$-cover of $\Theta$.
\begin{equation*}
    \sqrt{\log N_c(\epsilon L(x_1^N),\mathcal{G}^{ijk},\|\cdot\|_{Q_n})} \leq \sqrt{\log N_c(\epsilon,\Theta,\|\cdot\|_2)} \leq \frac{cC_1r \sqrt{|S|(\log K + \log d)}}{\epsilon},
\end{equation*}
and 
\begin{equation*}
    \sqrt{\log N_c(\epsilon,\mathcal{G}^{ijk},\|\cdot\|_{Q_n})} \leq \frac{cC_1r L(x_1^N) \sqrt{|S|(\log K + \log d)}}{\epsilon}.
\end{equation*}
As $\mathcal{G} = \cup_{i,j,k} \mathcal{G}^{ijk}$, we have that for some constant $C_2$,
\begin{equation*}
    \sqrt{\log N_c(\epsilon,\mathcal{G},\|\cdot\|_{Q_n})} \leq \sqrt{\log \left\{Kd^2 N_c(\epsilon,\mathcal{G}^{ijk},\|\cdot\|_{Q_n})\right\}} \leq \frac{C_2 r L(x_1^N) \sqrt{|S|(\log K + \log d)}}{\epsilon}.
\end{equation*}
The above display gives an upper bound on the entropy of the class $\mathcal{G}$.

We are now ready to apply Corollary 8.3 in \citet{geer2000empirical}, which is included as Theorem~\ref{corollary83} here for completeness.
\begin{theorem}[Corollary 8.3 in \citet{geer2000empirical}, uniform concentration]\label{corollary83} Suppose that $\sup_{g \in \mathcal{G}} \|g\|_{Q_n} \leq R$. Then for some constant C and $\delta_1>0$ satisfying $R>\delta_1$ and
\begin{equation}\label{requirement}
    \sqrt{N}\delta_1 \geq 2C \max\left\{R, \int_{\delta_1/8}^R \sqrt{\log N_c(\epsilon,\mathcal{G},\|\cdot\|_{Q_n})} d\epsilon\right\},
\end{equation}
we have
\begin{equation*}
    P\left(\sup_{g \in \mathcal{G}} \left|\frac{1}{N}\sum_{n=1}^N W_n g(x_n)\right| \geq \delta_1 \right) \leq C\exp\left\{-\frac{N\delta_1^2}{4C^2R^2}\right\}.
\end{equation*}
\end{theorem}
\noindent From step II-II, when $I_{\mathcal{A}}\{x_1,\ldots,x_N\}= 1$, we can take $R^2 = 64K_Y^4 + 1$. Here, we shall take 
\begin{equation*}
    \delta_1 = C_3 \sqrt{\frac{|S|L^2(x_1^N)(\log K + \log d)(\log T)^3 + (\log T)^4}{T}},
\end{equation*}
for an appropriate constant $C_3$, and show that this choice of $\delta_1$ satisfies the requirements in Theorem~\ref{corollary83}. The requirement that $R > \delta_1$ is easily met noting that $R$ is a constant while $\delta_1$ converges to 0 as $T \rightarrow \infty$ (given that $L(x_1^N)$ is well-behaved, and we will discuss this explicitly later.) To check \eqref{requirement}, we first note that $\sqrt{N}\delta_1 \geq 2CR$ for sufficiently large $T$, as $\sqrt{N}\delta_1$ approaches infinity when $T$ approaches infinity. Therefore, we focus on the entropy integral
\begin{align*}
    \int_{\delta_1/8}^R \sqrt{\log N_c(\epsilon,\mathcal{G},\|\cdot\|_{Q_n})} d\epsilon &\leq C_2 r L(x_1^N) \sqrt{|S|(\log K + \log d)} \int_{\delta_1/8}^R \frac{1}{\epsilon}d\epsilon \\
    &= C_2 r L(x_1^N) \sqrt{|S|(\log K + \log d)}\left\{ \log R + \log{8/\delta_1}\right\}.
\end{align*}
In particular,
\begin{align*}
    \log(8/\delta_1) &= \log(8/C_3) + \frac{1}{2}\left\{\log T - \log\left(|S|L(x_1^N)(\log K + \log d) (\log T)^3+ (\log T)^4\right)\right\} \\
    &\leq \log(8/C_3) + \frac{1}{2}\log T.
\end{align*}
Therefore
\begin{multline*}
    2C \int_{\delta_1/8}^R \sqrt{\log N_c(\epsilon,\mathcal{G},\|\cdot\|_{Q_n})} d\epsilon \leq \tilde{C}_2 r L(x_1^N) \sqrt{|S|(\log K + \log d)} \log T \\
    \leq C_3 \sqrt{|S|L^2(x_1^N)(\log K + \log d)(\log T)^2 + (\log T)^3} = \sqrt{N}\delta_1,
\end{multline*}
for a large enough constant $C_3$. Applying Theorem~\ref{corollary83} with the specified value of $\delta_1$, we have that
\begin{equation*}
    P\left(\sup_{g \in \mathcal{G}} \left|\frac{1}{N}\sum_{n=1}^N W_n g(x_n)\right| \geq \delta_1 \right) \leq C\exp\left\{-\frac{N\delta_1^2}{4C^2(64K_Y^4 + 1)}\right\}.
\end{equation*}

\textbf{Step II-III: marginalize over $(X_1,\ldots, X_N)$.} In the above argument, we derived an upper bound on the probability, conditioned on $(x_1,\ldots, x_N)$ such that $I_{\mathcal{A}} =1$. At the same time, the upper bound depends on a random entropy number involving $L(x_1^N)$. Now we marginalize over $(X_1,\ldots, X_N)$. For two events $E_1$ and $E_2$, let $E_1 \wedge E_2$ denote the event that $E_1$ and $E_2$ both happen, and $E_1 \vee E_2$ denote the event that at least one of $E_1$ and $E_2$ happens. Then, we have
\begin{align*}
    &\quad P\left( \sup_{g\in\mathcal{G}}\left|\frac{1}{N}\sum_{n=1}^N W_ng(X_n)\right| \geq \delta/4 \right) \\
    &= P\left( \sup_{g\in\mathcal{G}}\left|\frac{1}{N}\sum_{n=1}^N W_ng(X_n)\right| \geq \frac{\delta}{4} \wedge  \left\{L^2(X_1^N) < l(T,d) \wedge I_{\mathcal{A}} = 1\right\}\right)  \\
    &\quad + P\left( \sup_{g\in\mathcal{G}}\left|\frac{1}{N}\sum_{n=1}^N W_ng(X_n)\right| \geq \frac{\delta}{4} \wedge \left\{ L^2(X_1^N) \geq l(T,d) \vee I_{\mathcal{A}} = 0\right\}\right) \\
    &\leq P\left( \sup_{g\in\mathcal{G}}\left|\frac{1}{N}\sum_{n=1}^N W_ng(X_n)\right| \geq \frac{\delta}{4} \wedge  \left\{L^2(X_1^N) < l(T,d) \wedge I_{\mathcal{A}} = 1\right\}\right) \\
    &\quad + P\left(L^2(X_1^N) \geq l(T,d)\right) + P(I_{\mathcal{A}}= 0).
\end{align*}
In particular, 
\begin{align*}
    &\quad P\left( \sup_{g\in\mathcal{G}}\left|\frac{1}{N}\sum_{n=1}^N W_ng(X_n)\right| \geq \frac{\delta}{4} \wedge  \left\{L^2(X_1^N) < l(T,d) \wedge I_{\mathcal{A}} = 1\right\}\right) \\
    &= E\left[P\left( \sup_{g\in\mathcal{G}}\left|\frac{1}{N}\sum_{n=1}^N W_ng(X_n)\right| \geq \frac{\delta}{4} \wedge  L^2(X_1^N) < l(T,d) \wedge I_{\mathcal{A}}=1 \mid X_1, \ldots, X_N \right) \right] \\
    &= E\left[P\left( \sup_{g\in\mathcal{G}}\left|\frac{1}{N}\sum_{n=1}^N W_ng(X_n)\right| \geq \frac{\delta}{4} \mid X_1, \ldots, X_N \right) I\{L^2(X_1^N) < l(T,d)\} I_{\mathcal{A}}\{X_1^N\}\right].
\end{align*}
Note that condition on the event $L^2(X_1^N) < l(T,d)$, we have $\delta_1 \leq \delta/4$, when $C_3$ in the definition of $\delta_1$ is properly chosen. Hence the expectation in the above display is upper bounded by
\begin{align*}
    &\quad E\left[P\left( \sup_{g\in\mathcal{G}}\left|\frac{1}{N}\sum_{n=1}^N W_ng(x_n)\right| \geq \delta_1 \mid X_1, \ldots, X_N \right) I\{L^2(X_1^N) < l(T,d)\} I_{\mathcal{A}}\{X_1^N\}\right] \\
    &= E\left[C\exp\left\{-\frac{N\delta_1^2}{4C^2(64K_Y^2 + 1)}\right\}I\{L^2(X_1^N) < l(T,d)\} I_{\mathcal{A}}\{X_1^N\}\right].
\end{align*}
Now, we note that by definition $L^2(x_1^N) \geq 1$ for all values of $x_1^N$, and therefore,
\begin{align*}
    \delta_1^2 &= C_3^2 \frac{|S|L^2(x_1^N)(\log K + \log d)(\log T)^3 + (\log T)^4}{T} \\
    &\geq C_3^2 \frac{|S|(\log K + \log d)(\log T)^3 + (\log T)^4}{T} := \delta_2^2.
\end{align*}
In particular, the quantity $\delta_2$ is now independent of the values of $X_1,\ldots,X_N$. Then, we have
\begin{equation*}
    P\left( \sup_{g\in\mathcal{G}}\left|\frac{1}{N}\sum_{n=1}^N W_ng(x_n)\right| \geq \frac{\delta}{4} \wedge  L^2(X_1^N) < l(T,d) \wedge I_{\mathcal{A}}(X_1^N)\right) \leq C\exp\left\{-\frac{N\delta_2^2}{4C^2(64K_Y^2 + 1)}\right\},
\end{equation*}
and
\begin{equation*}
    P\left( \sup_{g\in\mathcal{G}}\left|\frac{1}{N}\sum_{n=1}^N W_ng(x_n)\right| \geq \delta/4 \right) \leq C\exp\left\{-\frac{N\delta_2^2}{4C^2(64K_Y^2 + 1)}\right\} + u(N,d) + u_1(N,d).
\end{equation*}
By Theorem~\ref{symmetrization}, 
\begin{equation*}
    P\left(\sup_{g\in\mathcal{G}}\left|\frac{1}{N}\sum_{n=1}^N g(x_n) - E[g(x_n)]\right| \geq \delta \right) \leq 4C\exp\left\{-\frac{N\delta_2^2}{4C^2(64K_Y^2 + 1)}\right\} + 4u(N,d) + 4u_1(N,d).
\end{equation*}

\textbf{Step III: uniform concentration for $\beta$-mixing processes.} Now we extend the above uniform concentration results to the process $\{Y_{t-s}^{t+s}\}$.
\begin{theorem}[Theorem 2 in \citet{karandikar2002rates}]\label{karandikarthm}
Let $P$ be a shift invariant probability measure and $P^*$ be the infinite product probability measure with the same one-dimensional marginals as $P$. Fix a sequence $\{k_T\}$ such that $k_T \leq T$ and let $l_T$ be the integer part of $T/k_T$. Then
\begin{equation*}
    q(T,\delta,P) \leq T b_{\textnormal{mix}}(k_T,P) + k_T\max\left\{q(l_T+1, \delta, P^*),q(l_T, \delta, P^*)\right\}.
\end{equation*}
\end{theorem}
Recall that we chose $s = \log (T)/(-2\log(\phi))$ and that $b_{\textnormal{mix}}(l) \leq 2\exp\{-c(l-2s)\}$ for $l-2s$ sufficiently large. So we apply the above theorem with $k_T = 3\log(T)/(2c) - \log (T)/\log(\phi) = C\log(T)$ such that the first term
\begin{equation*}
    Tb_{\textnormal{mix}}(k_T) \leq 2 T \exp\{-c3\log(T)/(2c)\}  = 2/\sqrt{T},
\end{equation*}
which converges to 0. Note that now $l_T$ is of the order $T/\{\tilde{c}\log T\}$, that is, $N$. So the second probability is of the same order as
\begin{equation*}
     4C\exp\left\{\log\log T-\frac{N\delta_2^2}{4C^2(64K_Y^4 + 1)}\right\} + 4(\log T) u(N,d) + 4(\log T) u_1(N,d).
\end{equation*}
Now plugging in the expression for $N$ and $\delta_2$, we have that the above quantity is upper bounded by
\begin{align*}
    &\quad 4C\exp\left\{\log\log T-C_3^2\frac{|S|(\log K + \log d)(\log T)^2 + (\log T)^3}{4C^2(64K_Y^4 + 1)} \right\} + \\
    &+ 4(\log T) u(T/\log T,d) \\
    &+ 4\exp\left\{2\log T + 2\log d + 2\log K-\frac{T^{1/2}}{(\log T)^{1/2}8K_Y^2\tilde{C}_1}\right\} \\
    &+ 4\exp\left\{\log\log T + 2\log d + \log K-\frac{T}{\log T\tilde{C}_2 (64K_Y^4)^2}\right\}
\end{align*}
which converges to 0 provided that 
\begin{equation*}
    (\log T) u(T/\log T,d) \rightarrow 0, \textnormal{ when } T \rightarrow \infty,
\end{equation*}
and that
\begin{equation*}
    \frac{\left(\log d + \log K\right)^2\log T}{T} = o(1).
\end{equation*}

\textbf{Step-IV: conclusion of the proof}. Combining everything, we have shown that under our conditions, for
\begin{equation*}
    \delta = C \sqrt{\frac{|S|l(T,d)(\log T)^3 (\log K + \log d) + (\log T)^4}{T}},
\end{equation*}
we have that as $T \rightarrow 0$,
\begin{equation*}
    P\left(\max_{i,j,k}\sup_{\theta\in\Theta}\left|\frac{1}{T}\sum_{t=1}^Tf_\theta^{ijk}(Y_{t-s}^{t+s}) -E\left[f_\theta^{ijk}(Y_{t-s}^{t+s})\right]\right| \geq \delta \right) \rightarrow 0.
\end{equation*}
This provides a means to control $\Delta$, and similar argument can be used to control $\Delta_p$ and $\Delta_\sigma$.
\end{proof}

\section{Useful lemmas}\label{app:usefullemma}
First we introduce an auxiliary lemma that will be useful to establish the approximation error bound in Lemma~\ref{lemma:approximationerror}. First, consider a generic (row) stochastic matrix $M \in \mathbb{R}^{a\times b}$, that is, a matrix whose entries are all non-negative and where entries in each row sum up to 1. Following \citet{hajnal1958weak}, we define the following quantities to measure the extent to which the rows of $M$ differ from each other. Define $\zeta(M) := \max_{1\leq i,k \leq a}\{1-\sum_{1\leq j \leq b}\min(M_{ij},M_{kj})\}$. We note that $\zeta(M)$ can be written as
\begin{align*}
    \zeta(M) &= \max_{1\leq i,k \leq a}\{1-\sum_{1\leq j \leq b}\min(M_{ij},M_{kj})\} \\
    &= \max_{1\leq i,k \leq a}\{\sum_{1\leq j \leq b} M_{ij}-\sum_{1\leq j \leq b}\min(M_{ij},M_{kj})\} \\
    &= \max_{1\leq i,k \leq a} \sum_{j: M_{ij} > M_{kj}} (M_{ij}-M_{kj}),
\end{align*}
and therefore $\zeta(M)$ is zero if and only if the rows of $M$ are all the same. Define $\psi(M) := \max_{1\leq i,k \leq a} \max_{1\leq j \leq b} |M_{ij} - M_{kj}|$. Again, $\psi(M) = 0$ if and only if the rows of $M$ are all the same. The following lemma establishes an important property of these measures.

\begin{lemma}[Lemma 3 in \citet{hajnal1958weak}]\label{hajnallemma}
If $M = M_1 M_2$ where $M_1$ and $M_2$ are both stochastic matrices, then $\psi(M) \leq \zeta(M_1) \psi(M_2)$.
\end{lemma}

Next, we state a key concentration result for $\beta$-mixing processes that we will apply when proving the restricted eigenvalue condition.
\begin{lemma}[Adapted from Lemma 13 of \citet{wong2020lasso}; see also \citet{merlevede2011bernstein}]\label{keyconcentration}
Let $\{X_t\}_{t=1}^T$ be a strictly stationary sequence of mean zero random variables that are subweibull$(\gamma_2)$ with subweibull constant $K_X$. Denote their sum by $S_T$. Suppose their $\beta$-mixing coefficients satisfy $b_{\textnormal{mix}}(l) \leq 2\exp(-c(l-s)^{\gamma_1})$ for $l \geq s$ and $s \leq C\log(T)$ for some constant $C$. Let $\gamma$ be a parameter given by $\gamma = (1/\gamma_1 + 1/\gamma_2)^{-1}$, and further assume $\gamma < 1$. Then for $T>4$ and any $t > T^{-1/2}$,
\begin{equation}
    P\left(\left|S_T/T\right| > t\right) \leq T\exp\left\{-\frac{(tT)^\gamma}{K_X^\gamma C_1}\right\} + \exp\left\{-\frac{t^2T}{K_X^2C_2}\right\},
\end{equation}
where the constants $C_1$ and $C_2$ depend only on $\gamma_1$, $\gamma_2$ and $c$.
\end{lemma}

\section{Additional numerical results}\label{app:additionalsim}

\subsection{Additional simulation results}
We consider the same three settings as in Section~\ref{sec: simulation} and include two additional estimators. The first one is based on a variant of the EM algorithm that replaces the unobserved indicator functions with the filtered probabilities in the E-step. Specifically, the filtered probabilities are defined as
\begin{align*}
    \tilde w_{j,\theta}(Y_0^t) &= P_\theta(Z_t=j| Y_0, Y_1,\ldots, Y_t); \\
    \tilde w_{ij,\theta}(Y_0^t) &= P_\theta(Z_{t-1}=i, Z_t=j| Y_0, Y_1,\ldots, Y_t).
\end{align*}
The second estimator uses approximations of these filtered probabilities in the E-step defined as follows
\begin{align*}
    \tilde m_{j,\theta}(Y_0^t) &= P_\theta(Z_t=j| Z_{t-s}=1,Y_{t-s}, Y_{t-s+1},\ldots, Y_t); \\
    \tilde m_{ij,\theta}(Y_0^t) &= P_\theta(Z_{t-1}=i, Z_t=j| Z_{t-s}=1,Y_{t-s}, Y_{t-s+1},\ldots, Y_t),
\end{align*}
with $s \asymp \log T$. Again, the approximations are equally accurate given $Z_{t-s}=i$ for any value $i.$

We present the results for Settings I-III in Figures~\ref{fig: Setting1_w_filter}-\ref{fig: Setting3_w_filter}. In addition to the patterns noted in Section~\ref{sec: simulation}, we also observe that using (approximations of) the smoothed probabilities and filtered probabilities leads to very similar estimation errors across all settings considered except with moderate sample size ($T=1000$) in Setting II. The filtered-probability-based variants have slightly larger estimation errors in the regression coefficients than their smoothed-probability-based counterparts.

\begin{figure}
    \centering
    \includegraphics[width=0.95\textwidth]{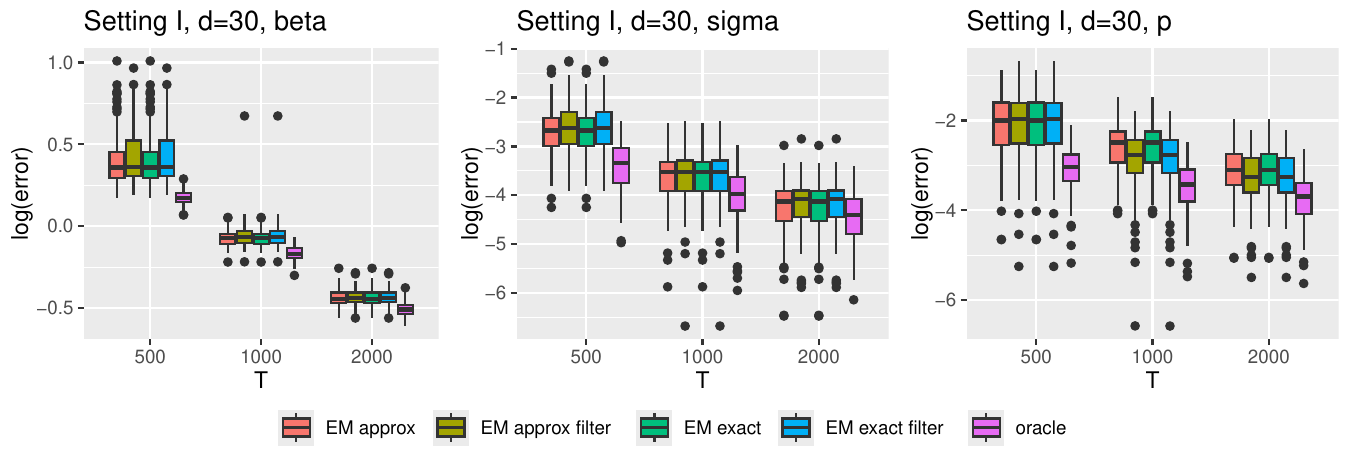}
    \includegraphics[width=0.95\textwidth]{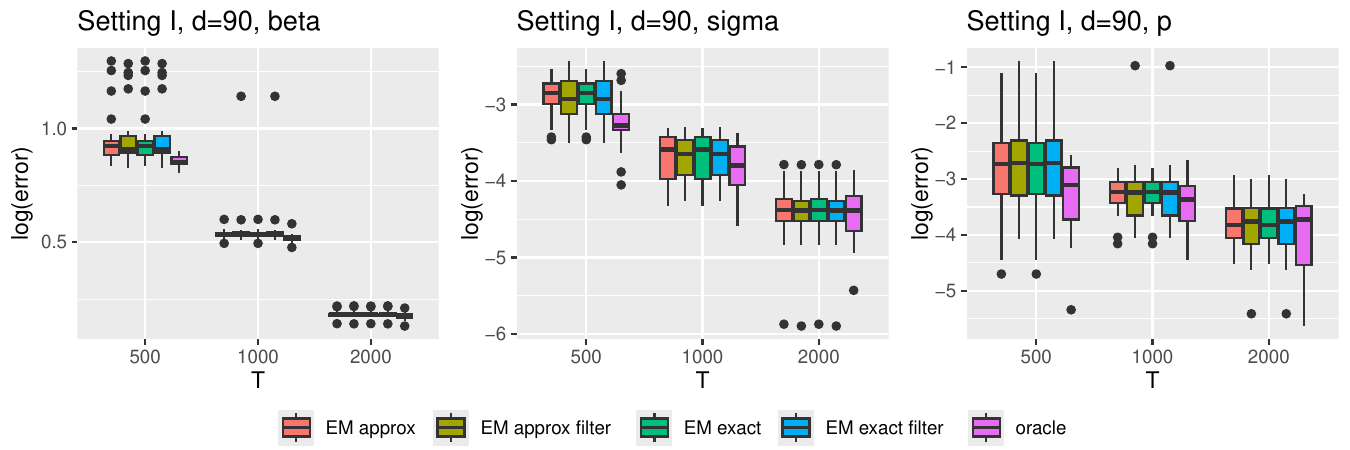}
    \caption{Estimation error of regression coefficients $\beta$ (left column), conditional variance $\sigma^2$ (middle column), and transition probabilities $p$ (right column) in Setting I, comparing the oracle estimator (oracle), the EM algorithm with exact smoothed probabilities (EM exact), the approximate EM algorithm with approximate smoothed probabilities (EM approx), variant of the EM using filtered probabilities (EM exact filter), and variant of the approximate EM using approximate filtered probabilities (EM approx filter). We vary $d \in \{30, 90\}$ and $T \in \{500,1000,2000\}$. Log error is defined as $\log(\|\hat\beta - \beta^*\|_2)$, $\log(\|\hat\sigma^2 - (\sigma^*)^2\|_2)$, and $\log(\|\hat p - p^*\|_2)$, respectively. Results are based on 100 simulation replications for $d=30$, and 20 for $d=90$.}
    \label{fig: Setting1_w_filter}
\end{figure}

\begin{figure}
    \centering
    \includegraphics[width=0.95\textwidth]{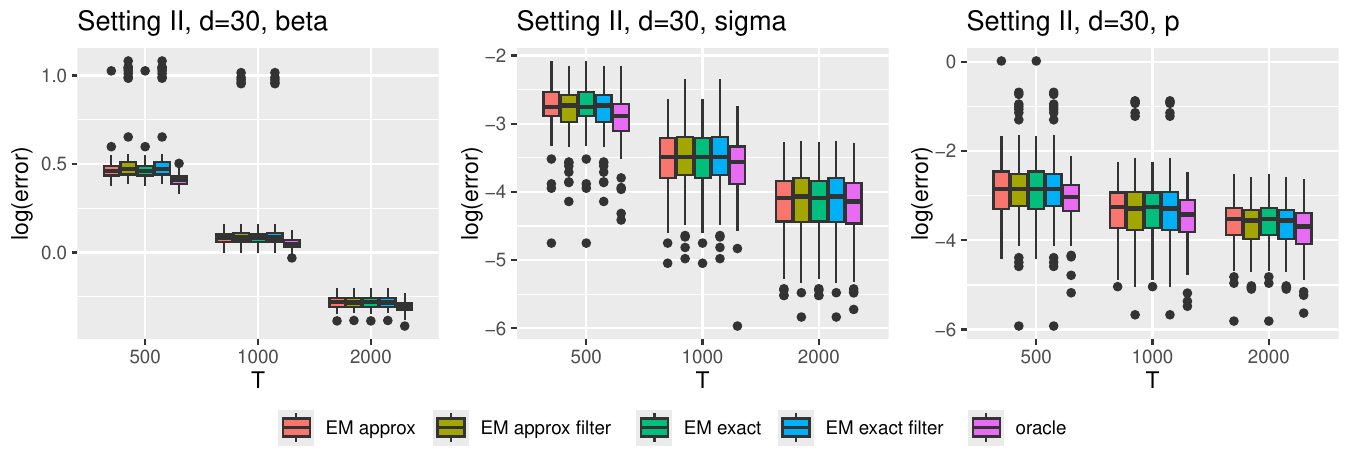}
    \includegraphics[width=0.95\textwidth]{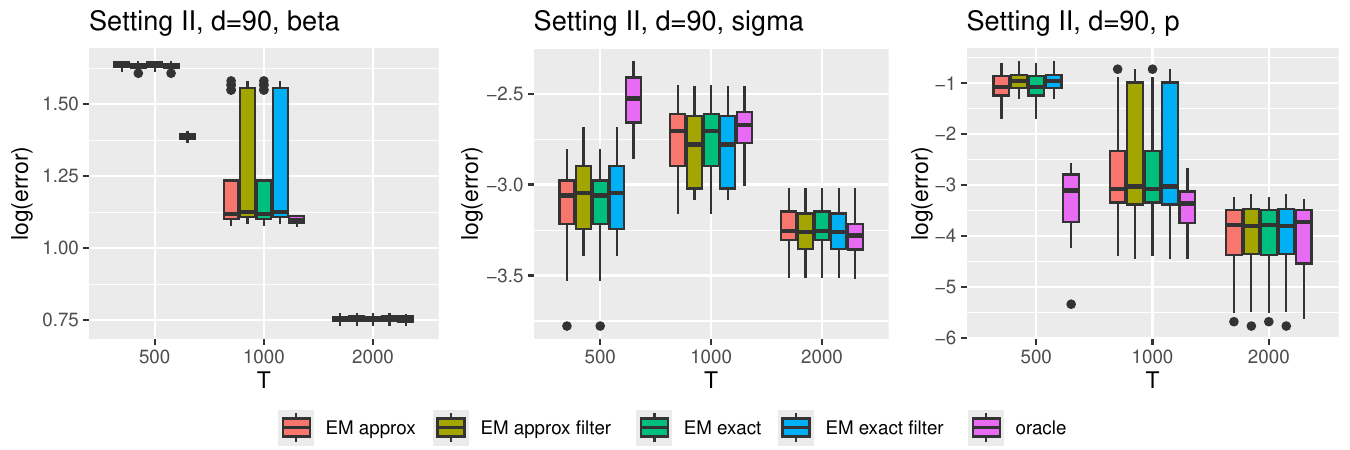}
    \caption{Estimation error of regression coefficients $\beta$ (left column), conditional variance $\sigma^2$ (middle column), and transition probabilities $p$ (right column) in Setting II, comparing the oracle estimator (oracle), the EM algorithm with exact smoothed probabilities (EM exact), the approximate EM algorithm with approximate smoothed probabilities (EM approx), variant of the EM using filtered probabilities (EM exact filter), and variant of the approximate EM using approximate filtered probabilities (EM approx filter). We vary $d \in \{30, 90\}$ and $T \in \{500,1000,2000\}$. Log error is defined as $\log(\|\hat\beta - \beta^*\|_2)$, $\log(\|\hat\sigma^2 - (\sigma^*)^2\|_2)$, and $\log(\|\hat p - p^*\|_2)$, respectively. Results are based on 100 simulation replications for $d=30$, and 20 for $d=90$.}
    \label{fig: Setting2_w_filter}
\end{figure}

\begin{figure}
    \centering
    \includegraphics[width=0.95\textwidth]{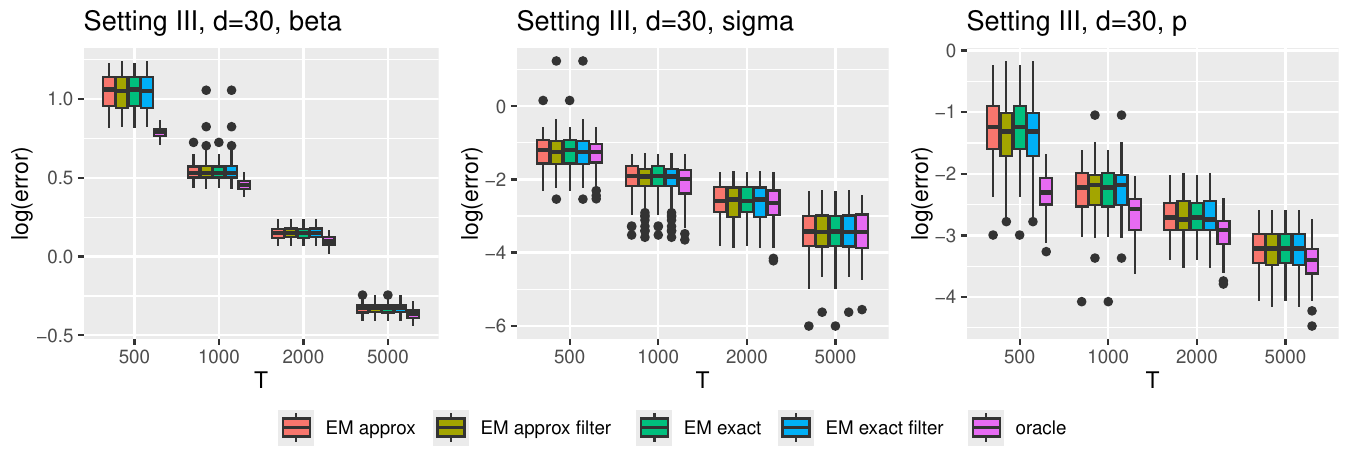}
    \includegraphics[width=0.95\textwidth]{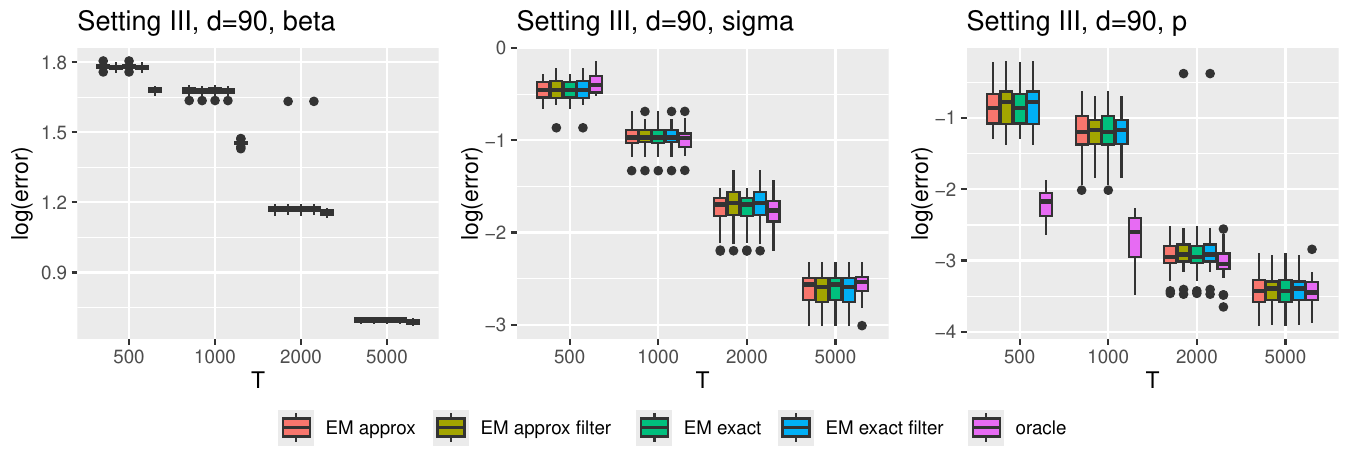}
    \caption{Estimation error of regression coefficients $\beta$ (left column), conditional variance $\sigma^2$ (middle column), and transition probabilities $p$ (right column) in Setting III with 3 regimes, comparing the oracle estimator (oracle), the EM algorithm with exact smoothed probabilities (EM exact), the approximate EM algorithm with approximate smoothed probabilities (EM approx), variant of the EM using filtered probabilities (EM exact filter), and variant of the approximate EM using approximate filtered probabilities (EM approx filter). We vary $d \in \{30, 90\}$ and $T \in \{500,1000,2000, 5000\}$. Log error is defined as $\log(\|\hat\beta - \beta^*\|_2)$, $\log(\|\hat\sigma^2 - (\sigma^*)^2\|_2)$, and $\log(\|\hat p - p^*\|_2)$, respectively. Results are based on 50 simulation replications for $d=30$, and 20 for $d=90$.}
    \label{fig: Setting3_w_filter}
\end{figure}

\subsection{Additional results on the EEG data analysis}
In the EEG data analysis, we fit Markov-switching VARs varying $K \in \{2,3,4\}$, and plot the estimated posterior probabilities of each state in Figure~\ref{fig: state prob comp}. We observe that the results with $K=3$ and $K=4$ are less interpretable as the estimated states do not correspond well to the seizure events and are generally mixed over time. To further validate our selection of $K=2$, we compute a high-dimensional BIC criteria for the estimated models with different numbers of states. Specifically, we use the EBIC given in \citet{wang2011consistent} \citep[see, also, ][]{chen2008extended},
\begin{equation*}
    \textnormal{EBIC}_\gamma(M) = -2\log \hat L + \left\{\log T + 2\gamma \log p \right\}|M|,
\end{equation*}
where $\hat L$ is the value of the likelihood function,$|M|$ is the number of active parameters in a model $M$, $p$ is the total number of parameters which equals to $Kd^2 + K^2$ in our setting, and $T$ is the number of time points. We choose the hyperparameter $\gamma = 1$. As the latent regime variable $Z_t$ is not observable, we cannot compute the full-data likelihood function or the observed-data likelihood function with ease without marginalizing over all possible paths of $Z_t$, which is computationally infeasible. Instead, we replace $\log \hat L$ with its conditional expectation, which is the unpenalized objective function in the M-step:
\begin{multline}\label{eq: expected log likelihood}
    \sum_{t=1}^{T}
    \Bigg[ \left(\sum_{i=1}^K\sum_{j=1}^K m_{ij,\hat\theta}(Y_{t-s}^{t+s}) \log \hat p_{ij}\right) 
    - \frac{d}{2} \sum_{j=1}^K m_{j,\hat\theta}(Y_{t-s}^{t+s}) \left(\log 2\pi + \log \hat\sigma_j^2 \right) \\
    - \left( \sum_{j=1}^K m_{j,\hat\theta}(Y_{t-s}^{t+s}) \frac{1}{2\hat\sigma_j^2} \left\|Y_t - (\textnormal{Id}_d \otimes Y_{t-1})^\top \hat\beta_j \right\|_2^2\right)\Bigg].
\end{multline}
Computing the EBIC for the fitted models, we get values $24,026.39$, $26,849.29$ and $29,999.76$ for $K=2,3,4$, respectively. This suggests choosing $K=2$ as this leads to the smallest EBIC value.

\begin{figure}
    \centering
    \includegraphics[width=0.75\linewidth,height=0.65\textwidth]{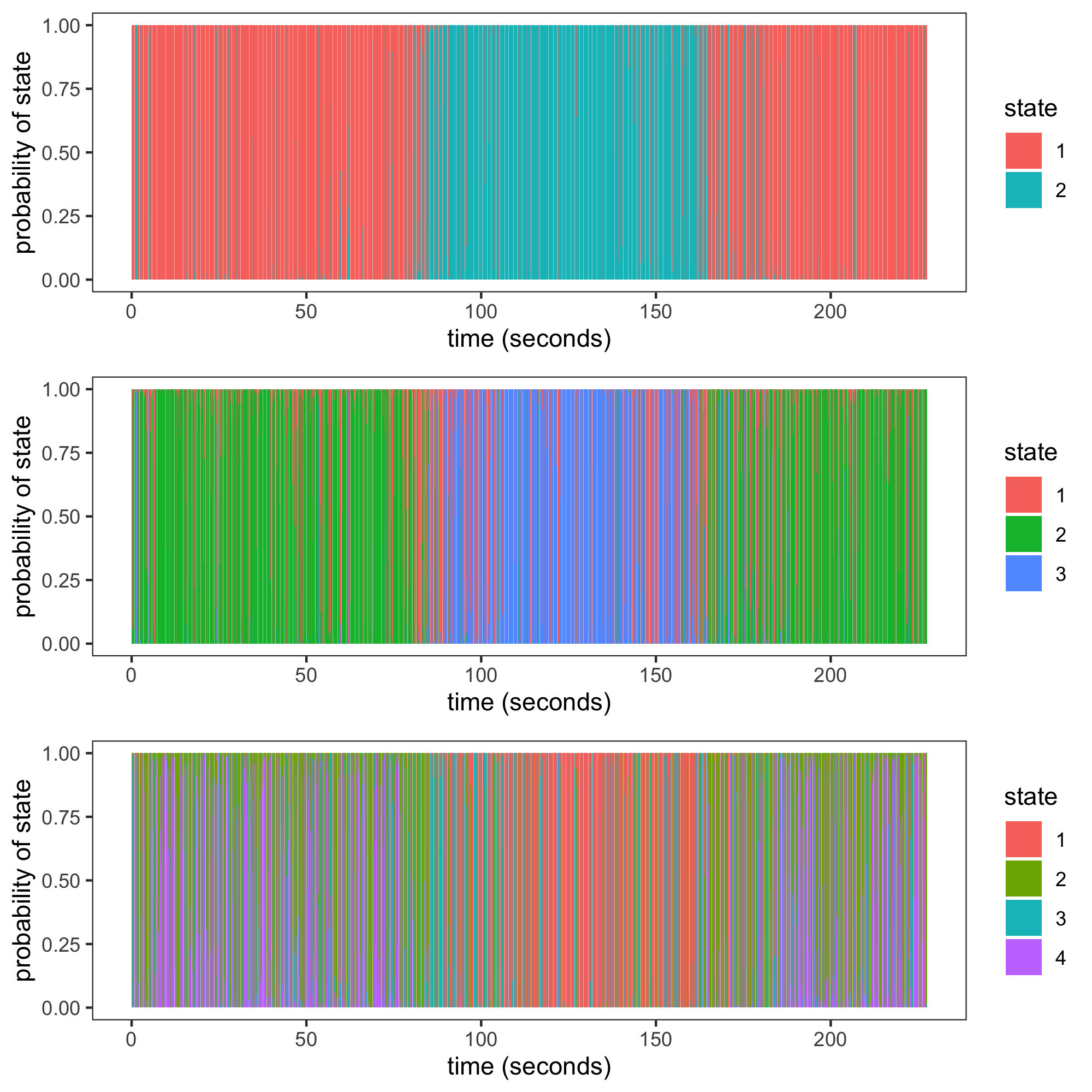}
    \caption{Conditional probability of each state estimated by the regularized approximate EM algorithm in a Markov-switching VAR model with 2, 3 or 4 states. Height of colored bar represents the estimated probability of the corresponding state}
    \label{fig: state prob comp}
\end{figure}

\section{Signal strength in a symmetric case with i.i.d. $Z_t$}\label{app:signalempirical}
In this Supplementary Appendix, we study the (inverse) signal strength measure $\kappa$ introduced in Assumption~\ref{signalstrength} more closely in an example. Specifically, we focus on the case where the regime variables $\{Z_t\}$ are binary with value 1 or 2, and are independent and identically distributed over time, that is, $p_{11} = p_{21}$ in the transition matrix. Moreover, we assume that the transition probability of $Z$, $p$, is known, the variance $\sigma^2$ is known to be 1 in both regimes, and the regression coefficient matrices in the two regimes are such that $B_1 = B$ and $B_2 = -B$. As $B_1  = -B_2$, we refer to this case as a symmetric mixture (of vector autoregression.) 

Let $\beta = \vectorize{B}$, and $\beta$ is the only parameter that needs to be estimated. Hence, hereafter in Supplementary Appendix~\ref{app:signalempirical}, we will write $\beta$ for the parameter vector instead of $\theta$. Let $p_1 = P(Z_t = 1)$ and $p_2 = P(Z_t = 2)$. The independence among the $Z_t$'s allows us to obtain a simple form for the smoothed probability $w_{j,\beta}(Y_0^T)$. Indeed, we have
\begin{align*}
    w_{1,\beta}(Y_0^T) &= P(Z_t = 1 | Y_0, Y_1,\ldots,Y_T) = P(Z_t = 1 | Y_{t-1},Y_t) \\
    &= \frac{P(Y_{t-1},Z_t = 1, Y_t)}{P(Y_{t-1},Y_t)} = \frac{P(Z_t = 1, Y_t | Y_{t-1})}{P(Y_t|Y_{t-1})} \\
    &= \frac{P(Y_t |Y_{t-1},Z_t = 1)P(Z_t = 1|Y_{t-1})}{\sum_{j=1}^2 P(Y_t |Y_{t-1},Z_t = j)P(Z_t = j|Y_{t-1})} \\
    &= \frac{p_1 P(Y_t |Y_{t-1},Z_t = 1)}{\sum_{j=1}^2 p_j P(Y_t |Y_{t-1},Z_t = j)} \\
    &= \frac{p_1 \exp\left(-\frac{1}{2}\left\|Y_t - (\textnormal{Id}_d \otimes Y_{t-1})^\top \beta\right\|_2^2\right)}{p_1 \exp\left(-\frac{1}{2}\left\|Y_t - (\textnormal{Id}_d \otimes Y_{t-1})^\top \beta\right\|_2^2\right) + p_2 \exp\left(-\frac{1}{2}\left\|Y_t + (\textnormal{Id}_d \otimes Y_{t-1})^\top \beta\right\|_2^2\right)},
\end{align*}
and $w_{2,\beta} = 1 - w_{1,\beta}$. As these filtered probabilities depend only on $Y_{t-1}$ and $Y_t$, we will write them as $w_{j,\beta}(Y_{t-1}^t)$. By some algebra, we have the gradient of $w_{1,\beta}(Y_{t-1}^t)$ with respect to $\beta$:
\begin{align*}
    \frac{\partial w_{1,\beta}(Y_{t-1}^t)}{\partial \beta} &= 2w_{1,\beta}(Y_{t-1}^t)\left\{1 - w_{1,\beta}(Y_{t-1}^t)\right\}\left(\textnormal{Id}_d \otimes Y_{t-1}\right) Y_t \\
    &=  2w_{1,\beta}(Y_{t-1}^t)\left\{1 - w_{1,\beta}(Y_{t-1}^t)\right\}\left(Y_t \otimes Y_{t-1} \right),
\end{align*}
where the second equality follows from the mixed-product property of Kronecker product. 

As both $p$ and $\sigma^2$ are assumed to be known, we only need to optimize the following objective function in the population EM algorithm
\begin{align*}
    Q(\tilde\beta | \beta) &= \frac{1}{T}\sum_{t=1}^T E_0 \left[ -\frac{1}{2}w_{1,\beta}(Y_0^T)\|Y_t - (\textnormal{Id}_d \otimes Y_{t-1})^\top \tilde\beta \|_2^2 -\frac{1}{2}\left\{1 - w_{1,\beta}(Y_0^T)\right\}\|Y_t + (\textnormal{Id}_d \otimes Y_{t-1})^\top \tilde\beta \|_2^2 \right] \\
    &= \frac{1}{T}\sum_{t=1}^T E_0 \left[ -\frac{1}{2}w_{1,\beta}(Y_{t-1}^t)\|Y_t - (\textnormal{Id}_d \otimes Y_{t-1})^\top \tilde\beta \|_2^2 -\frac{1}{2}\left\{1 - w_{1,\beta}(Y_{t-1}^t)\right\}\|Y_t + (\textnormal{Id}_d \otimes Y_{t-1})^\top \tilde\beta \|_2^2 \right] \\
    &= E_0 \left[ -\frac{1}{2}w_{1,\beta}(Y_{t-1}^t)\|Y_t - (\textnormal{Id}_d \otimes Y_{t-1})^\top \tilde\beta \|_2^2 -\frac{1}{2}\left\{1 - w_{1,\beta}(Y_{t-1}^t)\right\}\|Y_t + (\textnormal{Id}_d \otimes Y_{t-1})^\top \tilde\beta \|_2^2 \right].
\end{align*}
To maximize the objective function $Q(\tilde\beta|\beta)$, we first take its derivative with respect to $\tilde\beta$,
\begin{equation*}
    \frac{\partial Q(\tilde\beta|\beta)}{\partial \tilde\beta} = - E_0\left[(\textnormal{Id}_d \otimes Y_{t-1})(\textnormal{Id}_d \otimes Y_{t-1})^\top \tilde\beta + \left\{1-2w_{1,\beta}(Y_{t-1}^t)\right\}(\textnormal{Id}_d \otimes Y_{t-1})Y_t\right].
\end{equation*}
Setting the derivative to 0, we have that
\begin{equation*}
    M(\beta) = \left\{E_0\left[(\textnormal{Id}_d \otimes Y_{t-1})(\textnormal{Id}_d \otimes Y_{t-1})^\top\right]\right\}^{-1} E_0\left[\left\{2w_{1,\beta}(Y_{t-1}^t) -1 \right\}(\textnormal{Id}_d \otimes Y_{t-1})Y_t\right].
\end{equation*}
Consequently,
\begin{align*}
    &\frac{\partial M(\beta)}{\partial \beta} \\
    &= \left\{E_0\left[(\textnormal{Id}_d \otimes Y_{t-1})(\textnormal{Id}_d \otimes Y_{t-1})^\top\right]\right\}^{-1} E_0\left[(\textnormal{Id}_d \otimes Y_{t-1})Y_t 2\frac{\partial w_{1,\beta}(Y_{t-1}^t)}{\partial \beta}^\top\right] \\
    &= \left\{E_0\left[(\textnormal{Id}_d \otimes Y_{t-1})(\textnormal{Id}_d \otimes Y_{t-1})^\top\right]\right\}^{-1} E_0\left[4w_{1,\beta}(Y_{t-1}^t)\left\{1 - w_{1,\beta}(Y_{t-1}^t)\right\}(\textnormal{Id}_d \otimes Y_{t-1})Y_t \left(Y_t \otimes Y_{t-1} \right)^\top \right] \\
    &= \left\{E_0\left[(\textnormal{Id}_d \otimes Y_{t-1})(\textnormal{Id}_d \otimes Y_{t-1})^\top\right]\right\}^{-1} E_0\left[4w_{1,\beta}(Y_{t-1}^t)\left\{1 - w_{1,\beta}(Y_{t-1}^t)\right\}\left(Y_t \otimes Y_{t-1} \right) \left(Y_t \otimes Y_{t-1} \right)^\top \right] \\
    &= \left\{E_0\left[\textnormal{Id}_d \otimes (Y_{t-1} Y_{t-1}^\top)\right]\right\}^{-1} E_0\left[4w_{1,\beta}(Y_{t-1}^t)\left\{1 - w_{1,\beta}(Y_{t-1}^t)\right\}\left(Y_t \otimes Y_{t-1} \right) \left(Y_t \otimes Y_{t-1} \right)^\top \right] \\
    &= \left\{ \textnormal{Id}_d \otimes E_0\left[Y_{t-1} Y_{t-1}^\top\right]\right\}^{-1} E_0\left[4w_{1,\beta}(Y_{t-1}^t)\left\{1 - w_{1,\beta}(Y_{t-1}^t)\right\}\left(Y_t \otimes Y_{t-1} \right) \left(Y_t \otimes Y_{t-1} \right)^\top \right]
\end{align*}
where the third and fourth equality is again due to the mixed-product property of Kronecker product. 

We now derive an upper bound on $\|\partial M(\beta)/\partial \beta\|_2$, and later we will examine this upper bound empirically. Note that the matrix $\textnormal{Id}_d \otimes E_0[Y_{t-1} Y_{t-1}^\top]$ is positive definite, and so is its inverse. Thus, we have that
\begin{align*}
    &\left\|\frac{\partial M(\beta)}{\partial \beta}\right\|_2 \\
    &\leq \left\|\left\{ \textnormal{Id}_d \otimes E_0\left[Y_{t-1} Y_{t-1}^\top\right]\right\}^{-1}\right\|_2 \left\|E_0\left[4w_{1,\beta}(Y_{t-1}^t)\left\{1 - w_{1,\beta}(Y_{t-1}^t)\right\}\left(Y_t \otimes Y_{t-1} \right) \left(Y_t \otimes Y_{t-1} \right)^\top \right]\right\|_2 \\
    &= \lambda_{\max}\left(\left\{ \textnormal{Id}_d \otimes E_0\left[Y_{t-1} Y_{t-1}^\top\right]\right\}^{-1}\right) \left\|E_0\left[4w_{1,\beta}(Y_{t-1}^t)\left\{1 - w_{1,\beta}(Y_{t-1}^t)\right\}\left(Y_t \otimes Y_{t-1} \right) \left(Y_t \otimes Y_{t-1} \right)^\top \right]\right\|_2 \\
    &= \left\{\lambda_{\min}\left( \textnormal{Id}_d \otimes E_0\left[Y_{t-1} Y_{t-1}^\top\right]\right) \right\}^{-1} \left\|E_0\left[4w_{1,\beta}(Y_{t-1}^t)\left\{1 - w_{1,\beta}(Y_{t-1}^t)\right\}\left(Y_t \otimes Y_{t-1} \right) \left(Y_t \otimes Y_{t-1} \right)^\top \right]\right\|_2 \\
    &= \left\{\lambda_{\min}\left( E_0\left[Y_{t-1} Y_{t-1}^\top\right]\right) \right\}^{-1} \left\|E_0\left[4w_{1,\beta}(Y_{t-1}^t)\left\{1 - w_{1,\beta}(Y_{t-1}^t)\right\}\left(Y_t \otimes Y_{t-1} \right) \left(Y_t \otimes Y_{t-1} \right)^\top \right]\right\|_2.
\end{align*}
Define an inverse signal-to-noise ratio as
\begin{equation*}
    \textnormal{ISNR}_\beta = \left\{\lambda_{\min}\left( E_0\left[Y_{t-1} Y_{t-1}^\top\right]\right) \right\}^{-1} \left\|E_0\left[4w_{1,\beta}(Y_{t-1}^t)\left\{1 - w_{1,\beta}(Y_{t-1}^t)\right\}\left(Y_t \otimes Y_{t-1} \right) \left(Y_t \otimes Y_{t-1} \right)^\top \right]\right\|_2,
\end{equation*}
which is a continuous function of $\beta$. 

We empirically examine the magnitude of ISNR$_{\beta^*}$ in two sets of experiments as we increase the magnitude or dimension of $\beta^*$. In the first set of experiments, we consider a low-dimensional setting with $d=3$, and increase the magnitude of $\beta^*$. In particular, we define the matrix $A$ as
\[
A =
\begin{bmatrix}
0.5 & 0 & 0 \\
0.1 & 0.1 & 0.3\\
0 & 0.2 & 0.3
\end{bmatrix},
\]
and let $B^* = \mu A$, for a scaling factor $\mu$ that we vary from $0.3$ to $1.5$. The choice of the range of $\mu$ is such that $\|B^*\|_2$ does not exceed 1 per Assumption~\ref{cond: operatornorm}. Let $\beta^* = \vectorize{B^*}$, and we examine ISNR$_{\beta^*}$ as we increase $\mu$. The relevant expectations in the definition of ISNR are approximated with the corresponding sample average, using a sample of size 100,000 taken after a 50,000-step burn-in period. We set $p_1 = p_2 = 0.5$. The results are presented in Figure~\ref{ISNRscale}. We observe that as the magnitude of the regression coefficients increases, the inverse signal-to-noise ratio decreases. Therefore, the ISNR is indeed a reasonable measure of signal strength, and the signal becomes strong as the magnitude of the regression coefficients becomes larger. 
\begin{figure}
    \centering
    \includegraphics[width=0.8\textwidth]{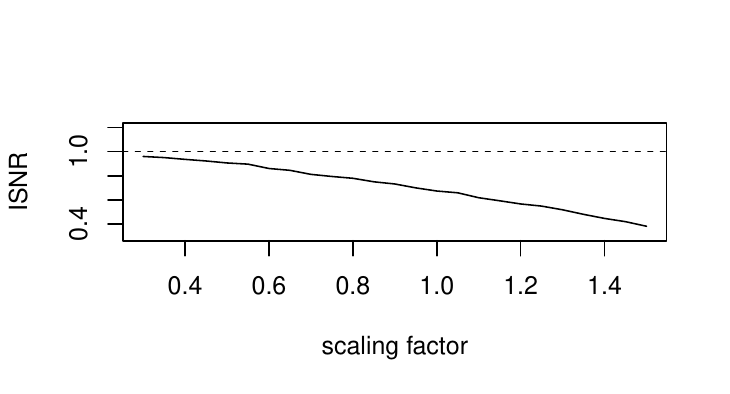}
    \caption{Inverse signal-to-noise ratio (ISNR) in a low-dimensional setting ($d=3$), when the scaling factor increases from 0.3 to 1.5. The magnitude of the regression coefficients is proportional to the scaling factor. Dashed line corresponds to ISNR being 1.}
    \label{ISNRscale}
\end{figure}

In the second set of experiments, we fix the scaling factor $\mu$ at 1, but increase the dimension $d$ from 3 to 30. We take $B^* = \textnormal{Id}_{d/3} \otimes A$ and $\beta^* = \vectorize{B^*}$. Figure~\ref{ISNRdimension} plots the ISNR against the dimension. We observe that when individual diagonal block in $B^*$ remains the same as the dimension $d$ increases, the signal becomes stronger and the ISNR becomes smaller. Moreover, this suggests that as dimension increases, we can allow the magnitude of the difference of a individual regression coefficient between regimes to shrink and still get a non-decreasing signal-to-noise ratio.
\begin{figure}
    \centering
    \includegraphics[width=0.8\textwidth]{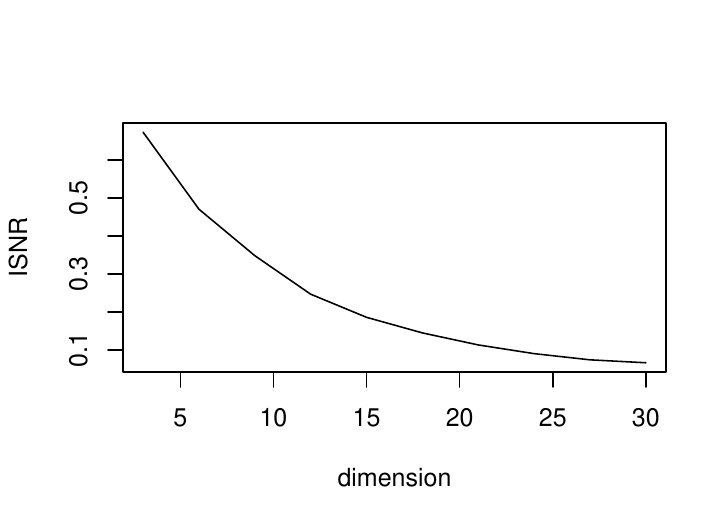}
    \caption{Inverse signal-to-noise ratio (ISNR) as dimension increases from 3 to 30, with the scaling factor fixed at 1.}
    \label{ISNRdimension}
\end{figure}

We also study the expectation of $\|\frac{\partial w_{1,\beta}}{\partial \beta}\frac{\partial w_{1,\beta}}{\partial \beta}^\top \|_2^2$. Recall that $\frac{\partial w_{1,\beta}(Y_{t-1}^t)}{\partial \beta} =  2w_{1,\beta}(Y_{t-1}^t)\{1 - w_{1,\beta}(Y_{t-1}^t)\}(Y_t \otimes Y_{t-1})$, and therfore
\begin{equation*}
    \frac{\partial w_{1,\beta}(Y_{t-1}^t)}{\partial \beta}\frac{\partial w_{1,\beta}(Y_{t-1}^t)}{\partial \beta}^\top =  \left[2w_{1,\beta}(Y_{t-1}^t)\left\{1 - w_{1,\beta}(Y_{t-1}^t)\right\}\right]^2\left(Y_t \otimes Y_{t-1} \right)\left(Y_t \otimes Y_{t-1} \right)^\top.
\end{equation*}
Hence, 
\begin{align*}
    \left\|\frac{\partial w_{1,\beta}(Y_{t-1}^t)}{\partial \beta}\frac{\partial w_{1,\beta}(Y_{t-1}^t)}{\partial \beta}^\top \right\|_2^2 
    &=  \left[2w_{1,\beta}(Y_{t-1}^t)\left\{1 - w_{1,\beta}(Y_{t-1}^t)\right\}\right]^4 \left\|\left(Y_t \otimes Y_{t-1} \right)\left(Y_t \otimes Y_{t-1} \right)^\top\right\|_2^2 \\
    &= \left[2w_{1,\beta}(Y_{t-1}^t)\left\{1 - w_{1,\beta}(Y_{t-1}^t)\right\}\right]^4 \left\|\left(Y_t \otimes Y_{t-1} \right)\right\|_2^4 \\
    &=  \left[2w_{1,\beta}(Y_{t-1}^t)\left\{1 - w_{1,\beta}(Y_{t-1}^t)\right\}\right]^4 \left\|Y_t\right\|_2^4 \left\|Y_{t-1}\right\|_2^4.
\end{align*}
We study the expectation of the quantity above when $\beta = \beta^*$ as we increase the magnitude of the regression coefficients or the dimension. The expectation is approximated using a sample average in the same way as described earlier. Figure~\ref{normscale} plots the expected squared norm as we increase the scaling factor from 0.3 to 1.5 while keeping the dimension fixed at $d=3$. Figure~\ref{normdimension} plots the expected squared norm as we increase the dimension $d$ from 3 to 90 while keeping the scaling factor at 1. The observed trends suggest that Assumption~\ref{randomentropy} is plausible using arguments based on Markov inequality with the expectation being bounded. 

\begin{figure}
    \centering
    \includegraphics[width=0.8\textwidth]{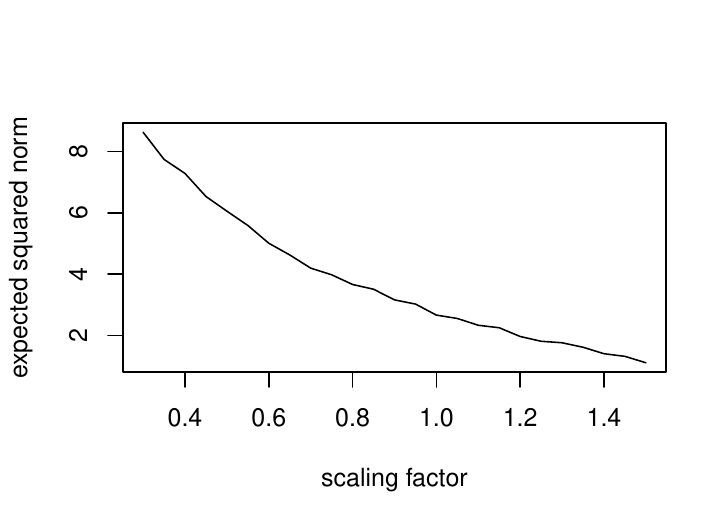}
    \caption{Expected squared norm of $(\partial w_{1,\beta}/\partial \beta)(\partial w_{1,\beta}/\partial \beta)^\top$ in a low-dimensional setting ($d=3$), when the scaling factor increases from 0.3 to 1.5. The magnitude of the regression coefficients is proportional to the scaling factor.}
    \label{normscale}
\end{figure}

\begin{figure}
    \centering
    \includegraphics[width=0.8\textwidth]{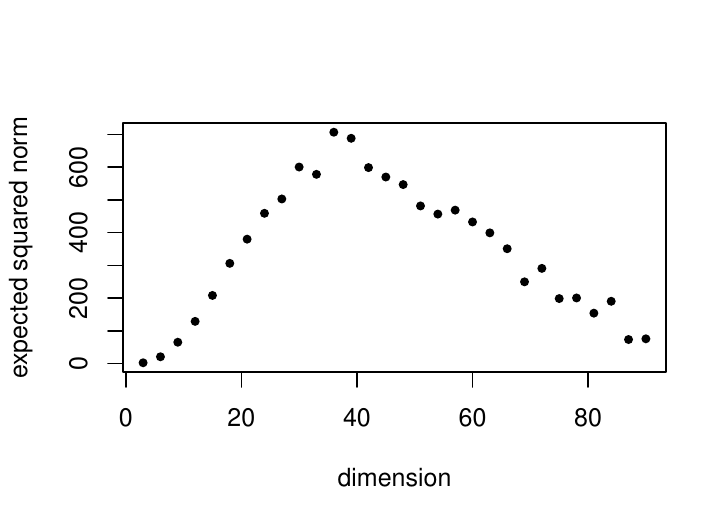}
    \caption{Expected squared norm of $(\partial w_{1,\beta}/\partial \beta)(\partial w_{1,\beta}/\partial \beta)^\top$ as dimension increases from 3 to 90, with the scaling factor fixed at 1.}
    \label{normdimension}
\end{figure}

\section{An expectation-maximization-truncation algorithm}\label{app:EMTalgorithm}
To control the number of false positives in the estimate of the regression coefficients $\hat\beta$, we introduce an additional thresholding step. Specifically, for a given function $\xi$ of sample size $T$ and dimension $d$, we let $\xi(T,d)$ be the threshold level. Define the thresholded estimate $\hat\beta_{thres}$ such that $\hat\beta_{thres}^k = \hat\beta^k I\{|\hat\beta^k| \geq \xi(T,d)\}$, where $\hat\beta_{thres}^k$ and $\hat\beta^k$ denote the $k$-th element of the vector $\hat\beta_{thres}$ and $\hat\beta$, respectively, for $k \in \{1,\ldots, Kd^2\}$. This thresholding step allows us to control the (non-)sparsity level of the regression coefficient estimates uniformly throughout the EM iterations, and can potentially facilitate the theoretical analysis of the algorithm as discussed in Section~\ref{sec:theoretical}.

We outline such an expectation-maximization-truncation (EMT) algorithm in Algorithm~\ref{modifiedEMT}. The penalty parameter $\lambda$ and the threshold level $\xi$ will change over iterations, and we use $\lambda^{(q)}$ and $\xi^{(q)}$ to denote their values in the $q$-th iteration. 
\begin{algorithm}
\caption{An EMT algorithm for high-dimensional Markov-switching VAR model}\label{modifiedEMT}
\hspace*{\algorithmicindent} \textbf{Input:} Observations $\{Y_0,Y_1,\ldots,Y_T\}$, number of regimes $K$;\\
\hspace*{\algorithmicindent} 
\textbf{Output:} Parameter estimate $\hat\theta$ 
\begin{algorithmic}[]
\State Initialize the parameter $\theta^{(0)} = ((\beta^{(0)})^\top, (p^{(0)})^\top,(\sigma^{(0)})^\top)^\top$
\State $q \leftarrow 1$
\While{Convergence condition not met}
\State (a) choose tuning parameter $\lambda^{(q)}$ and threshold level $\xi^{(q)}$
\State (b) optimize the objective \eqref{sampleQ}: $\hat\theta = (\hat\beta^\top,\hat p^\top, \hat\sigma^\top)^\top = \argmax_{\tilde\theta} Q_{n,\lambda^{(q)}} (\tilde\theta|\theta^{(q-1)})$ 
\State (c) update $p$ and $\sigma$: $p^{(q)} \leftarrow \hat p$, $\sigma^{(q)} \leftarrow \hat\sigma$
\State (d) update $\beta$: $\beta^{(q)} \leftarrow \hat\beta_{thres}$ with $\hat\beta_{thres}^k = \hat\beta^k I\{|\hat\beta^k| \geq \xi^{(q)}\}$
\State (e) update $\theta$: $\theta^{(q)} \leftarrow ((\beta^{(q)})^\top,(p^{(q)})^\top,(\sigma^{(q)})^\top)^\top$
\State $q \leftarrow q+1$
\EndWhile
\State $\hat\theta \leftarrow \theta^{(q-1)}$

\end{algorithmic}
\end{algorithm}

\end{document}